\documentclass{article}

\usepackage{multirow}

\usepackage[utf8]{inputenc}
\usepackage[T1]{fontenc}
\usepackage[a4paper, margin=0.7in]{geometry}
\usepackage{tikz}
\usepackage{soul} % For highlighting
\usepackage{enumitem}
\usepackage{subfigure}
\usepackage{array}
\usepackage{float}
\usepackage{booktabs}
\usepackage{geometry}
\usepackage{color}
\usepackage{amsmath, amssymb, amsthm, mathtools, bm, dsfont}
\usepackage{threeparttable}
\usepackage{lmodern}
\usepackage[tracking=true,kerning=true,final]{microtype}
\usetikzlibrary{positioning,calc,arrows.meta}
\usepackage{hyperref}
\hypersetup{colorlinks=true, linkcolor=blue,  anchorcolor=blue, citecolor=blue, filecolor=blue, menucolor=blue, urlcolor=blue}
\hypersetup{pdfauthor={},pdftitle={}}

\usepackage{caption}   
\usepackage{subcaption}
\usepackage{color}
\usepackage{graphicx}
\graphicspath{{./}{./figs/}}

\usepackage{amsfonts}
\usepackage{bm}
\usepackage{mathtools}
\theoremstyle{plain}

\newtheorem{lemma}{Lemma}
\newtheorem{proposition}{Proposition}

\theoremstyle{definition}

\title{Risk-Sensitive Learning in Population Games under Extreme Events: Bifurcations and Chaotic Dynamics}

\author{Konstantinos Metaxas\thanks{Email: \href{mailto:metaxas@mit.edu}{metaxas@mit.edu}}, Themistoklis P. Sapsis\thanks{Email: \href{mailto:sapsis@mit.edu}{sapsis@mit.edu}}\\
Department of Mechanical Engineering,
\\ Massachusetts Institute of Technology, \\
77 Massachusetts Ave., Cambridge, MA 02139\\
}
\date{}

\usepackage{amsmath}
\usepackage{amssymb}

\usepackage[normalem]{ulem}

\begin{document}

 \maketitle\ 
\begin{abstract}
Inspired by nonequilibrium phenomena in game dynamics and behavioral evidence on the impact of extreme events on decision making, we investigate the nonlinear dynamics of a discrete-time multiagent learning rule in population congestion games under extreme events affecting one of the actions. The population state, following a risk-sensitive variant of the Multiplicative Weights Update (MWU), is coupled with a belief variable capturing the agents' perceived risk and updated through an adaptive expectation rule. We perform a two-parameter bifurcation analysis with respect to the agents' controlled parameters, identifying regions of qualitatively distinct behavior. Equilibria are studied first from both game-theoretic and dynamical perspectives. The resulting two-dimensional system exhibits complex behavior, including multistability among fixed points, invariant curves, periodic and chaotic attractors. Despite this complexity, the attractors can be grouped into distinct families, while the Cesàro averages of the trajectories are shown to converge to the stationary equilibrium. The incorporation of risk associated with the extreme event leads to new dynamical phenomena: attracting invariant curves arise and give rise to phase-locking Arnold tongues, within which the dynamics is qualitatively similar. In this setting, codimension-two resonances are identified as organizing centers, both within individual tongues and along the bifurcation curves associated with the fixed-point family. Chaotic attractors emerge and are destroyed through Feigenbaum cascades and forward or reverse boundary crises, with interior and merging crises also observed, along with transient chaos and narrow periodic windows. For each qualitatively distinct region, representative phase portraits and the associated basins of attraction are examined.

  
\end{abstract}

\section{Introduction}\label{intro}
Congestion games constitute one of the most fundamental and widely studied families of games in game theory \cite{rosenthal,monderer}. They are typically used as models in settings where agents jointly utilize a set of shared resources, with the cost of each resource determined by its level of congestion. Mirroring biological adaptation, multiagent reinforcement learning suggests that agents utilize the outcomes of their past actions to adapt to the game and select their subsequent actions \cite{camerer,sandholm}. In this context, the agents' behavior in congestion games can be viewed as aiming to minimize a common potential function \cite{monderer,bielawski2}. In the limit of a very large population of interacting agents, the learning process is studied at the population level and the game is referred to as population or nonatomic. Several results establish convergence to the Nash equilibrium in both continuous-time and discrete-time settings, provided, in the latter case, that the learning rate is sufficiently small \cite{monderer,sandholm,palaiopanos,krichene}.

The behavior of the system as learning rates deviate from this threshold, or more broadly the question of the potential complexity of agents' behavior in games, has attracted significant attention. The seminal work of \cite{sato} has shown that even in simple zero-sum two-player games, such as rock-paper-scissors, the continuous-time replicator dynamics of evolutionary game theory \cite{sandholm,hofbauer,allen,borgers,fudenberg,fudenberg2} can exhibit Hamiltonian chaos, while Hamiltonian and Poisson structures of zero-sum replicator dynamics have also been further investigated in \cite{griffin1}. Several subsequent works have further investigated chaotic dynamics in continuous-time learning \cite{sato2,sato3}, while chaos has also been established in fields at the intersection of evolution and game theory \cite{nowak,chattopadhyay,you,krieger,pati}. In fact, the dynamics of any complex system can be reproduced by agents learning a suitable matrix game under replicator dynamics \cite{andrade}. Important works on stochastic variants of replicator dynamics in zero-sum games \cite{engel}, and more generally on stochastic effects in evolutionary game dynamics \cite{li2}, have also been carried out. Related nonlinear evolutionary-game models have investigated the effects of delayed payoffs, memory, mutation feedback, density dependence, and finite-population or spatial effects on the stability and complexity of the dynamics \cite{yuan,bischi_memory,wang1,wang2,griffin2}. Applications of evolutionary game dynamics include social dilemmas, cooperation, public-goods interactions, and resource-management problems \cite{mintz,lin,feng,ding,li,liu_pool,yang_cooperation,yoshioka}.

In discrete time, where players make decisions in a round-by-round manner, the Multiplicative Weights Update (MWU) \cite{freund} and its generalization Follow-the-Regularized-Leader can be viewed as discrete-time analogues of the corresponding continuous-time Hamiltonian dynamics which are recurrent \cite{bailey4,boone}. However, when agents learn using these rules in zero-sum games, the resulting dynamics can exhibit Lyapunov chaos \cite{mertikopoulos1,bailey,cheung1,cheung2}. Complex behavior has also been shown to arise in nonatomic two-strategy congestion games with agents employing MWU \cite{palaiopanos,chotibut1,chotibut2,bielawski}, and persists in heterogeneous settings \cite{bielawski2}, where agents have different learning rates \cite{hadhikanloo} or when learning rates are adaptively updated \cite{vlatakis1}. Related studies of nonlinear evolutionary-game dynamics have also emphasized the role of adaptive update rules, network structure, attractors, and basins of attraction \cite{hakhamanesh,sun}. In more complex games with two or more players, chaos has been observed when agents utilize the Experience-Weighted Attraction rule of behavioral game theory, although memory loss, acting as a perturbation to standard MWU dynamics, can lead to more predictable behavior \cite{galla,sanders,bielawski3,pangallo2}. Chaotic behavior has also been observed in auctions \cite{cheung3} and in dynamics induced by fictitious play learning rules \cite{sparrow,vanStrien}, gradient descent-ascent algorithms \cite{bailey2}, as well as in discretized replicator dynamics \cite{mukhopadhyay,pandit,vilone}. Beyond learning and evolutionary dynamics, nonlinear strategic adjustment models in economic games, including duopoly, triopoly, oligopoly and Cournot-type games, have been known to exhibit complex bifurcation structures and chaos \cite{puu,kopel,tramontana,matouk,andaluz,wei}.

According to behavioral decision theory, low-probability, high-impact outcomes may affect individuals' responses by shaping their perceived risk, memory, and subsequent decision making \cite{kahnemann1,kahnemann2}. Thus, an extreme event not only has a material consequence but also adaptively affects an individual's perception and memory. From the perspective of decisions from description and experience, the influence of extreme events on decision making is closely tied to memory and learning from realized outcomes \cite{plonsky,hertwig}. In the context of learning in games, extreme events become particularly relevant in real-world settings where agents adapt their strategies under uncertainty, and the learning process becomes coupled with their internal perception of risk \cite{camerer,camerer2}. In fact, consistent with experience-based decision making, the behavioral impact of extreme events extends to the agents' adaptive learning process, making them an important endogenous factor \cite{plonsky2}. Despite extensive research on complex dynamics induced by learning in games and behavioral game theory, the effect of rare, high-impact events on population dynamics has received limited direct attention. This raises the question of what new bifurcations, attractors and dynamical phenomena may arise from the behavioral impact of such events.

From an algorithmic perspective, risk-sensitive learning has been studied extensively in reinforcement learning. Existing approaches typically incorporate risk either through a modification of the optimization objective or through an explicit risk constraint \cite{mihatsch_risk,tamar_cvar,prashanth_risk}. Related approaches have also considered multi-agent and game-theoretic learning settings, including Markov games \cite{tian_riskgame,qiu_rmix}. However, these works are primarily oriented toward designing or optimizing algorithms to achieve a prescribed objective, rather than studying the new dynamical phenomena that may arise when agents' perceptions of rare-event risk evolve endogenously together with their strategies, and how these phenomena affect the predictability of their behavior.

Motivated by the above discussion, the present work investigates the nonlinear dynamics and bifurcations induced by learning in the presence of relatively extreme events in the simplest class of population games, namely two-action congestion games. It is assumed that an extreme event affects the first resource, which agents perceive through a shared signal. To illustrate the complexity of the learning dynamics, it suffices to consider a simple linear dependence of this signal on the population state. Consistent with observations from human behavior, the population state is augmented by a common belief variable capturing the agents' perception or memory of the risk associated with the first action. We stress that the belief variable is common, representing a homogeneous population receiving the same risk information; a heterogeneous setting would further increase the complexity of the system and represents a promising direction for future work. This variable is updated through a standard discrete-time adaptive expectation rule \cite{hommes}. As a result, it can be readily interpreted as an exponentially weighted memory of the state-dependent risk signal. The belief-adjustment parameter controls the responsiveness of risk perception, with small values corresponding to persistent memory and large values to rapid adaptation. Each agent follows a risk-sensitive variant of the discrete-time reinforcement learning MWU algorithm, in which the belief about the risk enters the cost, making the system two dimensional. Importantly, compared with the algorithmic perspective, the increase in dimensionality, resulting from the agents' perception about risk, is a key distinction, as a mere modification of the objective would leave the dimension of the system unchanged. The question addressed here is therefore how this endogenous risk-learning mechanism reorganizes the global dynamics and affects the agents' behavior.

We first study the equilibria of the system from both dynamical and game-theoretic perspectives. Similar to traditional MWU dynamics, the system admits two trivial equilibria corresponding to the boundary cases of the population state. The equilibrium associated with the full-population state is always a saddle, while the one corresponding to the zero-population state is stable only in parameter regimes associated with non-extreme events. In this case, the dynamics is trivial, since the equilibrium is globally attracting and also game-theoretically stationary. When this equilibrium is a saddle, a third interior equilibrium exists, which is always stationary. Using numerical continuation techniques \cite{matcont,datseris}, we compute bifurcation curves in a two-dimensional parameter space defined by the key parameters controlled by the agents, namely the learning rate and the belief adjustment rate, thereby identifying regions of qualitatively distinct behavior.
%In addition to detailed bifurcation diagrams, we analyze representative cases of multistability through basins of attraction and phase portraits. 
The system under study exhibits complex behavior, including chaotic attractors emerging through Feigenbaum cascades, crisis bifurcations \cite{crisis,crisis2,crisis3,lichtenberg1} leading to the destruction, creation, or expansion of strange attractors, stable invariant curves, phase-locked periodic orbits, codimension-two resonances, and multistability among these attractors. In addition to detailed bifurcation diagrams, we analyze representative cases of multistability through basins of attraction and phase portraits. We also calculate the associated basin entropy ($S_b$) and boundary basin entropy ($S_{bb}$), which quantify the unpredictability and mixing properties of the basins \cite{entropy1,entropy2,entropy3}. In this way, the fractal structure of the basin boundaries is inferred. A case of fully mixed ($S_b=S_{bb}$) non-Wada basins is also identified. Transient chaos and narrow periodic windows in parameter space are observed in the presence of chaotic saddles. Despite the complexity of the system, the attractors form well-defined families and exhibit qualitatively similar dynamical phenomena. In particular, the codimension-two structure of the periodic families is consistent across different families, while the basins of attraction are typically highly intermingled. Interestingly, the dynamical phenomena analyzed in this work can be viewed in direct analogy with the concept of chaotic itinerancy discussed in the cognitive literature \cite{kaneko,nicolis,tsuda,basios}. Finally, the Cesàro time averages of the orbits are shown to converge to the stationary equilibrium, even in cases exhibiting complex and chaotic dynamics. 

Note that the two-dimensional nature of the system allows for several phenomena that are absent from the MWU variants studied so far. Indicative examples include Neimark--Sacker bifurcations, Arnold tongues, and codimension-two resonances, which have not been identified in traditional MWU dynamics for population games, even in the heterogeneous setting \cite{bielawski2}. Moreover, the stationary interior equilibrium may be stable without being the global attractor, coexisting with both chaotic and periodic attractors. Finally, for large values of the learning rate, the system exhibits predictable periodic behavior, in stark contrast to the one-dimensional MWU case.

From an application perspective, the analyzed framework admits a direct interpretation. If we consider, for example, a transportation network, then the two actions may represent two alternative routes, while the usual congestion cost may represent travel time. The extreme event associated with the risky route may correspond to a disruption causing large additional delays, such as an accident or severe weather. The state-dependent risk signal models the fact that the perceived likelihood of such an event may increase with the utilization of this resource, while the belief variable represents the population's current estimate of this risk based on both past and current information. Under an analogous interpretation, the analyzed framework can be directly applied to resource-allocation problems in which agents choose between shared resources, such as communication channels, computational servers, or energy resources. Thus, although the model is deliberately minimal, it isolates the feedback between resource utilization, adaptive risk perception, and repeated decision making that may be present in these applications. The qualitatively distinct dynamical regimes resulting from this coupling correspond to very different patterns of agents' behavior: convergence to the Nash equilibrium represents the desired behavior, while other regimes involve either predictable wandering around the equilibrium without convergence, as in periodic orbits or invariant curves, or unpredictable behavior, as in the chaotic regime.

This paper is organized as follows. In Section~\ref{section1}, we first briefly discuss traditional learning in congestion games under MWU, and then introduce and study the fundamental properties of the system induced by the risk-sensitive learning rule. Technical details for this Section are provided in Appendix~\ref{app}. In Section~\ref{section2}, a detailed bifurcation analysis of the system is conducted, qualitatively distinct attractor families are identified, and the complex and chaotic dynamics of the system are discussed in depth. Finally, Section~\ref{section3} concludes the article.

\section{Model and fundamental properties}
\label{section1}
In this Section, we introduce and discuss the fundamental properties of the risk-sensitive variant of the Multiplicative Weights Update algorithm (MWU), whose nonlinear dynamics we study in detail in Section~\ref{section2}. We first briefly discuss the setting of learning congestion games under the standard multi-agent reinforcement learning algorithm MWU \cite{palaiopanos,arora}.

\subsection{Learning congestion games with MWU}\label{subsec2.1}
We consider a two-strategy non-atomic congestion game, i.e., a game with a continuum of agents \cite{bailey,cheung2}. Let \(x_n\in[0,1]\) denote the fraction of agents choosing action \(1\) at time \(n\), so that the corresponding flows on actions \(1\) and \(2\) are \(Nx_n\) and \(N(1-x_n)\), respectively, where \(N>0\) is the total population mass. The realized stage costs are given by
\begin{equation}
c_1=N\gamma x_n, \qquad c_2=Nb(1-x_n),
\label{eq00}
\end{equation}
where the parameters \(\gamma,b>0\) are the coefficients of proportionality and satisfy \(\gamma+b=1\) \cite{bielawski3}. When the agents learn according to MWU, their choices evolve as
\begin{align} x_{n+1}&=\frac{x_n \exp{\left(-\lambda c_1(x_n) \right)}}{x_n\exp{\left(-\lambda c_1(x_n) \right)}+(1-x_n) \exp{\left(-\lambda c_2(1-x_n) \right)}} \nonumber\\ &=\frac{x_n}{x_n+(1-x_n) \exp{\left(a(x_n-b) \right)}}, \label{trad_mwu} 
\end{align}
where \(\lambda=-\log(1-\epsilon)\) describes the intensity of choice, \(a=N\lambda\) denotes the population intensity of choice, and \(\epsilon\in(0,1)\) is the common learning rate.

The dynamical system \eqref{trad_mwu} has the boundary fixed points \(0\) and \(1\), as well as the interior fixed point \(b\). The boundary fixed points are always repelling, while the interior fixed point is globally attracting if \(a\leq\frac{2}{b(1-b)}\) \cite{chotibut1}. If this condition does not hold, then complex behavior may arise, including periodic orbits and chaos for sufficiently large values of $a$ \cite{chotibut1}. However, the Cesàro mean of any orbit always converges to \(b\) \cite{chotibut1}-- cf. the analogous result described in Proposition \ref{prop1} for the risk-sensitive case we will be considering.

From a game-theoretic perspective, a \emph{Nash equilibrium} is a mixed strategy determined by a fraction $x^N$, with the first action played with probability $x^N$ and the second with probability $1-x^N$, such that no agent can improve their payoff by unilaterally deviating to another strategy. Recalling \eqref{eq00}, one can immediately see that the interior equilibrium $b$ is always a Nash equilibrium.

The intricate dynamics of this simple and fundamental model in multi-agent reinforcement learning have led to the investigation of several variants of the standard MWU learning rule \cite{bielawski2,palaiopanos,chotibut2,bielawski,vlatakis1,bielawski3}, including heterogeneous learning rates \cite{bielawski2,hadhikanloo} and adaptively updated learning rates \cite{vlatakis1}. These works indicate that complex behavior persists under several extensions, while discounting through memory loss may, under certain conditions, lead to the avoidance of chaos \cite{bielawski3}.

Despite the plethora of studied extensions of MWU, as also discussed in Section~\ref{intro}, the complexity of the dynamics when agents incorporate perceived risk and memory, associated with the costs of actions, stemming from a low-probability extreme event has not received much direct attention. This is particularly relevant because extreme events are known to affect the decision-making and adaptation process by coupling it with risk perception \cite{camerer,plonsky,hertwig,camerer2}. Moreover, from a dynamical-systems perspective, the consideration of such a model offers the possibility of new phenomena that are absent from the existing literature on the dynamics of MWU. The purpose of this paper is to offer a first step towards addressing this gap. To this end, in the next subsections of Section~\ref{section1}, we introduce and discuss the risk-sensitive model whose dynamics and bifurcations we study in detail in Section~\ref{section2}.

\subsection{Risk-sensitive MWU variant}\label{risk_sens}
In this subsection, we introduce the risk-sensitive variant considered in this paper. We consider a non-atomic congestion game as in Section~\ref{subsec2.1} and use the same notation. However, we now suppose that an extreme event affects action $1$, i.e., the realized stage costs are given by
\begin{equation}
C_1=N(\gamma x_n+MU_n), \qquad C_2=Nb(1-x_n),
\label{eq0}
\end{equation}
where \(U_n\in\{0,1\}\) is a random variable corresponding to an  extreme (and, thus, low-probability) event affecting action \(1\) and $M>1$ corresponds to its intensity. The parameters $\gamma, \; b>0$ are as in Section \ref{subsec2.1}. We suppose that the agents, instead of knowing explicitly the law of $U_n,$ perceive the occurrence of such an event through a (shared) signal generated by the environment. To illustrate the complexity induced in the learning dynamics by this mechanism, it is sufficient to consider the simple model that represents the probability (risk) for an extreme event
\begin{equation}
p(x_n)=p_0+p_1x_n,
\label{eq_sig}
\end{equation}
where \(p_0,p_1> 0\) (of the order of $10^{-3}$ or less). Here \(p_0\) represents the exogenous baseline risk of action \(1\), while the term \(p_1x_n\) captures the endogenous dependence of the perceived risk on the use of that action. Thus, $p_0+p_1x_n$ is interpreted as the perceived risk signal for the extreme event associated with action \(1\) at time $n.$ 

To model risk perception and belief, we introduce a homogeneous belief (or memory) variable $q_n\in[0,1]$, representing the agents' common internal assessment of the extreme-event risk at time $n$. We stress that $q_n$ is common across all agents representing a population receiving the same risk information. Learning and remembering the perceived risk of an action and making decisions based on it is a common modeling approach in decision-making from experience motivated by human behavior \cite{hertwig,prelec}. Outside of the connection to extreme events, the $q$ state can be thought of as representing an internal variable that is recursively updated and that, together with the observable state $x$, determines the players' choices, a modeling approach commonly used in learning and behavioral dynamics \cite{camerer}.

The evolution of this belief is assumed to follow the standard adaptive rule
\begin{equation}
q_{n+1}=(1-k)q_n+k(p_0+p_1x_n),
\qquad k\in(0,1),
\label{eq_q}
\end{equation}
where the next iterate belief is a convex combination of the current internal assessment and the perceived risk signal. The parameter \(k\) governs the speed of belief adjustment. Stated otherwise, it controls the responsiveness of risk perception, with small values corresponding to persistent memory and large values to rapid adaptation. Equivalently, \eqref{eq_q} can be seen as the standard discrete-time form of adaptive expectations or exponential smoothing, showing that $q_n$ constitutes an exponentially weighted memory of the past risk signals. Rather than reacting directly to the instantaneous signal $p_0+p_1x_n$, agents evaluate the extreme-event component of action $1$ using $q_n$ as their current perceived event probability and amplify the intensity $M$ by a factor $1/(1-\alpha),\; \alpha \in (0,1)$. This constitutes essentially a simplified version of the Conditional Value-at-Risk measure at confidence level $\alpha$; more details on this matter are given in Appendix \ref{sec2_app}. Accordingly, the perceived cost of action \(1\) is
\begin{equation}
\tilde c_1(x_n,q_n)
=
N\left(\gamma x_n+cq_n\right),
\label{eq:c1tilde}
\end{equation}
where $c=M/(1-\alpha)$ is the effective intensity. The perceived cost of action \(2\) is
\begin{equation}
\tilde c_2(x_n)=Nb(1-x_n).
\label{eq_c1}
\end{equation}
Under the normalization \(\gamma+b=1\), the perceived relative cost becomes
\begin{equation}
\Delta(x_n,q_n)=\tilde c_1(x_n,q_n)-\tilde c_2(x_n)
=
N\left(x_n-b+c q_n\right).
\label{eq_c2}
\end{equation}
Let \(\epsilon\in(0,1)\) denote the common learning rate and set $\lambda=-\log(1-\epsilon), \; a=N\lambda$. Then, following the same computations as in the derivation of \eqref{trad_mwu}, the risk-sensitive Multiplicative Weights dynamics is given by 
\begin{align}
x_{n+1}&=\frac{x_n}{x_n+(1-x_n)\exp\!\left(
a\left(x_n-b+cq_n \right)
\right)} \nonumber\\
q_{n+1}&=(1-k)q_n+k(p_0+p_1x_n).
\label{sys}
\end{align}
Note that the dynamics of the nonlinear system \eqref{sys} depends on the effective parameters $(a,b,k,cp_0,cp_1)$, as can be observed after setting $Q_n=cq_n.$ We refer to the previous dynamics as \emph{risk-sensitive}, due to the coupling of the population state $x$ with the perceived risk modeled by $q$, resulting from the evaluation of risk in the cost of the first resource. In the following sections, we study the complex dynamics induced by the learning rule defining the system \eqref{sys}, which we denote by $(x,q)\mapsto f(x,q)=(f_1(x,q),f_2(x,q)).$

\subsection{Equilibria and fundamental properties}\label{sec2.2}
In this subsection, we determine the system's equilibria (both dynamically and game-theoretically) and discuss the fundamental properties of system \eqref{sys}, including the convergence of the time averages of every orbit to the stationary equilibrium, a property similar to that of the traditional MWU discussed in Section \ref{subsec2.1}.

The dynamical system \eqref{sys} has the trivial fixed points $(0,p_0), \; (1,p_0+p_1)$. The first equilibrium is stable if and only if $b\leq cp_0,$ while the second one is always a saddle, as it is unstable in the $x-$direction. If $b\leq cp_0,$ then the dynamics is trivial, since the fixed point $(0,p_0)$ attracts all orbits, as shown in Lemma \ref{lem1} in the Appendix \ref{sec1_app}. If $b>cp_0$ the system has the interior equilibrium $(x^\star, q^\star)$, where
\begin{equation}
x^\star=\frac{b-cp_0}{1+cp_1}, \; q^\star=\frac{p_0+p_1b}{1+cp_1}, \; c=\frac{M}{1-\alpha}.
\label{fp}
\end{equation}  
In that case, regarding the boundary equilibria, we also have the following property (shown in Lemma \ref{lem_delta}): for every orbit there is a $\delta>0$ such that $x_n \in (\delta,1-\delta)$ for all $n$. Although the more interesting questions concerning the stability of the interior equilibrium and its associated bifurcations are studied in detail in the following sections, we make here the preliminary observation that, from \eqref{eq_c2} and \eqref{sys}, the perceived relative cost $\Delta$ is increasing in the belief coordinate $q$, i.e., $\partial \Delta/\partial q>0$, while the belief update is increasing in the population state, $\partial q_{n+1}/\partial x_n>0$, and $\partial x_{n+1}/\partial \Delta<0$. Therefore, an increase in $x$ induces, with a delay, a response that tends to decrease $x$, giving rise to a lagged negative-feedback mechanism. Since the parameter $k$ determines how rapidly $q_n$ adjusts when $x$ is held fixed, the characteristic relaxation time is of order $1/k$. In the Jacobian matrix
\begin{equation}
J=
\begin{pmatrix}
1-ad & -acd\\
kp_1 & 1-k
\end{pmatrix},
\end{equation}
where we have set $d=x^\star(1-x^\star)$, this structural mechanism is reflected in the opposite signs of the off-diagonal entries, which allow for a complex-conjugate pair of multipliers and, consequently, Neimark--Sacker bifurcations, a possibility that is absent in the traditional MWU model. In fact, the associated Neimark--Sacker curve in the $(a,k)$ plane can be computed from the conditions $\det J=1$ and $-2\leq \mathrm{tr}J\leq 2$, where the boundary equalities correspond generically to codimension-two bifurcations discussed in the next sections, as follows
\begin{equation}
a=\frac{k}{d\left(k(1+cp_1) -1\right)},
\label{NS_eq}
\end{equation}
when $cp_1>1/3$ and for appropriate $k.$ The condition $cp_1>1/3$ and the restriction on $k$ stem from the condition $-2\leq \mathrm{tr}J$. Further details on when the Neimark--Sacker bifurcation and the codimension-two point belong to the admissible region of the bifurcation plane, i.e., $k\in(0,1)$, are discussed in Section \ref{sec3.4}.

In terms of game-theoretic equilibria, 
%for a fixed value of $q$, a Nash equilibrium is a mixed strategy determined by a probability $x^N$, with the first action played with probability $x^N$ and the second with probability $1-x^N$, such that no agent can improve their payoff by unilaterally deviating to another strategy. 
a stationary equilibrium with consistent beliefs is a pair $(x^N,q^N)$ with $x^N$ a Nash equilibrium (cf. Section \ref{subsec2.1}) and $q^N=p_0+p_1x^N$. Stationary equilibria form a subset of the fixed points of the system. It is easy to check that $(1,p_0+p_1)$ is never a stationary equilibrium. If $b<cp_0$, i.e., the $(0,p_0)$ equilibrium is stable and the interior equilibrium does not exist, then $(0,p_0)$ is also a stationary equilibrium. Thus, in this trivial case, where the extreme event is sufficiently strong and relatively likely to occur, all trajectories converge to the Nash equilibrium of the system. In the more interesting case of $b>cp_0$, the trivial equilibrium $(0,p_0)$ is not a stationary equilibrium of the system, while the interior one, $(x^\star,q^\star)$, satisfies both the condition of Nash equilibrium $\Delta(x^\star, q^\star)=0$ and $q^\star=p_0+p_1x^\star$, and is therefore the only stationary equilibrium.

%Let $\tilde{f}$ denote the map defining \eqref{sys}, but without the minimum operator in the equation, i.e.,
%\begin{equation}
%\tilde{f}(x,q)=\left(\frac{x_n}{x_n+(1-x_n)\exp\!\left(
%a\left(x_n-b+Mq/(1-\alpha)\right)
%\right)}, \;(1-k)q+k(p_0+p_1x)  \right).
%\label{ftilde}
%\end{equation}
As noted in the previous Section, system \eqref{sys} is closely related to the one obtained when agents evaluate the perceived risk using the CVaR$_\alpha$ risk measure. The dynamics of the two systems are also closely related. These connections are discussed in Appendix \ref{sec2_app}. 

We now prove the convergence of time-averages of the trajectories of the system \eqref{sys} to the stationary equilibrium, even when the induced dynamics is complex, as is shown in the next sections. This result can be viewed as an analogue of the corresponding result for traditional MWU.

\begin{proposition}\label{prop1}
If $x_0\in (0,1),$ then we have the following convergences 
\begin{equation}
\lim_{n\to \infty} \frac{1}{n+1} \sum_{i=0}^n x_i =x^\star, \; \lim_{n\to \infty} \frac{1}{n+1} \sum_{i=0}^n q_i =q^\star. 
\label{cesaro}
\end{equation}
\end{proposition}
\begin{proof}
The case $b\leq cp_0$ is trivial, as the stationary equilibrium is a global attractor of the system. Suppose that $b>cp_0$. Inductively we have that 
$$
x_{n+1}=\frac{x_0}{x_0+(1-x_0)\exp{ \left(a \sum_{i=0}^n (x_i-b_i) \right)}}, \; b_i=b-cq_i.
$$
From Lemma \ref{lem_delta}, there exists a $\delta\in (0,1)$ such that $\delta<x_n<1-\delta$ for all $n$. Thus,
$$
\frac{x_0}{1-\delta}<x_0+(1-x_0)\exp{ \left(a \sum_{i=0}^n (x_i-b_i) \right)}<\frac{x_0}{\delta}
$$
and so
$$
\delta^2<\frac{\delta^2}{1-x_0}<\frac{x_0\delta}{(1-\delta)(1-x_0)}<\exp{ \left(a \sum_{i=0}^n (x_i-b_i) \right)}<\frac{x_0(1-\delta)}{\delta(1-x_0)}<\frac{1}{\delta}.
$$
As a result,
$$
\lim_{n \to \infty}\frac{1}{n+1} \sum_{i=0}^n (x_i+cq_i)=b.
$$
Now let 
$$
A_n=\frac{1}{n+1} \sum_{i=0}^n x_i, \qquad B_n=\frac{1}{n+1} \sum_{i=0}^n q_i.
$$
From \eqref{eq_q} we have that
$$
\frac{q_{n+1}-q_0}{n+1}=k(p_0+p_1A_n-B_n),
$$
and thus
$$
\lim_{n \to \infty} (1+cp_1)A_n+cp_0-\frac{c(q_{n+1}-q_0)}{k(n+1)} =b.
$$
This implies, since $q_n$ is bounded, that 
$$
A_n=\frac{1}{n+1} \sum_{i=0}^n x_i \xrightarrow{n \to \infty} x^\star.
$$
The convergence of the time average of $q_n$ follows.

\end{proof}
The previous property suggests immediately one further regularity result: if $\mu$ is a probability measure on $[0,1]\times[0,p_0+p_1]$ invariant under $f$ and such that $\mu\left(\{ 0,1\} \times [0,p_0+p_1] \right)=0,$ then Birkhoff's ergodic theorem implies that 
\[
\int (x,q) \;\mathrm{d}\mu= (x^\star, q^\star).
\]

Despite the regularity implied by the convergence of the Ces\`aro sums, the normalized time-averaged social cost along an orbit with initial condition $x_0$,
\begin{equation}
\bar{\mathcal{C}}(x_0)=\lim_{M \to \infty}\frac{1}{M+1}\sum_{n=0}^{M}c(x_n),
\end{equation}
assuming that the limit is well-defined and where $c(x_n)=(1-b)x_n^2+b(1-x_n)^2$ is the normalized social cost at time $n$, generally depends not only on the attractor family to which the orbit belongs, but also on the particular attractor itself. To see this, let $\bar{\mathcal{C}}_M(x_0)$ denote the normalized time-averaged social cost up to time $M$, i.e.,
\begin{equation}
\bar{\mathcal{C}}_M(x_0)=\frac{1}{M+1}\sum_{n=0}^{M}c(x_n).
\label{soc_cost}
\end{equation}
Then,
\begin{equation}
\bar{\mathcal{C}}_M(x_0)=\mathrm{Var}_M(x_0)+c(\bar{x}_M),
\label{soc_cost_2}
\end{equation}
where $\bar{x}_M$ denotes the average and
\begin{equation}
\mathrm{Var}_M(x_0)=\frac{1}{M+1}\sum_{n=0}^M(x_n-\bar x_M)^2
\end{equation}
denotes the variance of the orbit with initial condition $x_0$ up to time $M$. For an initial condition $x_0$ belonging to the basin of attraction of a period-$m$ attractor $\{x_1^{(m)},\dots,x_m^{(m)}\}$, it follows from \eqref{cesaro} that
\begin{equation}
\bar{\mathcal{C}}(x_0)
=
\lim_{M\to\infty}\bar{\mathcal{C}}_M(x_0)
=
\frac{1}{m}\sum_{n=1}^{m}\left(x_n^{(m)}-x^\star\right)^2+c(x^\star).
\label{soc_per_m}
\end{equation}
For an initial condition $x_0$ belonging to the basin of attraction of an invariant curve born through a Neimark--Sacker bifurcation, $\bar{\mathcal{C}}_M(x_0)$ is again convergent as $M\to\infty$. In particular,
\begin{equation}
\bar{\mathcal{C}}(x_0)
=
\lim_{M\to\infty}\bar{\mathcal{C}}_M(x_0)
=
\mathrm{Var}_{\infty}(x_0)+c(x^\star),
\label{soc_per_ch}
\end{equation}
where $\mathrm{Var}_{\infty}(x_0)=\mathrm{Var}_{\mu_A}$ and $\mu_A$ denotes the ergodic physical measure associated with the attractor. Note that \eqref{soc_per_m} can also be written in the form \eqref{soc_per_ch} for appropriate $\mu_A$. We assume that, as is the case of interest in the present paper, $\mathrm{Var}_{\infty}(x_0)$ is well-defined for the chaotic attractors as well, so that \eqref{soc_per_ch} continues to hold in this case. A theoretical condition guaranteeing this is provided by the standard ergodicity assumptions, since, by invoking Birkhoff's ergodic theorem, we would again have $\mathrm{Var}_{\infty}(x_0)=\mathrm{Var}_{\mu_A}$ for the corresponding ergodic physical measure. As a result, for the attractors of interest in the present work, \eqref{soc_per_m} and \eqref{soc_per_ch} imply that the normalized time-averaged social cost is minimized when $x^\star$ is stable and $x_0$ belongs to its basin of attraction. Every other orbit induces a positive efficiency loss and, in fact, for the normalized average excess social cost we have $\bar{\mathcal{C}}(x_0)-c(x^\star)=\mathrm{Var}_{\infty}(x_0)$. Importantly, we conclude that not only can distinct coexisting attractors yield different time-averaged social costs, but different attractors belonging to the same family can do so as well. Moreover, recalling that the Price of Anarchy (PoA) \cite{chotibut2} for the orbit with initial condition $x_0$ is
\begin{equation}
\mathrm{PoA}(x_0)=\frac{\bar{\mathcal{C}}(x_0)}{c(b)}=\frac{(x^\star-b)^2}{b(1-b)}+\frac{\mathrm{Var}_{\infty}(x_0)}{b(1-b)},
\label{anarchy}
\end{equation}
as $c(\cdot)$ is minimized at $b$, we also note that the PoA depends directly on the variance of the attractor and is minimized when $x^\star$ is stable and $x_0$ belongs to its basin of attraction.

\section{Bifurcation analysis and chaotic dynamics}\label{section2}

\begin{table}[!h]
\centering
\begin{threeparttable}
\caption{Primary attractor families discussed in the text. The bifurcation curves refer to those shown in Figs.~\ref{bifs1}, \ref{bifs2}, and \ref{bifs3}. The descriptor ``t'' denotes an attracting invariant curve. When the original attractor is periodic or a fixed point, the last column indicates its period. Saddle-node bifurcations are denoted by SN, and Neimark--Sacker bifurcations by NS. All families are destroyed through the corresponding boundary crisis, except for family VII. The families below the line refer to attractor families obtained through the merging of the principal families presented above the line.}
\label{tbl1}

\begin{tabular*}{.85\linewidth}{@{\extracolsep{\fill}} l l l @{}}
\toprule
Family & Birth & Original Family's Period\\
\midrule
I  & SN$_1$  &  3\\
II & --  &   1  \\
II--t  & NS$_2^1$ or reverse boundary crisis &   --\\
III & SN$_3$  & 9\\
IV & SN$_4$ & 4\\
V & SN$_5$ & 7\\
II--$i$ & SN$_{2-i}$ or NS$_{2-5}$ & $i$\\
VI& SN$_{6-j}$ or NS$_{6-j}$\tnote{1} & $j$\\
VII& SN$_{7-5}$& 5\\
\midrule
I/IV/V& Interior or reverse boundary crisis& --\\
I/II& Interior crisis& --\\
I/II/IV& Interior crisis&--\\
\bottomrule
\end{tabular*}

\begin{tablenotes}
\item[1] The SN$_{6-j}$ curves are denoted by SN$_j$ in Figs.~\ref{bifs2} and \ref{bifs3} for simplicity, while the labels NS$_{6-j}$ have been omitted from the corresponding curves in the figures for clarity. 
\end{tablenotes}

\end{threeparttable}
\end{table}

In this Section, we study in detail the nonlinear dynamics and the associated bifurcations of the system described by \eqref{sys}. The analysis is based on numerical techniques, since analytical tools (such as the center manifold theorem or normal form theory) are only valid locally around the corresponding equilibrium or cycle and would not reveal the global dynamics. We consider the parameters of the extreme event $p_0,p_1,M$ and the parameters $b,\alpha$ fixed, and study how the parameters controlled by the agents, namely $a$ and $k$, affect the learning dynamics. We assume, as stated, that $b\geq cp_0$, which induces non-trivial dynamics, as otherwise the $(0,p_0)$ equilibrium is the global attractor of the system. In the following, let $p_0=10^{-5},\; p_1=10^{-3}, \; M=15, \; \alpha=0.97, \; b=0.3$. These parameters constitute a naturally selected set based on the physical meaning of each parameter, while at the same time allowing for a comprehensive exploration of the dynamical complexity inherent in the system under consideration. Different choices of these parameters would result in quantitative differences, retaining the qualitative picture discussed next. In the last subsection we elaborate briefly on this matter. The bifurcation curves have been computed via continuation techniques following  \cite{matcont,datseris}. More details on the computational aspects are discussed in Appendix \ref{app2}.

Despite the regularity results discussed in the previous Section, the dynamics is fairly complex and for this reason we divide the $(a,k)$ parameter space into distinct regions, on which we focus. Although most bifurcations are contained in, and can be described within, the first two regions, we also study a third region to highlight differences in the stability, coexistence and merging of the attractors depending on the order of the occurring bifurcations. Importantly, we show in the following that, despite the complexity of the underlying dynamics, the attractors are organized into specific families, and the mechanisms governing their birth and destruction follow similar qualitative patterns. Table \ref{tbl1} summarizes the main attractor families of the system, with Arabic numerals replacing the Roman numerals in the respective figures for the sake of clarity. In what follows saddle-node bifurcations are denoted by SN, period-doubling bifurcations by PD, and Neimark--Sacker bifurcations by NS.

\subsection{Dynamics within the first region}\label{sec3.1}

\begin{figure}[!t]
    \centering
    \includegraphics[width=4.5in]{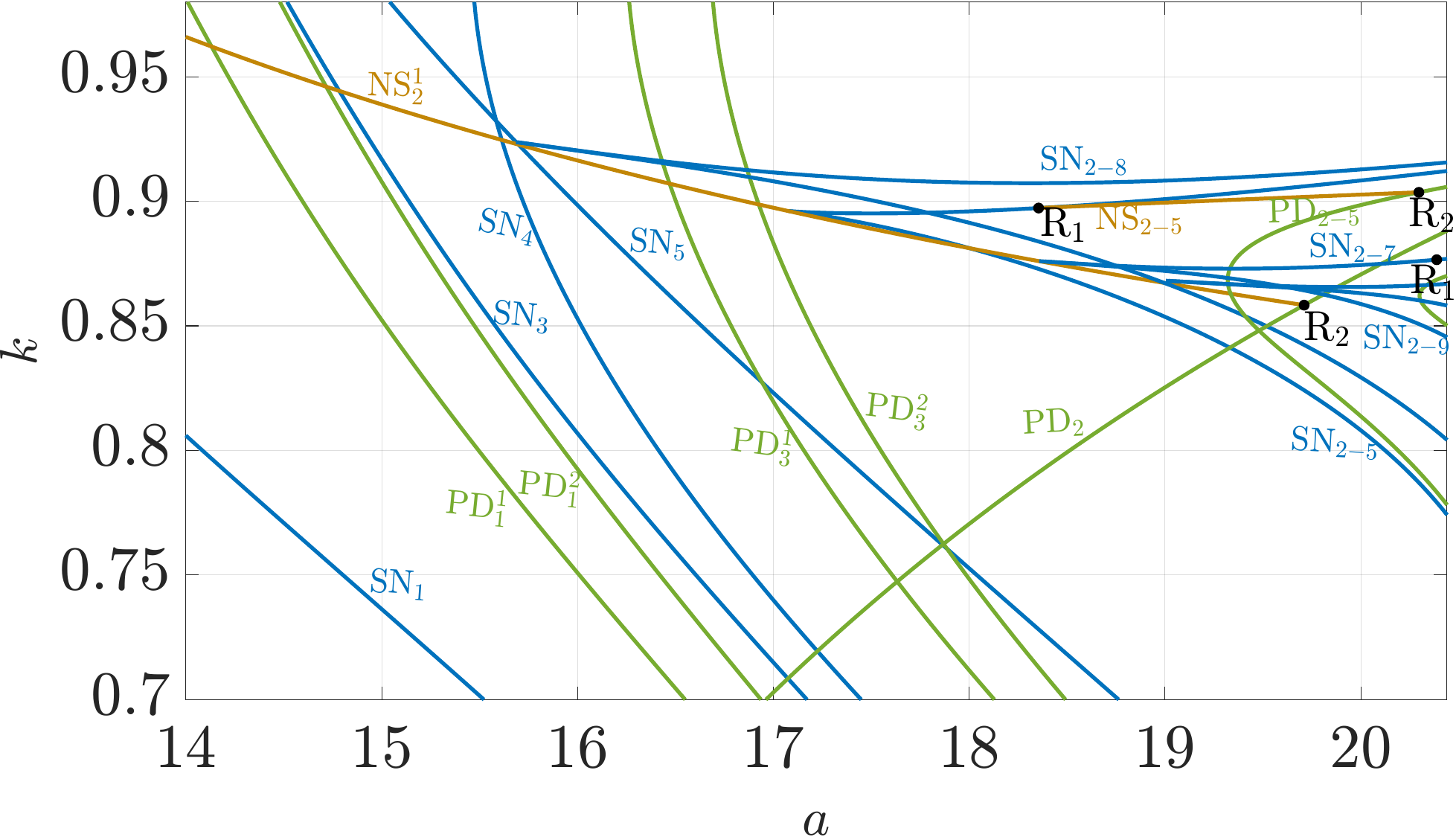}
    \caption{\small Bifurcation curves in the first region. Saddle-node bifurcations are denoted by SN, period-doubling bifurcations by PD, and Neimark--Sacker bifurcations by NS. Bifurcation curves of the same type share the same color. The index of each curve refers to the attractor family of Table~\ref{tbl1}. The superscript in the PD and NS curves refers to their order in the respective family. The labels R$_1$ correspond to $1:1$ strong resonances, while the labels R$_2$ to $1:2$ strong resonances.}
    \label{bifs1}
\end{figure}
We begin with the region in Fig. \ref{bifs1}, containing the primary bifurcation curves of the periodic orbits. To illustrate the qualitatively distinct attractors, suppose we fix $k=0.865$ and gradually increase $a$. Fig. \ref{bifa1} depicts the resulting bifurcations as $a$ varies. We have colored each orbit blue if the maximum Lyapunov exponent satisfies $\lambda<0$, indicating regular behavior, and red if $\lambda\geq 10^{-2}$, indicating chaotic behavior. At low values of the parameter $a$, the Nash equilibrium is the unique attractor of the system. As the learning rate, and thus $a$, increases, a stable period-three cycle (family I, cf. Table \ref{tbl1}) emerges corresponding to the crossing of the SN$_1$ curve of Fig. \ref{bifs1}. Its basin of attraction is shown in Fig. \ref{basins1} (a). This cycle undergoes a cascade of period-doubling bifurcations, the first two of which correspond to the PD$_1^1$, PD$_1^2$ curves of Fig. \ref{bifs1}. We note that, in relation to the normalized time-averaged social cost discussed in Section \ref{sec2.2} and \eqref{soc_cost_2}, it is minimized at the fixed point, while it varies with $a$ along each of the attractor families discussed. As a particular example, Fig. \ref{bifa1}(d) shows the normalized average excess social cost of the family I attractor as $a$ varies, which, by the discussion in Section \ref{sec2.2}, equals the limiting variance. We observe a continuous dependence of the excess social cost on $a$ in the periodic regime, as the points of the periodic orbit move farther from the equilibrium and the successive period-doubling bifurcations unfold. Moreover, in the chaotic regime (cf. Fig. \ref{bifa1}(a)), we numerically verify that the limiting variance in \eqref{soc_per_ch} is well-defined. Interestingly, in this regime the excess social cost exhibits discontinuities and non-monotonic behavior. Nevertheless, its values in the chaotic regime are generally comparable to the largest values attained in the periodic regime. Thus, there is no universal value of the time-averaged social cost associated with a given attractor family, excluding, of course, the fixed point $(x^\star,q^\star)$. Similar considerations apply to the remaining attractor families studied in this section and are therefore not discussed separately.

\begin{figure}[t!]
 % \centering
 % \subfigure[]{
 %   \includegraphics[width=0.31\linewidth]{figures/bif_k_12.3_18.3-eps-converted-to.pdf}
 % }
 % \hfill
 % \subfigure[]{
  %  \includegraphics[width=0.31\linewidth]{figures/bif_k_18.6_19.3-eps-converted-to.pdf}
 % }
 % \subfigure[]{
  %  \includegraphics[width=0.31\linewidth]{figures/bif_k_0.87_0.93-eps-converted-to.pdf}
 % }
    \centering
  \subfigure[]{
    \includegraphics[width=0.483\linewidth]{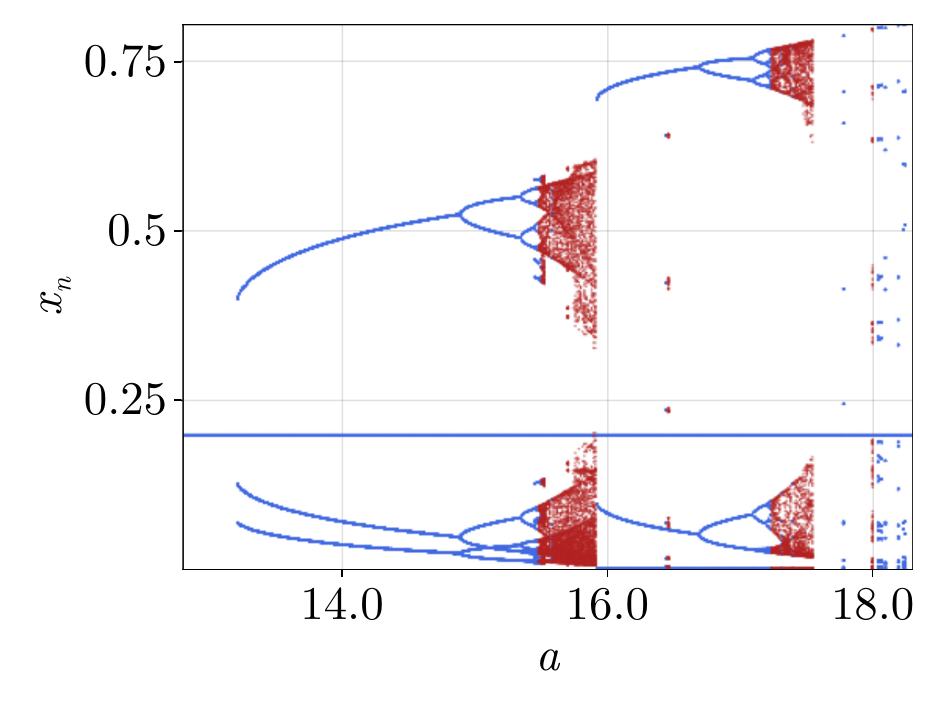}
  }
  \hfill
  \subfigure[]{
    \includegraphics[width=0.483\linewidth]{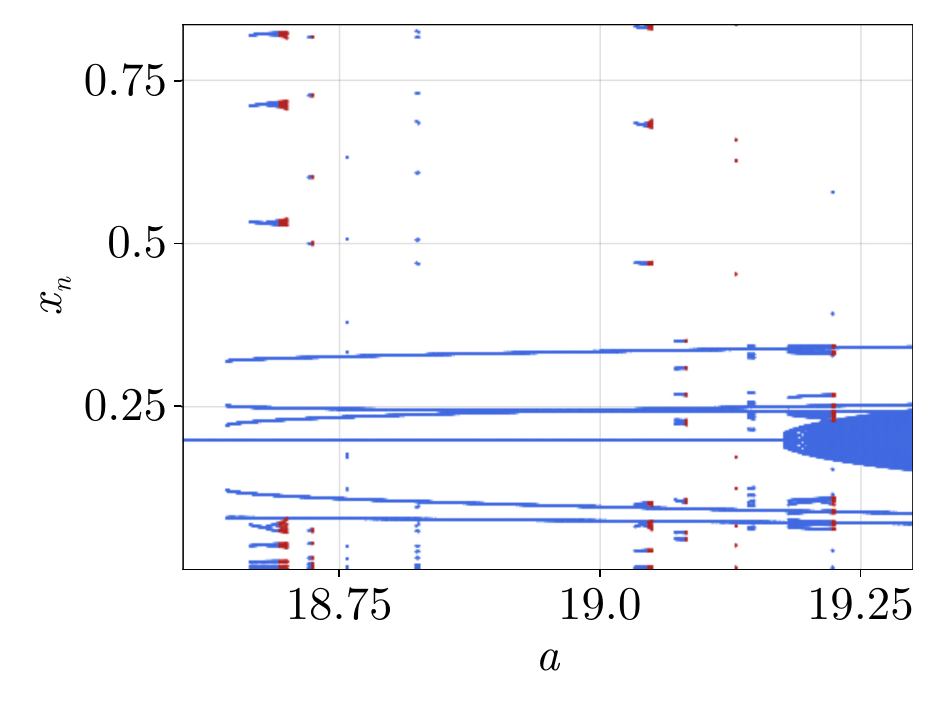}
  }
  
  \subfigure[]{
    \includegraphics[width=0.483\linewidth]{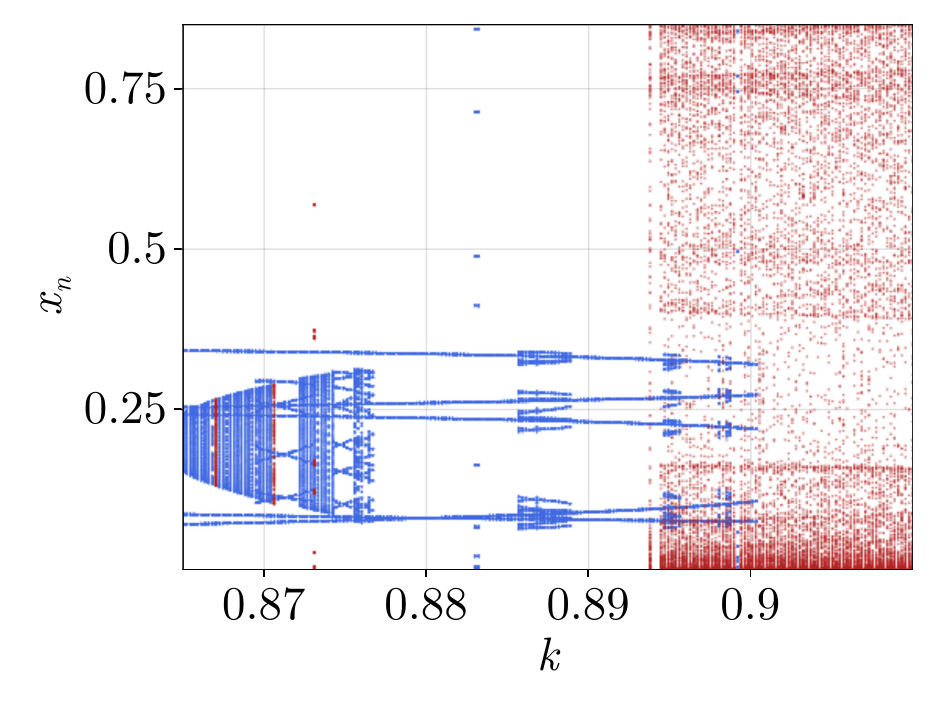}
  }
    \hfill
\subfigure[]{
    \includegraphics[width=0.483\linewidth]{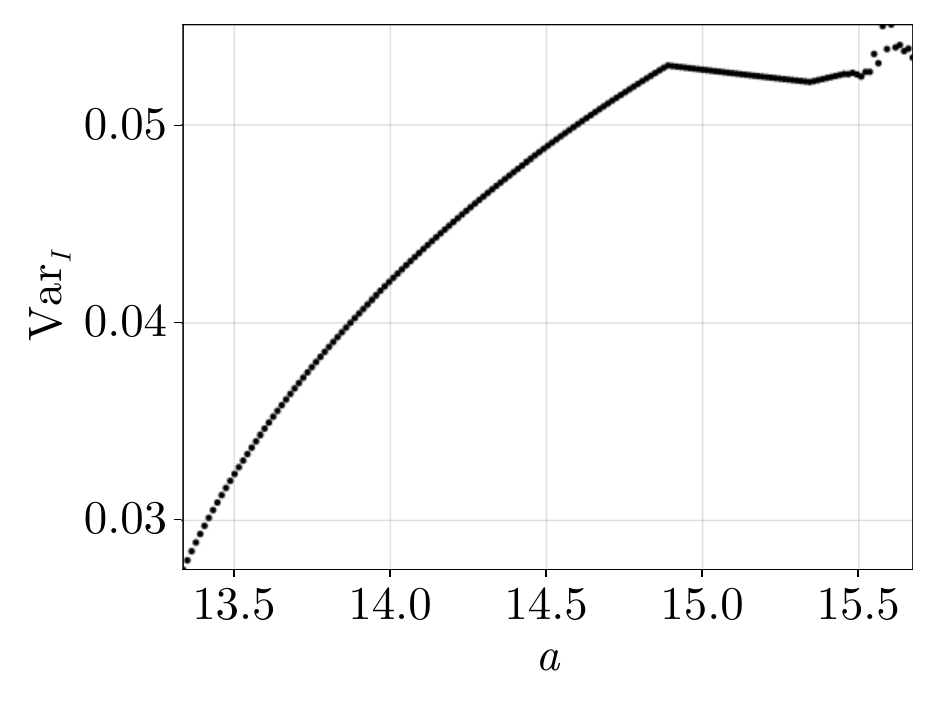}
  }

  \caption{\small (a), (b) The various emerging attractors as $k=0.865$ is held fixed and $a$ is gradually increased. Each attractor birth corresponds to the crossing of a bifurcation curve of Fig.~\ref{bifs1} or to a crisis. (c) Attractors as $a=19.3$ is fixed and $k$ varies. Orbits colored in blue correspond to regular behavior, characterized by a maximum Lyapunov exponent $\lambda<0$. Orbits colored in red correspond to chaotic behavior, with $\lambda\geq 10^{-2}$. (d) The average excess social cost, which equals the associated variance, of the family I attractor as a function of $a$. Note the discontiuity in the chaotic regime.
  }
\label{bifa1}
%\vspace*{12pt}
\end{figure}

As SN$_3$ is crossed a period 9 cycle is born (family III), which in turn experiences a Feigenbaum route to chaos (cf. also the positivity of the Lyapunov exponent as illustrated in Fig. \ref{bifa1}). The basins of attraction and the phase space, when both strange attractors and the stable equilibrium are present, are illustrated in Fig. \ref{basins1}(b) and Fig. \ref{phspace1}(a), respectively. It is worth noting that the basin boundaries have become intermingled and are of fractal type. A natural way to quantify this observation is via determining the basin entropy ($S_b$) and boundary basin entropy $S_{bb}$ \cite{entropy1,entropy2,entropy3}, reported in Table \ref{tblentropy}. The former quantifies the overall unpredictability of the system in terms of which attractor an initial condition converges to, while the latter measures the intensity of mixing at the boundary. Since $S_{bb}$ satisfies the fractality criterion proposed in \cite{entropy2,entropy3}, the basin boundaries have indeed fractal structure, in contrast to the boundaries in Fig. \ref{basins1} (a).

\begin{figure}
    \centering
    \includegraphics[width=0.41\linewidth]{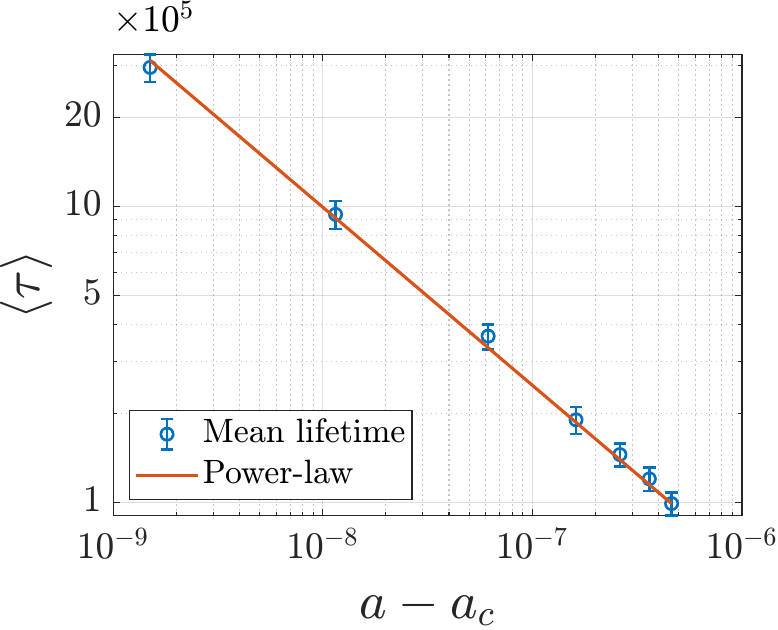}
    \hfill
    \includegraphics[width=0.41\linewidth]{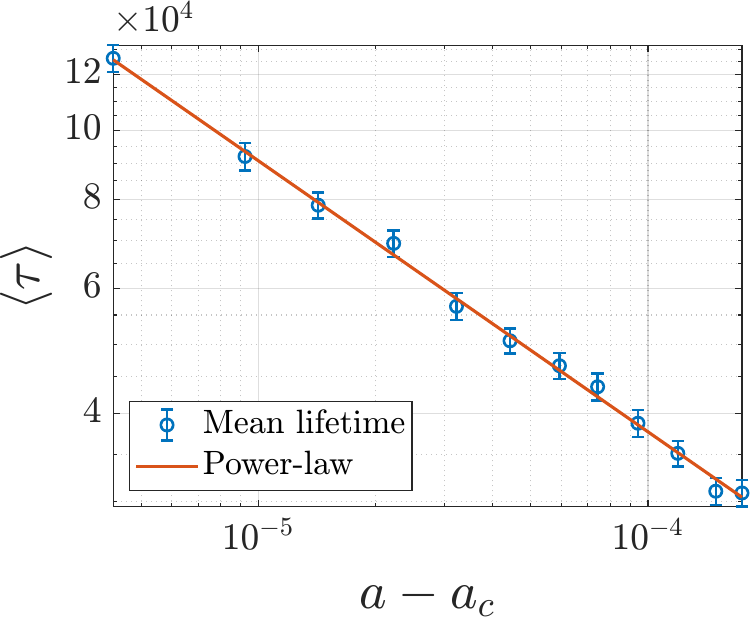}
    
    \caption{\small (a) Interior crisis. Mean laminar time between bursts from the old, pre-crisis part of the attractor and the newly accessible region. The dependence is of the form $\langle \tau \rangle \sim (a-a_c)^{-\gamma}$ with $a_c=15.74$ and $\gamma=0.60.$ (b) Boundary crisis. Mean lifetime of the chaotic transient after the attractor has been destroyed as a function of $a.$ The dependence is of the form $\langle \tau \rangle \sim (a-a_c)^{-\gamma}$ with $a_c=15.89$ and $\gamma=0.38.$
    }
    \label{crisis_plots}
\end{figure}

\begin{table}[!t]
\centering
\caption{\small Basin and boundary basin entropy. According to the fractal criterion proposed in \cite{entropy2,entropy3}, the basin boundary has fractal structure if $S_{bb}>0.439$.}
\label{tblentropy}

\begin{tabular*}{0.85\linewidth}{@{\extracolsep{\fill}} l c c l c c @{}}
\toprule
Figure & $S_b$ & $S_{bb}$ &
Figure & $S_b$ & $S_{bb}$\\
\midrule

Fig.~\ref{basins1}(a) & 0.110 & 0.350 &
Fig.~\ref{basins3}(e) & 0.790 & 0.810\\

Fig.~\ref{basins1}(b) & 0.386 & 0.572 &
Fig.~\ref{basins3}(f) & 0.589 & 0.611\\

Fig.~\ref{basins1}(c) & 0.198 & 0.540 &
Fig.~\ref{basins3}(g) & 0.879 & 0.898\\

Fig.~\ref{basins2}(a) & 0.220 & 0.563 &
Fig.~\ref{basins3}(h) & 0.620 & 0.639\\

Fig.~\ref{basins2}(b) & 0.490 & 0.577 &
Fig.~\ref{coex_16_25}(b) & 0.731 & 0.731\\

Fig.~\ref{basins2.5}(e) & 0.790 & 0.826 &
Fig.~\ref{basins4}(e) & 0.012 & 0.423\\

Fig.~\ref{basins2.5}(f) & 0.650 & 0.635 &
Fig.~\ref{basins4}(f) & 0.001 & 0.220\\

Fig.~\ref{basins2.5}(g) & 0.795 & 0.826 &
Fig.~\ref{basins4}(g) & 0.006 & 0.556\\

Fig.~\ref{basins2.5}(h) & 0.588 & 0.620 &
Fig.~\ref{basins4}(h) & 0.096 & 0.241\\

\bottomrule
\end{tabular*}

\end{table}

    Varying further the parameters, the chaotic attractor originating from the nine-period cycle experiences a basin crisis bifurcation, resulting in its conversion to a chaotic saddle. Qualitatively, the boundary crisis can be identified when a chaotic attractor suddenly disappears and is replaced by a chaotic saddle, as in the case discussed here. Transient chaos is then typically observed, with orbits spending long periods near the chaotic saddle before converging to the respective attractors; for a visualization of such a transient, see also Fig. \ref{basins2}(c), which shows a qualitatively similar case, and the discussion that follows shortly. Near the boundary crisis, the characteristic time, i.e., the mean lifetime $\langle \tau \rangle$ of the chaotic transient, diverges according to $(a-a_c)^{-\gamma}$, where $a_c$ is the parameter value at which the crisis occurs and $\gamma$ represents the crisis exponent \cite{crisis2}. This property can be thought of as the signature of the boundary crisis. We elaborate further on this behavior in the next paragraph by considering another representative boundary crisis that illustrates this phenomenon more clearly. A part of its saddle set collides with the basin of the chaotic attractor of family I, leading to its expansion, a typical qualitative signature of an interior crisis bifurcation, as shown in Fig. \ref{phspace1}(b) and Fig. \ref{basins1}(c). Near the interior crisis bifurcation, the characteristic time scale again diverges according to $\langle \tau \rangle \sim (a-a_c)^{-\gamma}$, but in this case $\langle \tau \rangle$ denotes the mean laminar time between bursts from the old, pre-crisis region of the attractor into the newly accessible region \cite{crisis2}. In other words, the crisis gives rise to crisis-induced intermittency. This property can be thought of as the signature of the interior crisis. Fig. \ref{crisis_plots}(a) shows the dependence of $\langle \tau \rangle$ on $a$, verifying the power-law scaling $\langle \tau \rangle \sim (a-a_c)^{-\gamma}$ with $a_c=15.74$. Similar considerations apply to the subsequent boundary crises discussed in this work as well. The associated basin and boundary basin entropies are reported in Table \ref{tblentropy}. We keep referring to this expanded family as I.

\begin{figure}[t!]
  \centering
  \subfigure[]{
    \includegraphics[width=0.44\linewidth]{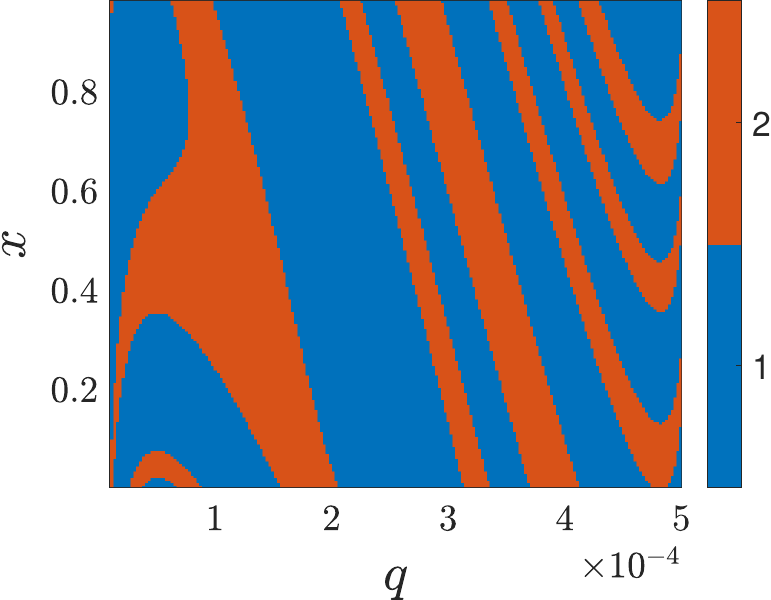}
  }
  \hfill
  \subfigure[]{
    \includegraphics[width=0.44\linewidth]{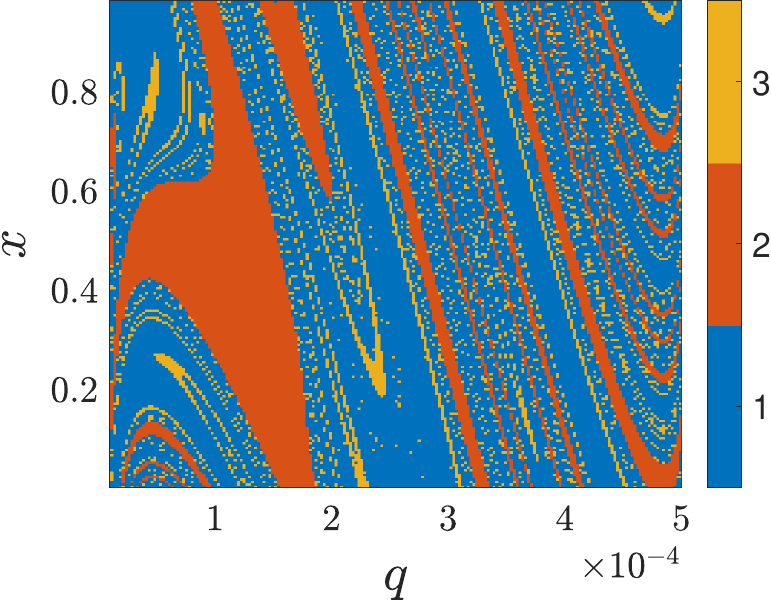}
  }

\subfigure[]{
    \includegraphics[width=0.44\linewidth]{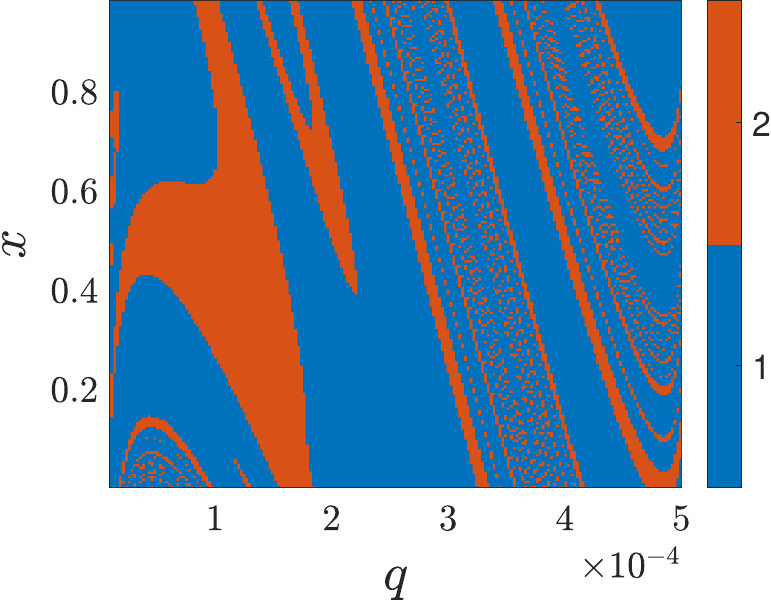}
  }
  
  \caption{\small Basins of attraction of the coexisting attractors with $k=0.865$ fixed, corresponding to Fig. \ref{bifa1}(a). The numerical labels $1,2,\ldots$ correspond to the Roman-numeral family indices I, II, \ldots, defined in Table~\ref{tbl1}. (a) $a=14.02:$ The period three orbit and the equilibrium coexist. (b) $a=15.5225:$ The chaotic attractor of family I, the equilibrium and the chaotic attractor of family III coexist. (c) $a=15.8:$ The chaotic attractor of family I and the equilibrium coexist. Note the intermingling of the basins in (b) and (c).}
\label{basins1}
%\vspace*{12pt}
\end{figure}

\begin{figure}[t!]
  \centering
  \subfigure[]{
    \includegraphics[width=0.41\linewidth]{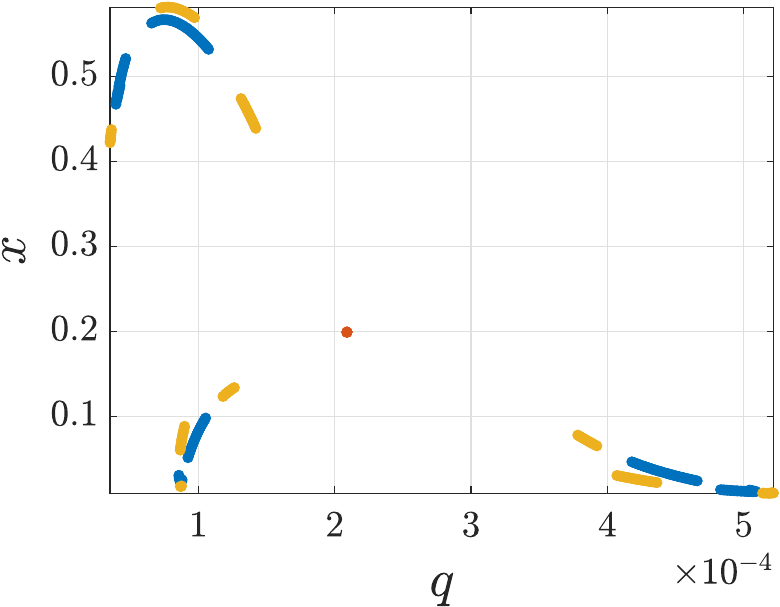}
  }
  \hfill
  \subfigure[]{
    \includegraphics[width=0.41\linewidth]{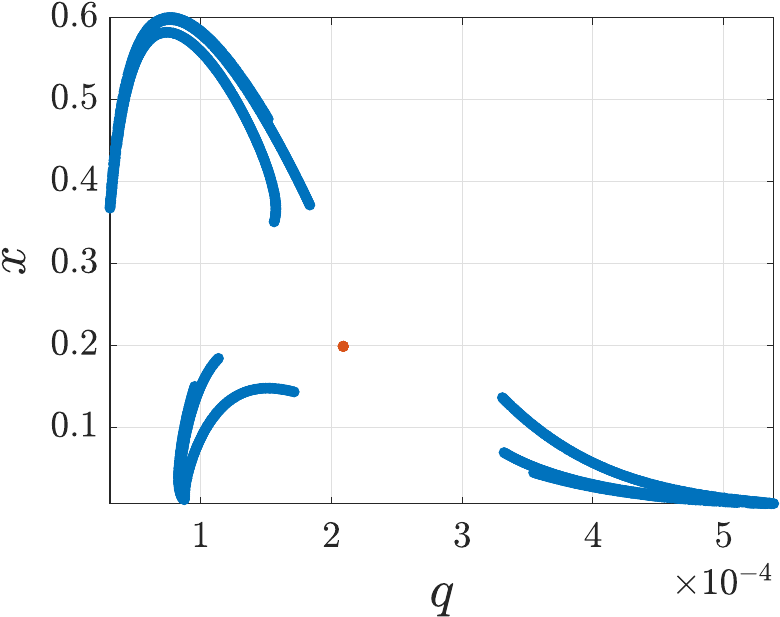}
  }

\subfigure[]{
    \includegraphics[width=0.41\linewidth]{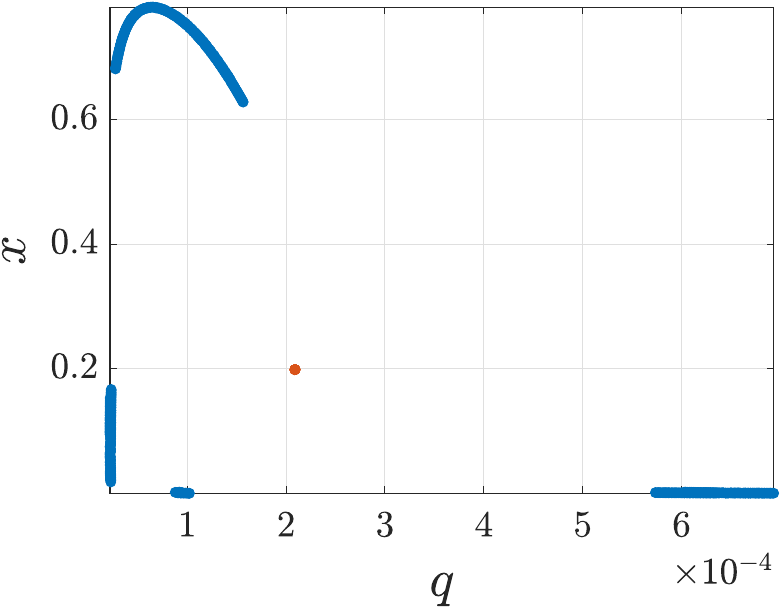}
  }
  \hfill
\subfigure[]{
    \includegraphics[width=0.41\linewidth]{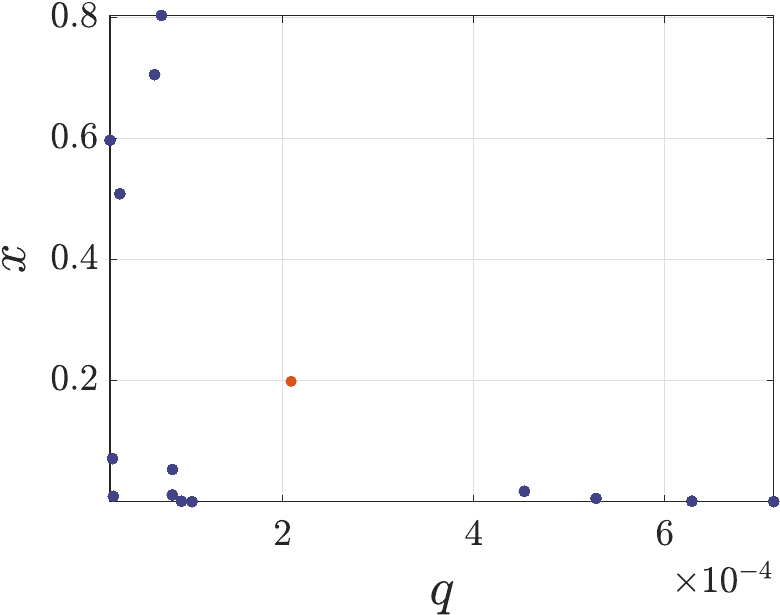}
  }
  
  \caption{\small Coexisting attractors in the phase space. (a) $a=15.5225:$ The chaotic attractor of family I, the equilibrium and the chaotic attractor of family III coexist. (b) $a=15.8:$ The chaotic attractor of family I and the equilibrium coexist. (c) $a=17.54:$ The merged chaotic attractor I/IV/V and the equilibrium coexist. (d) $a=18.245:$ A high-period orbit of the I/IV/V family coexists with the stable equilibrium. In all cases, $k=0.865$ is fixed.}
\label{phspace1}
%\vspace*{12pt}
\end{figure}

Crossing SN$_4$ of Fig. \ref{bifs1} results in the birth of a stable four-period cycle (family IV), while almost simultaneously another boundary crisis bifurcation involving the chaotic attractor of family I results in its transformation into a chaotic saddle (cf. Fig. \ref{bifa1}(a)). As discussed previously, near the boundary crisis, the characteristic time, i.e., the mean lifetime $\langle \tau \rangle$ of the chaotic transient after the chaotic saddle has formed, diverges according to $(a-a_c)^{-\gamma}$, where $a_c$ is the parameter value at which the crisis occurs and $\gamma$ is the corresponding crisis exponent. As shown in Fig. \ref{crisis_plots}(b), the dependence of the mean lifetime near the boundary crisis follows $(a-a_c)^{-\gamma}$, with $a_c=15.89$ and $\gamma=0.38$, thus providing a quantitative verification of this characteristic feature of the boundary crisis. Similar considerations apply to the subsequent boundary crises, for which analogous scaling laws can be obtained, albeit with different characteristic exponents. A stable period-seven cycle (family V) appears at SN$_5$, coexisting with families II and IV, and, within a very narrow parameter region, period-doubles to chaos (cf. also the positivity of the Lyapunov exponent as illustrated in Fig. \ref{bifa1}) and in turn transforms to a chaotic saddle through a boundary crisis. Family IV undergoes a similar pattern of Feigenbaum cascades, expansion through an interior crisis and eventual destruction through a boundary crisis. The phase space and the basin of attraction before the destruction of the IV family chaotic attractor are shown in Fig. \ref{phspace1}(c) and Fig. \ref{basins2}(a). Although the families I, IV and V have lost their stability, they affect the dynamics of the system through the infinitely many unstable periodic orbits living inside the chaotic saddle. In fact, the chaotic saddles of these families become connected, while the embedded unstable periodic orbits lead to the emergence of high-period stable periodic orbits--whose family we call, abusing notation, I/IV/V--within small regions of the parameter space, which in turn undergo period-doubling bifurcations before they get destroyed via boundary crises. Typical examples are shown in Fig. \ref{bifa1}(a), where, within narrow regions of the parameter space, high-period cycles follow the described pattern. It is worth noting that the basins of the I/IV/V-type attractors are highly intermingled--cf. Fig. \ref{basins2}(b) and Table \ref{tblentropy}. The phase space, when a high-period orbit of the I/IV/V family coexists with the stable equilibrium is shown in Fig. \ref{phspace1}(d). Another manifestation of the chaotic saddle is that, despite the stationary equilibrium being the sole attractor after the boundary crisis, the system is transiently chaotic with orbits spending a large amount of time near the chaotic saddle, and thus having positive finite Lyapunov exponent, as illustrated in Fig. \ref{basins2} (c). We note that none of the I,III,IV and V periodic attractors exhibit any codimension-two bifurcations.

\begin{figure}[t!]
  \centering
  \subfigure[]{
    \includegraphics[width=0.46\linewidth]{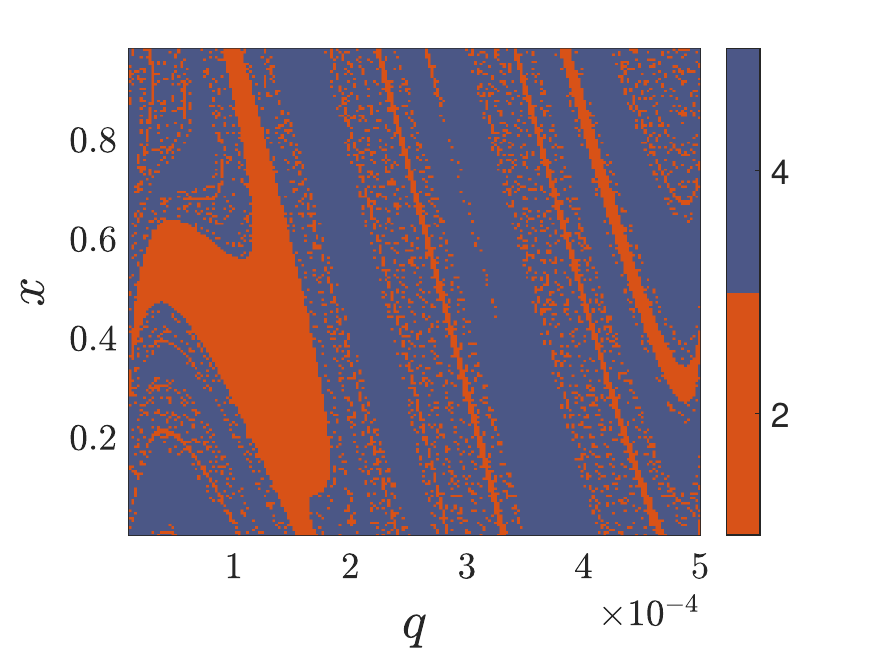}
  }
  \hfill
  \subfigure[]{
    \includegraphics[width=0.46\linewidth]{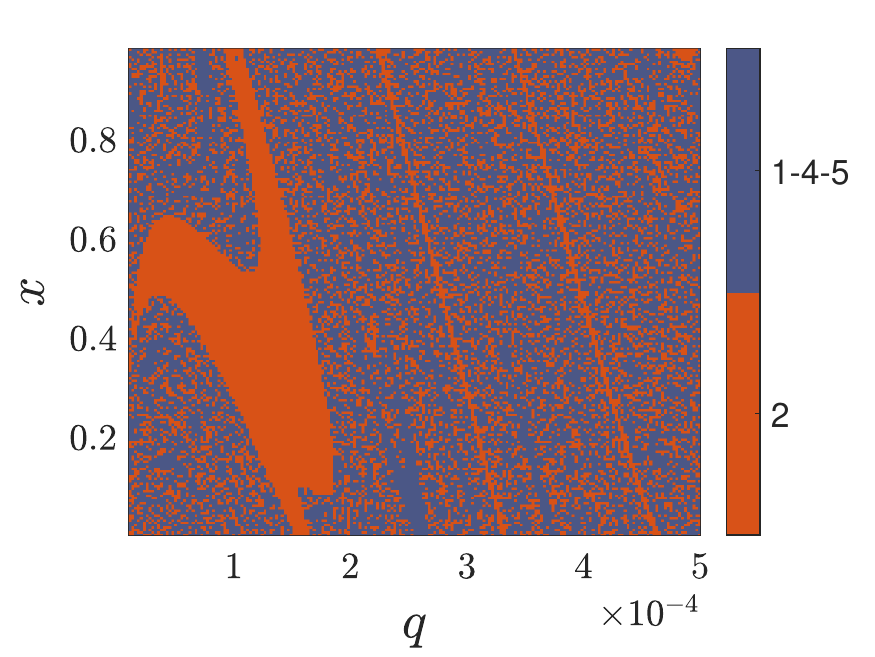}
  }

     \subfigure[]{
    \includegraphics[width=0.41\linewidth]{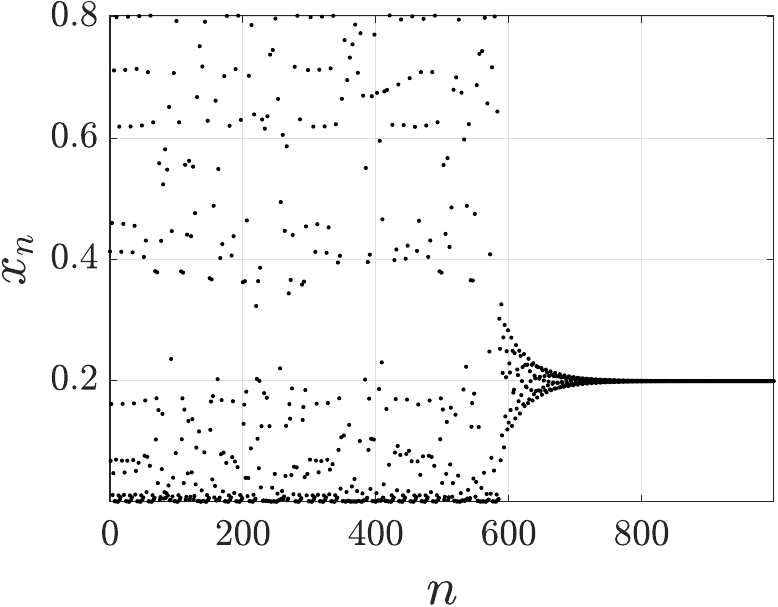}
  }
%   \hfill
%  \subfigure[]{
%    \includegraphics[width=0.3\linewidth]{figures/%basins_eq_5_8_a_19.0725_k_0.865-eps-converted-to.pdf}
 % }
 %  \hfill
%  \subfigure[]{
%    \includegraphics[width=0.3\linewidth]{figures/basins_tr_5_15_a_19.2213_k_0.865-eps-converted-to.pdf}
%  }
%   \hfill
%  \subfigure[]{
%    \includegraphics[width=0.3\linewidth]{figures/basins_tor_5_a_19.2282_k_0.865-eps-converted-to.pdf}
%  }
  \caption{\small (a) Basins of attraction with $a=17.54$ and $k=0.865$: the chaotic attractor of family IV and the equilibrium coexist. (b) Basins of attraction with $a=18.245$ and $k=0.865$: a high-period orbit of the I/IV/V family and the equilibrium coexist. In both cases the numerical labels correspond to the Roman-numeral family indices, defined in Table~\ref{tbl1}. Note that the basins are intermingled. (c) Illustration of the transient chaos of an orbit converging to the equilibrium due to the I/IV/V chaotic saddle.}
\label{basins2}
%\vspace*{12pt}
\end{figure}

%\begin{figure}[t!]
%    \centering
%    \includegraphics[width=0.4\linewidth]{figures/ch_tr_18.0927_a_k_0.865-eps-converted-to.pdf}
%    \caption{Transient chaos of an orbit converging to the equilibrium due to the I/IV/V chaotic saddle.}
%    \label{tr_time}
%\end{figure}

At NS$_2^1$ (cf. Fig. \ref{bifs1}), the stable equilibrium loses stability and a stable invariant curve emerges, whose family in Table \ref{tbl1} is referred to as II--t (essentially, it belongs to II, but we wish to emphasize its emergence). This invariant curve can be seen in Fig. \ref{bifa1}(b), (c), Fig. \ref{basins2.5}, and Fig. \ref{basins3}; more detailed discussion on these figures will follow next. Saddle-node bifurcations originating from phase-locking on the primary NS$_2^1$ curve are crossed (cf. the saddle-node SN$_{2-i}$ curves in Fig. \ref{bifs1}), resulting in the appearance of cycles of various periods. We call the families of these orbits II--$i$ in Table \ref{tbl1}, with $i$ representing their period, to stress that they are associated with the invariant curve, which in turn is connected to the equilibrium. In Fig. \ref{bifs1}, the period-eight, period-five, period-seven, and period-nine tongues are illustrated. Near NS$_2^1$, the Arnold tongues do not overlap; however, as the parameters move farther away from it, global bifurcation phenomena may lead to the coexistence of attractors originating from different tongues. These dynamical phenomena are qualitatively similar across all II--$i$ type attractors. For this reason, we first describe the generic scenario and then provide representative examples of specific II--$i$ orbits in what follows. 

\begin{figure}[t!]
  \centering
  \subfigure[]{
    \includegraphics[width=0.41\linewidth]{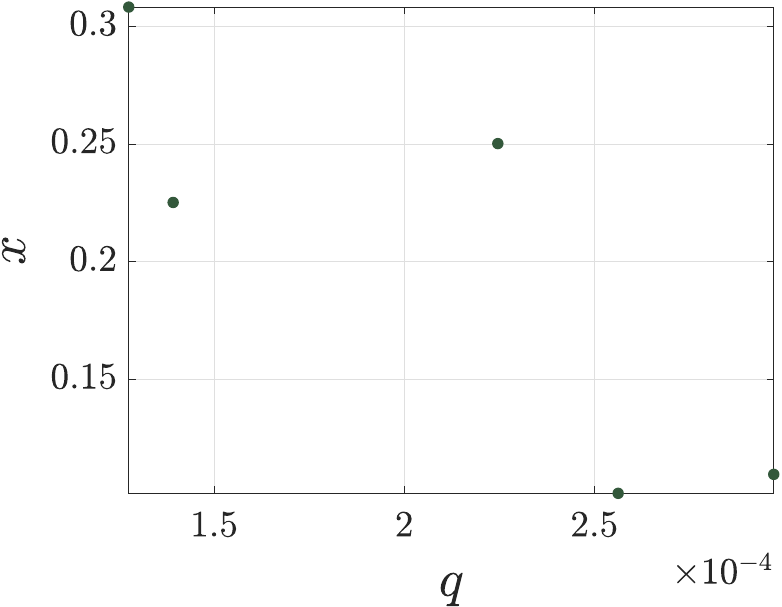}
  }
  \hfill
  \subfigure[]{
    \includegraphics[width=0.41\linewidth]{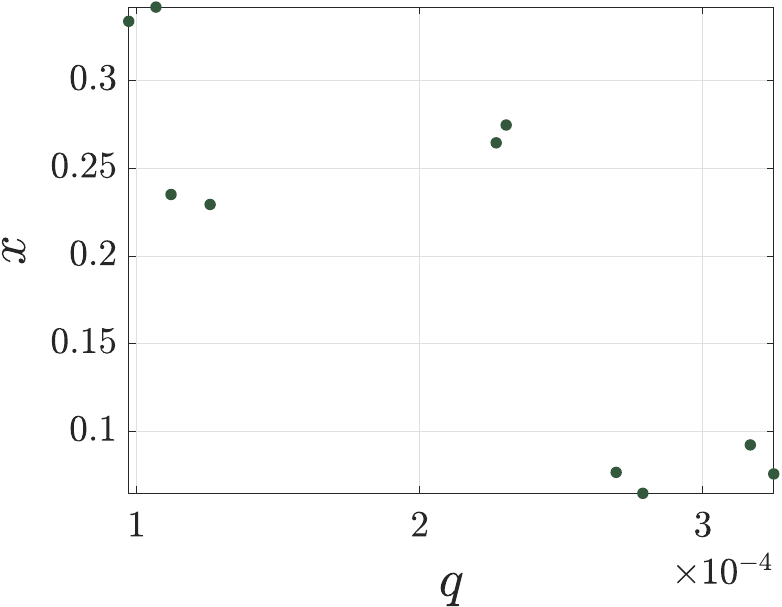}
  }

\subfigure[]{
    \includegraphics[width=0.41\linewidth]{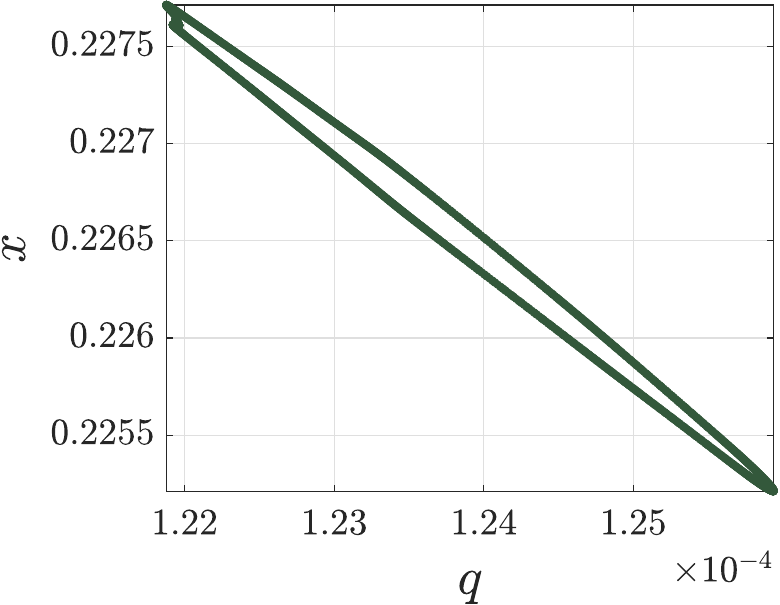}
  }
  \hfill
  \subfigure[]{
    \includegraphics[width=0.41\linewidth]{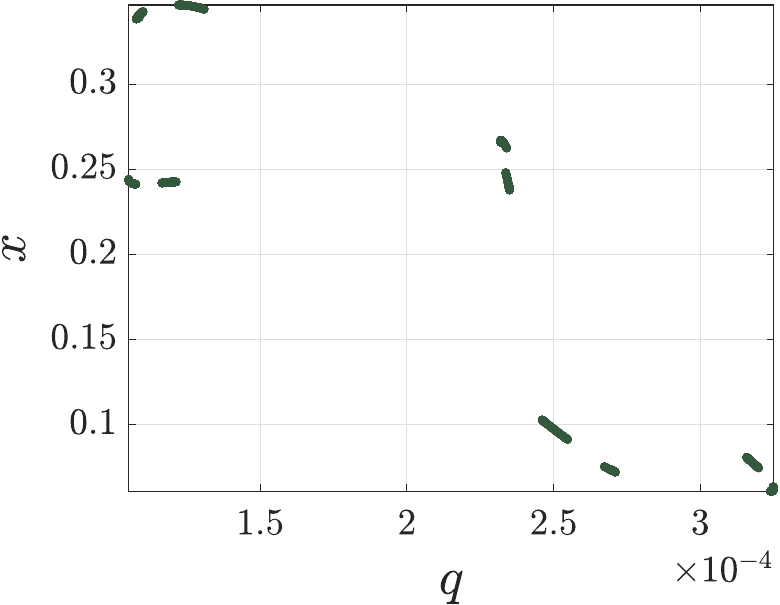}
  }
  \caption{\small The various attractors of the II--$5$ family in the phase space. (a) The attractor is a period-five orbit. (b) The attractor is a period-ten orbit, born at the PD$_{2-5}$ curve of Fig.~\ref{bifs1}. (c) The attractor is an invariant curve for the tenth iterate of the map, born at the NS$_{2-5}$ curve of Fig.~\ref{bifs1}. (d) The attractor is chaotic after the period-doubling cascade.}
\label{5phspace}
%\vspace*{12pt}
\end{figure}

First, note that the Arnold tongues are wider as the period of the resulting periodic orbit is lower. A $1:1$ resonance with positive normal-form coefficient--denoted by R$_1$ in Fig. \ref{bifs1}--, i.e., a codimension-two point where both multipliers are equal to 1, is found on the upper branch of the tongue. Codimension-two resonances act as organizing centers for the dynamics of the system. Crossing the saddle-node curves defining the upper branch from the left of R$_1$, a stable cycle is born, which, as we move around R$_1$, is destroyed in an NS bifurcation emanating from the R$_1$ point and a saddle invariant curve is born. We note that homoclinic structures are present near this point, implying the birth and destruction of infinitely many long-period saddle cycles \cite{kuznetsov1}. Inside the tongue, the corresponding periodic orbit--because it is possible for other orbits to coexist as the tongues overlap--undergoes a period-doubling route to chaos (cf. also the positivity of the Lyapunov exponent as illustrated in Fig. \ref{bifa1}), and ultimately the family is destroyed through a boundary crisis. In Figs. \ref{bifs1} and \ref{bifs2}, the first period-doubling curves in the Feigenbaum cascade of the period-five, period-seven, period-eight, and period-nine Arnold tongues are shown. It is worth noting that a $1:2$ resonance with negative normal-form coefficient--denoted by R$_2$ in Fig. \ref{bifs1}--, i.e., a codimension-two point where both multipliers are equal to -1, is found on the upper branch of the PD curve. Similar to the R$_1$ resonance, this point acts as an organizing center for the dynamics of the system. The NS curve, crossing which results in the destruction of the stable cycle and the birth of a saddle invariant curve, connects the two codimension-two points of the same family. Crossing the PD curve above the R$_2$ point results in the creation of an unstable period-doubled orbit, which gains stability as another NS curve emanating from the R$_2$ point is crossed. Within the region defined by the PD and (the second) NS curves, a stable invariant curve for the iterate of the map corresponding to the double period of the tongue is present. Similar to the $1:1$ resonance, near the R$_2$ point infinitely many long-period unstable cycles appear and disappear due to the involved homoclinic structures \cite{kuznetsov1}. The region around the NS curve of the fixed point is filled with overlapping resonance structures of the above-described type, thus implying multistability among the corresponding periodic cycles.

\begin{figure}[!p]
  \centering
  \subfigure[]{
    \includegraphics[width=0.38\linewidth]{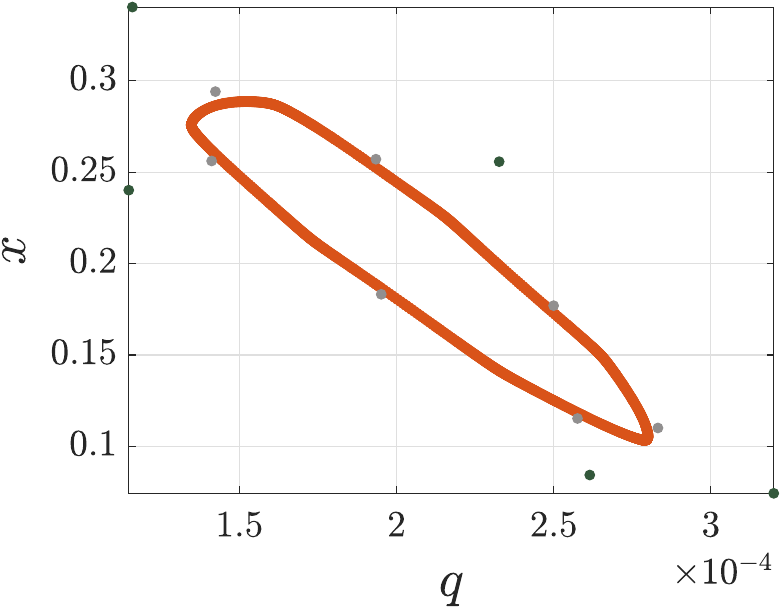}
  }
   \hfill
  \subfigure[]{
    \includegraphics[width=0.38\linewidth]{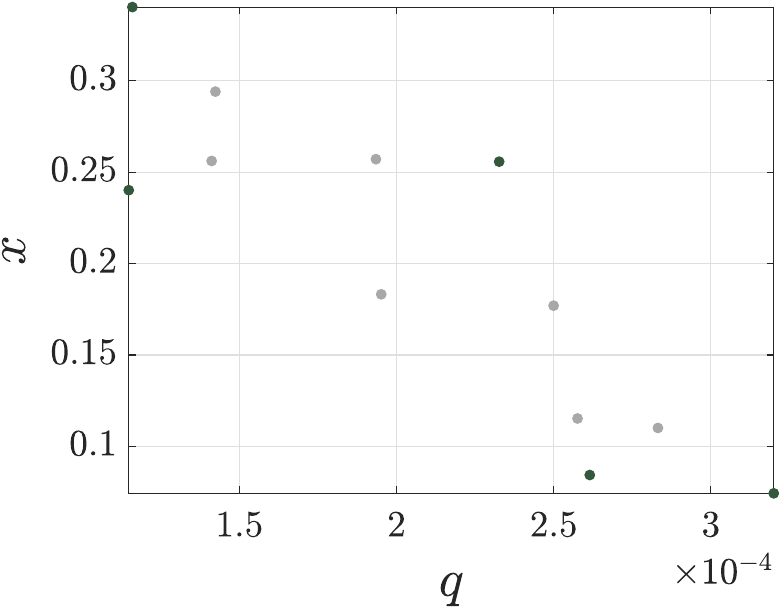}
  }
   \hfill
  \subfigure[]{
    \includegraphics[width=0.38\linewidth]{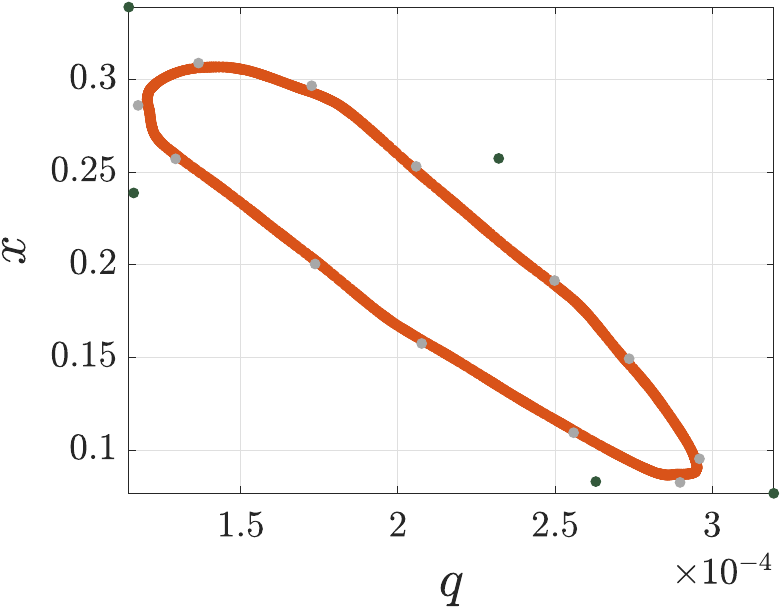}
  }
   \hfill
  \subfigure[]{
    \includegraphics[width=0.38\linewidth]{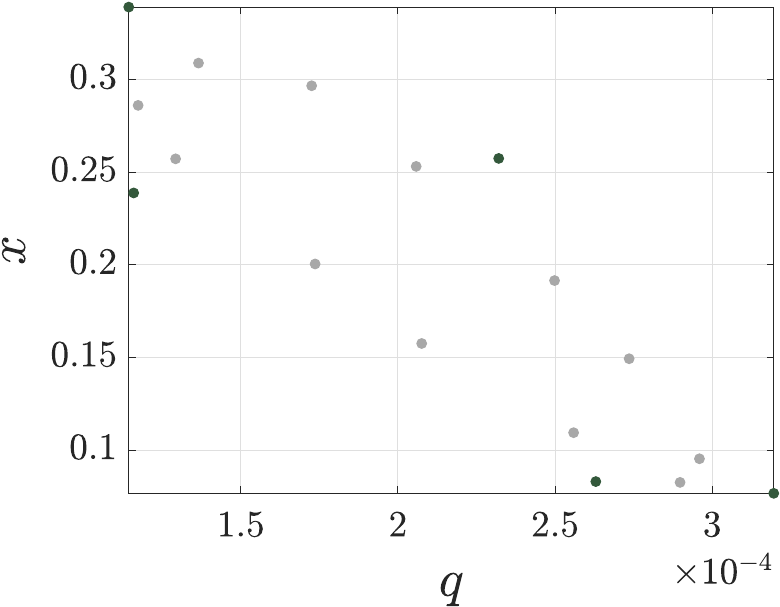}
  }
  \hfill
  \subfigure[]{
    \includegraphics[width=0.4\linewidth]{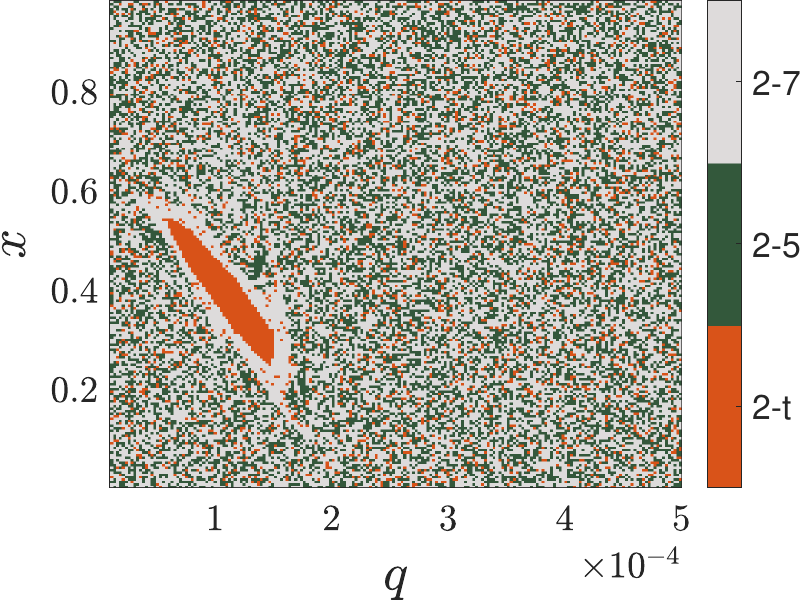}
  }
  \hfill
  \subfigure[]{
    \includegraphics[width=0.4\linewidth]{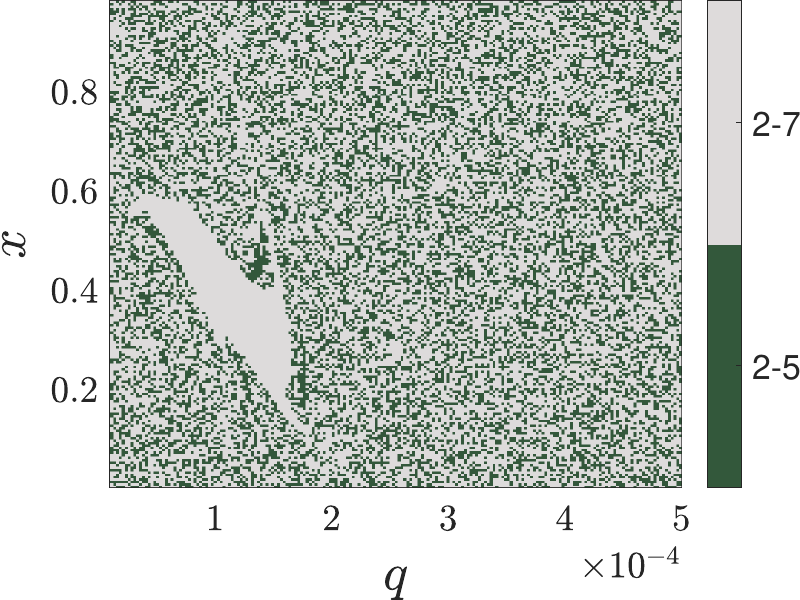}
  }
  \subfigure[]{
    \includegraphics[width=0.4\linewidth]{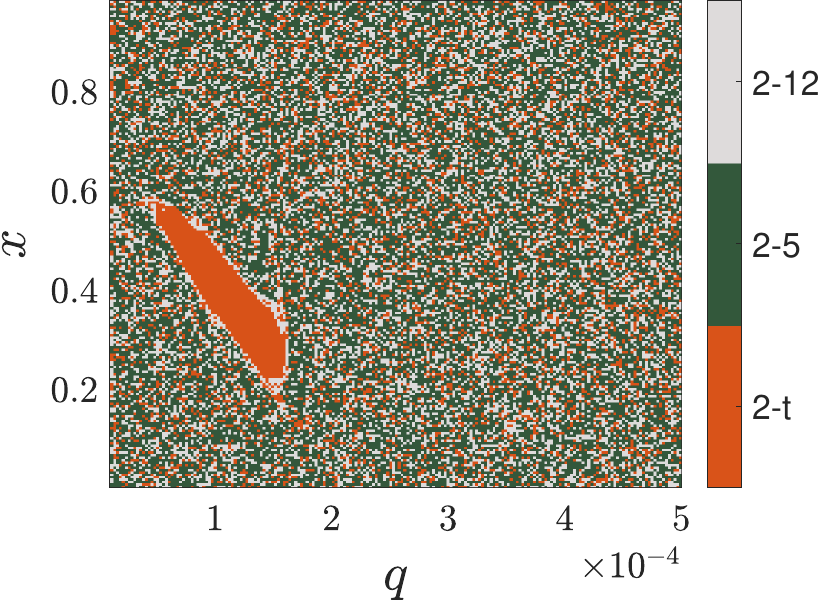}
  }
  \hfill
  \subfigure[]{
    \includegraphics[width=0.4\linewidth]{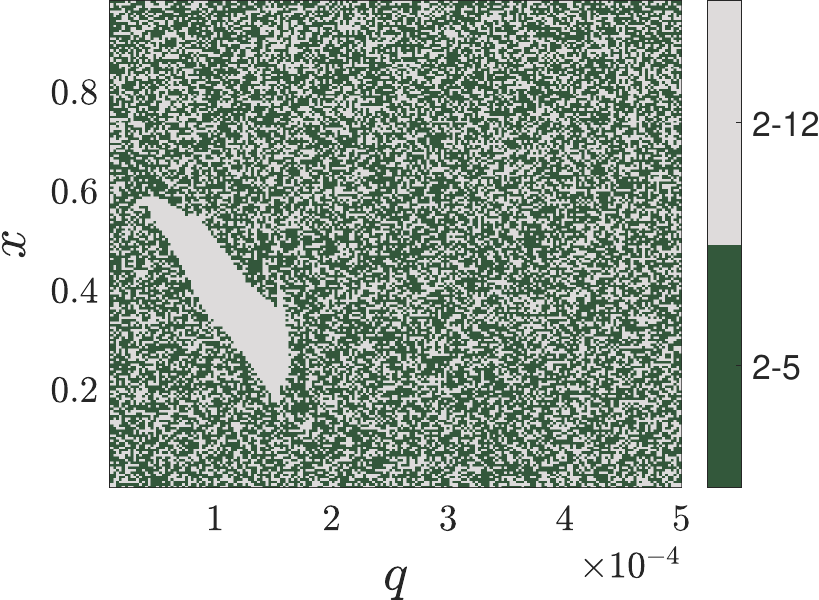}
  }
  \caption{\small Phase space (a)-(d) and associated basins of attraction (e)-(h) with $a=19.3$ fixed, corresponding to Fig.~\ref{bifa1}(c). (a) and (e): $k=0.8706,$ the period-five attractor of the II--$5$ family, the period-seven attractor of the II--$7$ family, and the stable invariant curve coexist. (b) and (f): $k=0.8708,$ the period-five attractor of the II--$5$ family and the period-seven attractor of the II--$7$ family coexist; the invariant curve has been destroyed. (c) and (g): $k=0.8744,$ the period-five attractor of the II--$5$ family, the period-twelve attractor of the II--$12$ family, and the reemerged invariant curve coexist. (d) and (h): $k=0.8747,$ the period-five attractor of the II--$5$ family and the period-twelve attractor of the II--$12$ family coexist.}
\label{basins3}
%\vspace*{12pt}
\end{figure}

\begin{figure}[!p]
  \centering
    \subfigure[]{
    \includegraphics[width=0.38\linewidth]{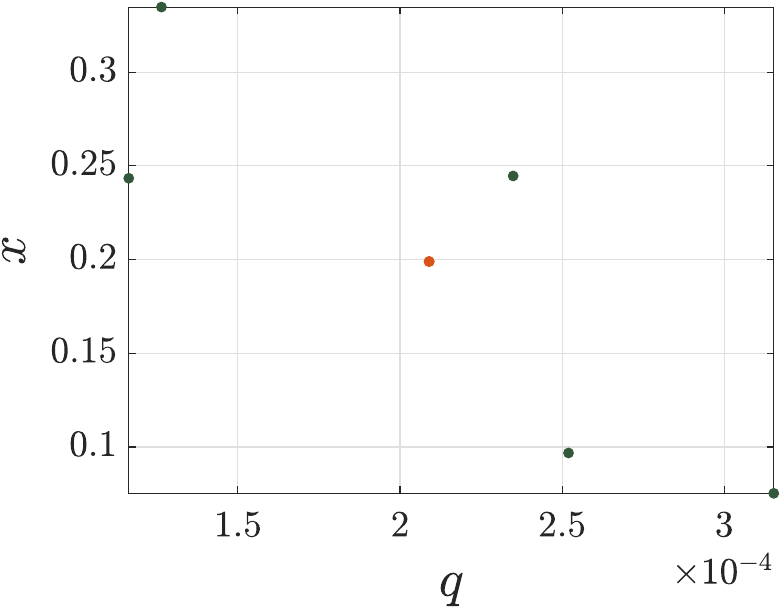}
  }
  \hfill
  \subfigure[]{
    \includegraphics[width=0.38\linewidth]{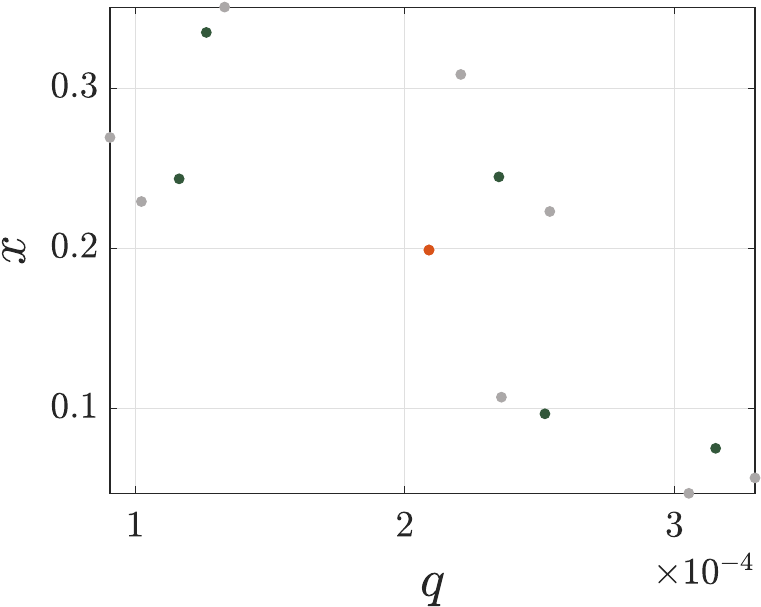}
  }
   \hfill
  \subfigure[]{
    \includegraphics[width=0.38\linewidth]{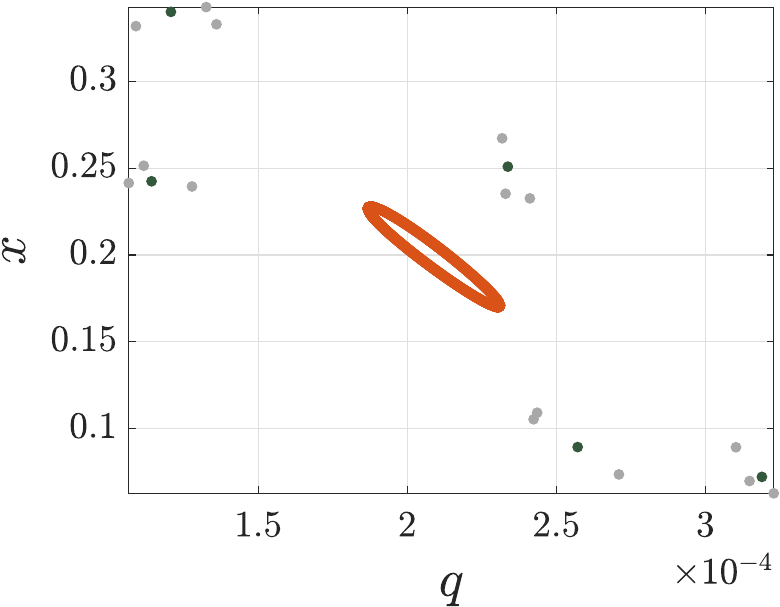}
  }
   \hfill
  \subfigure[]{
    \includegraphics[width=0.38\linewidth]{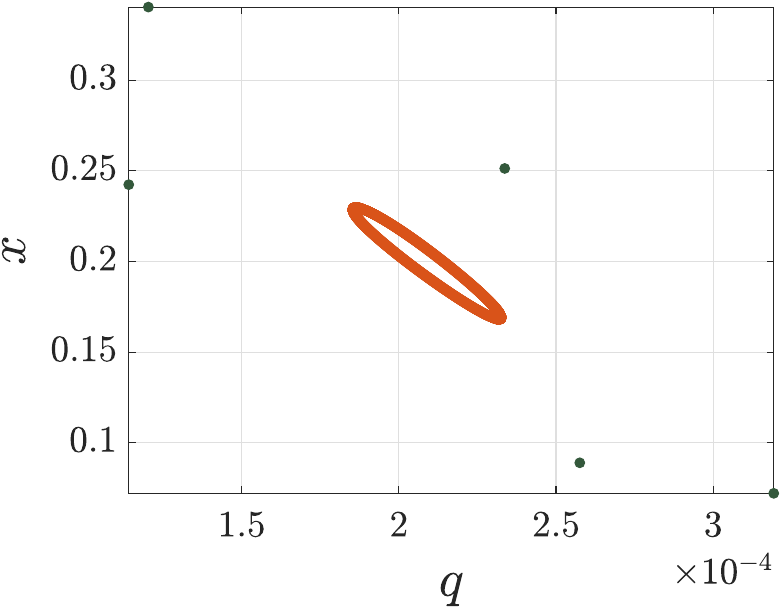}
  }
  \hfill
  \subfigure[]{
    \includegraphics[width=0.4\linewidth]{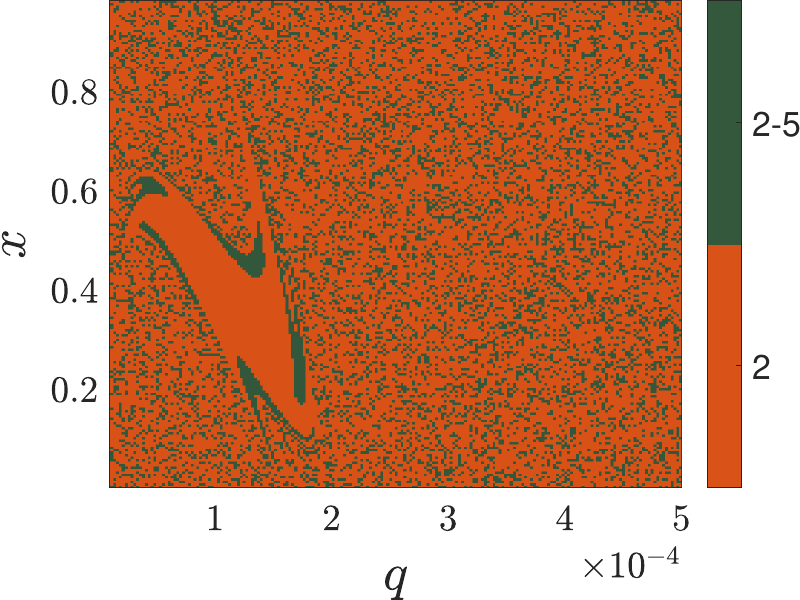}
  }
  \hfill
  \subfigure[]{
    \includegraphics[width=0.4\linewidth]{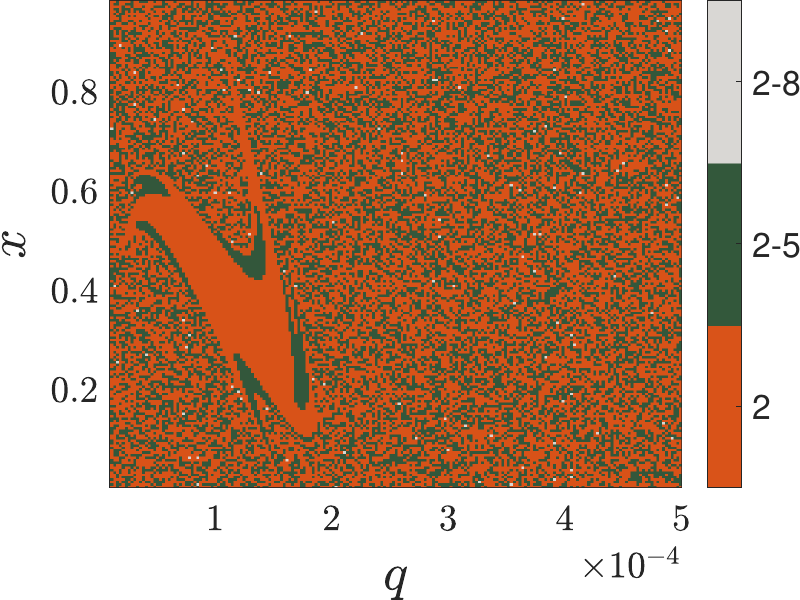}
  }
   \hfill
  \subfigure[]{
    \includegraphics[width=0.4\linewidth]{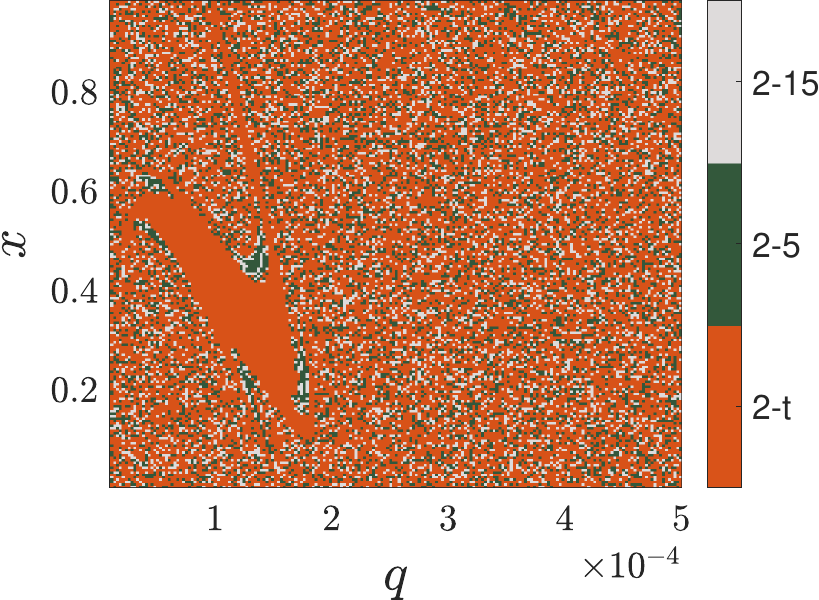}
  }
   \hfill
  \subfigure[]{
    \includegraphics[width=0.4\linewidth]{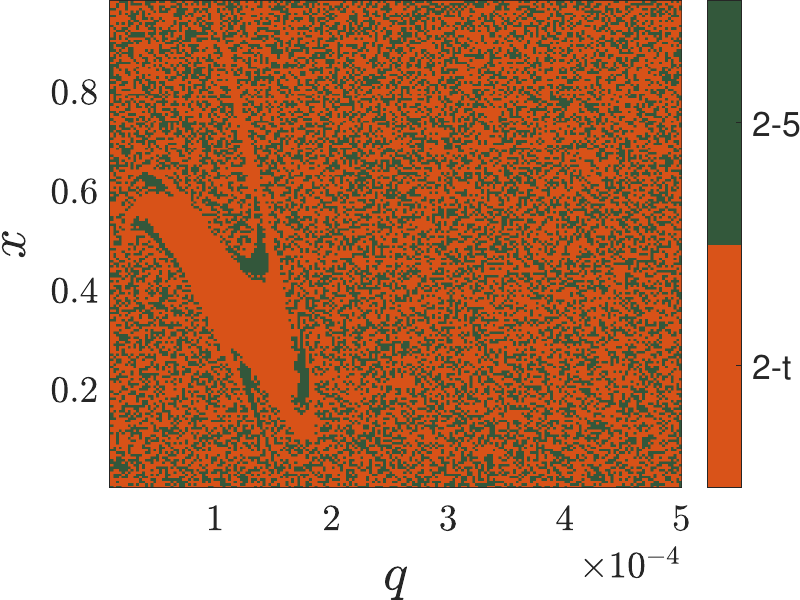}
  }
  \caption{\small Phase space (a)-(d) and associated basins of attraction (e)-(h) with $k=0.865$ fixed, corresponding to Fig. \ref{bifa1}(b). (a) and (e): $a=19,$ the period-five attractor of the II--$5$ family and the equilibrium coexist. (b) and (f): $a=19.0725,$ the period-five attractor of the II--$5$ family, the period-eight attractor of the II--$8$ family, and the equilibrium coexist. (c) and (g): $a=19.221$ the period-five attractor of the II--$5$ family, the period-fifteen attractor of the II--$15$ family, and the equilibrium coexist. (d) and (h): $a=19.2282,$ the period-five attractor of the II--$5$ family and the equilibrium coexist. In all cases, note the intermingling of the basins.
}
\label{basins2.5}
%\vspace*{12pt}
\end{figure}

To illustrate the bifurcations involving each of the II-$i$ cycles described above, we consider, indicatively, the five-period cycle, illustrated in Fig. \ref{5phspace}(a), created upon crossing SN$_{2-5}$, which is also the first resonance attractor in Fig. \ref{bifa1}. Referring to Fig. \ref{bifs1}, this cycle exists in the region defined by the two branches of the period-five Arnold tongue, the curves PD$_{2-5}$ and NS$_{2-5}$ connecting the two codimension-two resonance points. When the parameters cross SN$_{2-5}$ comprising the tongue, the cycle collides with a saddle orbit and disappears, while when the parameters cross NS$_{2-5}$, a saddle invariant curve for the fifth iterate of the map is created. The crossing of PD$_{2-5}$ results in a period-ten cycle shown in Fig. \ref{5phspace}(b). The transverse homoclinic tangle implies the existence of horseshoe-like dynamics near the R$_1$ point, but the chaotic invariant set is non-attracting due to the sign of the normal-form coefficient. From the R$_2$ resonance point lying on PD$_{2-5}$, apart from NS$_{2-5}$ connecting it with the R$_1$ point, another NS curve emanates (not shown in Fig. \ref{bifs1}), which upon crossing results in the birth of a stable invariant curve for the tenth iterate of the map, illustrated in Fig. \ref{5phspace}(c). The Feigenbaum cascade of the period-five cycle leads to the formation of the chaotic attractor (cf. also the positivity of the Lyapunov exponent as illustrated in Fig. \ref{bifa1}), shown in Fig. \ref{5phspace}(d). We close by noting that the dynamical phenomena described above are shared by all the attractors of II--$i$ type, as also indicated by inspecting the bifurcation curves in Fig. \ref{bifs1}. 

The Neimark--Sacker curve NS$_2^1$ of the fixed point has itself a codimension-two R$_2$ resonance point, which, as previously discussed, also acts as an organizing center. Beyond this point, NS$_2^1$ essentially no longer corresponds to a Neimark--Sacker bifurcation, but rather to a neutral saddle. As shown in Fig. \ref{bifs1}, a PD curve--denoted by PD$_{2}$--crosses NS$_2^1$ at the R$_2$ point; crossing it from below R$_2$ results in the emergence of a stable period-two orbit. However, crossing it from above results in the creation of a saddle period-two orbit, which gains stability after crossing another NS curve (NS$_2^2$) emanating from the R$_2$ point and connecting it with an R$_1$ point on the PD curve of the period-two cycle (cf. Fig. \ref{bifs2}). At this curve, a period-doubled cycle that was born upon crossing PD$_{2}$ below R$_2$ is correspondingly destroyed. The NS$_2^2$ curve in question will also be discussed in the next subsection, in the examination of the second region of the parameter space, since for the parameter value $k=0.865$ considered here the stable period-two orbit is created in the second parameter region.

Returning to Fig. \ref{bifa1} (b), the emerging attractors can be grouped into the following types: II (equilibrium), I/IV/V and II--$i$. The first resonance attractor is the five-period cycle discussed above. The phase space, in the case where only this cycle and the equilibrium coexist as attractors, is shown in Fig. \ref{basins2.5}(a) and the associated basins of attraction in Fig. \ref{basins2.5} (e). As is qualitatively apparent, the entropies in Table \ref{tblentropy} have increased, confirming the more complex nature of the basins of attraction. Its corresponding tongue overlaps with, and in fact contains, several of the remaining tongues, suggesting that multistability, along with other resonance periodic orbits, will be observed. Several examples are shown in Fig. \ref{bifa1}(b): as $a$ varies, the parameters first enter the region of the eight-Arnold tongue, shown in Fig. \ref{bifs1}, giving rise to a period-eight cycle which, following the generic pattern discussed previously, ultimately leads to the formation of a chaotic saddle. The phase space in which the II--$8$ family period-eight cycle coexists with the II--$5$ family period-five cycle and the fixed point is shown in Fig. \ref{basins2.5}(b). It is worth noting that the basin of attraction of the eight-cycle, as shown in Fig. \ref{basins2.5}(f) is essentially intermingled within the basins of the other attractors. Similar considerations apply to the other resonance higher-order periodic orbits shown in Fig. \ref{bifa1}. When the parameters ultimately cross NS$_{2}^1$, the equilibrium loses stability and a stable invariant curve emerges, coexisting with the II--$i$ attractors. Note that because the parameter values are far from the points from which the Arnold tongues emanate, the corresponding periodic orbits are no longer necessarily close to the invariant curve in a geometric sense, as shown for example in Fig. \ref{basins2.5} (g) and (h). The basins of attraction of the invariant curve coexisting with the II--5 and II--15 attractors, and with the II--5 attractor after the destruction of the II--15 one, are shown in Fig. \ref{basins2}(e) and (f). It is worth noting that the basin and boundary basin entropies attain high values in the first case and decrease in the latter case, while still remaining above the fractality threshold--cf. Table \ref{tblentropy}.

Concluding the study of the parameter region of Fig.~\ref{bifs1}, we note that for large values of $k$ no periodic attractors exist, as all have been destroyed by the upper branch of the corresponding Arnold tongue. However, the chaotic saddle, which as marked in Fig.~\ref{basins2}(c) causes transient chaos, becomes attracting when the basin boundary is crossed in parameter space, namely through a reverse boundary crisis. The chaotic attractor remains the sole attractor of the system for large $k$. Fig.~\ref{bifa1}(c) shows the bifurcations involving the various periodic attractors, the invariant curve, as well as the emergence of the chaotic attractor, as $a=19.3$ is fixed and $k$ increases. We note that the invariant curve undergoes a sequence of global bifurcations associated with its collision with the stable manifold of a resonant saddle cycle, its reappearance through the reverse process, and its final destruction. Fig. \ref{basins3} shows the phase space and the basins of attraction of the coexisting attractors as the bifurcations of the invariant curve unfold. It is worth noting that the II--$7$ and II--$12$ periodic attractors appear very close to the invariant curve as a result of their proximity to the corresponding Arnold tongues; cf., in contrast, the II--$5$ cycle. Moreover, the entropies in Table \ref{tblentropy} attain high values, especially in the case of three coexisting attractors shown in Fig. \ref{basins3}(c), confirming the qualitative complexity of the basins of attraction shown in Fig. \ref{basins3}. The same situation will also arise in the following discussion, when the second region of the bifurcation diagram is examined.

\subsection{Dynamics within the second region}\label{sec3.2}
\begin{figure}
    \centering
    \includegraphics[width=4.5in]{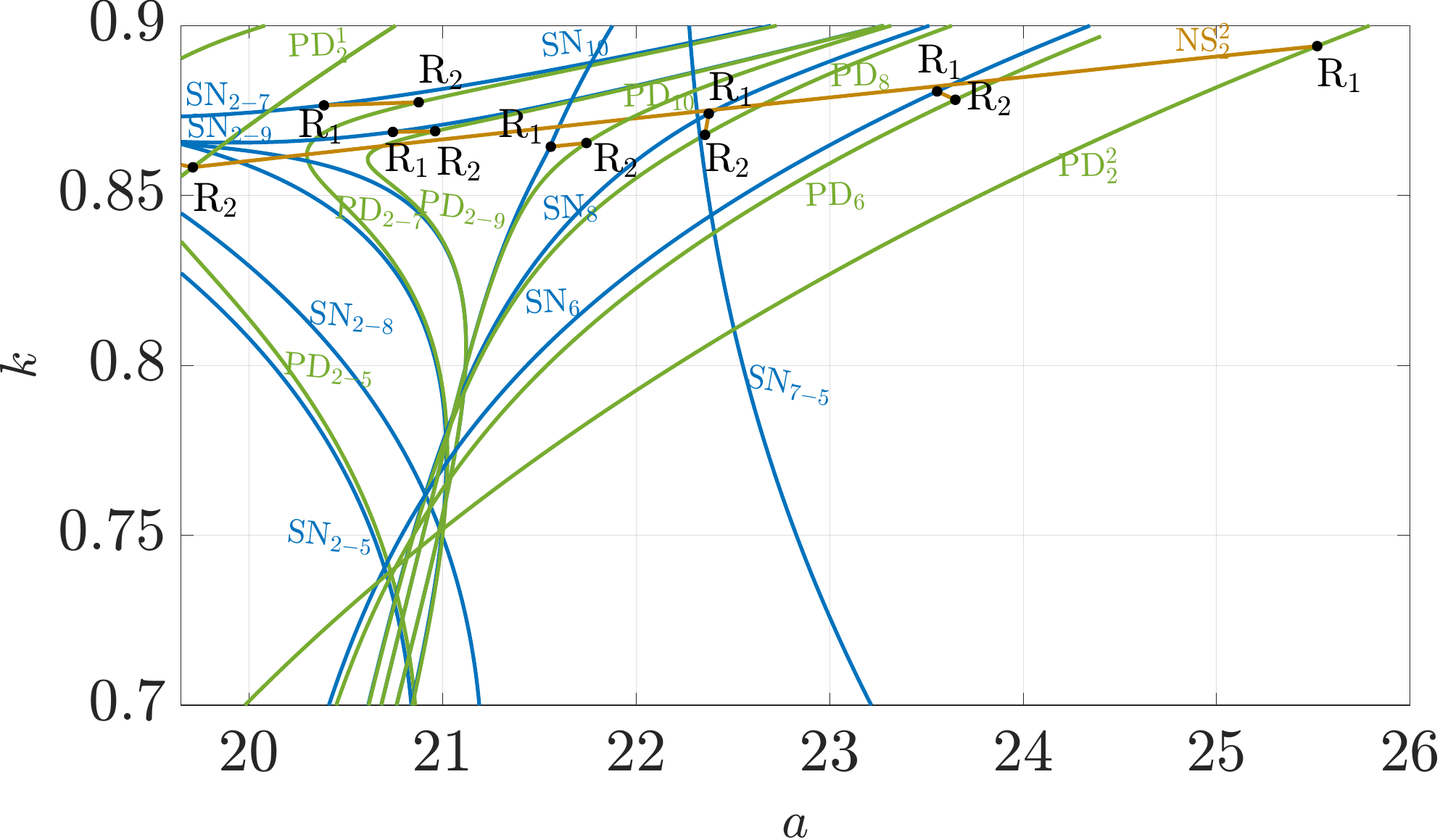}
    \caption{\small Bifurcation curves in the second region. Saddle-node bifurcations are denoted by SN, period-doubling bifurcations by PD, and Neimark--Sacker bifurcations by NS. Bifurcation curves of the same type share the same color. The index of each curve refers to the attractor family of Table~\ref{tbl1}. For clarity, the SN$_{6-j}$ curves are denoted by SN$_j,$ while, except for NS$_2^2$, the labels of the NS curves connecting R$_1$ and R$_2$ points have been omitted. The superscript in the PD and NS curves refers to their order in the respective family. The labels R$_1$ correspond to $1:1$ strong resonances, while the labels R$_2$ to $1:2$ strong resonances.}
    \label{bifs2}
\end{figure}

We continue with the study of the second region of the parameter space, shown in Fig.~\ref{bifs2}. To illustrate the emerging attractors as the various bifurcations unfold, we fix $k=0.865$ and vary the learning rate, i.e., $a$, focusing first on Fig. \ref{bifa2}(a). As in the last region of Fig.~\ref{bifa1}(b), the invariant curve and the five-period orbit coexist for smaller values of $a$, with the latter undergoing a period-doubling route to chaos as each of the PD curves is crossed, the first of which is shown in Fig.~\ref{bifs2}. Of course, within very narrow regions of the parameter space, we observe periodic attractors of the I/IV/V type, originating from the presence of the chaotic saddle. As the parameters enter the period-seven and period-nine Arnold tongues, shown in Fig.~\ref{bifs2}, the corresponding periodic cycles appear. Note at this point that the tongues have the qualitative structure discussed previously: the upper branch exhibits a codimension-two R$_1$ resonance from which an NS curve emanates (not labeled in Fig. \ref{bifs2} for the sake of clarity), connecting it with a codimension-two R$_2$ resonance point located on the PD curve of the corresponding cycle of the tongue. The dynamics induced by this structure is qualitatively similar to that discussed for the five-period cycle.

The invariant curve, born at NS$_2^1$ follows the familiar bifurcation pattern of disappearing and reappearing already encountered in the examination of the first region. For higher values of $a$, we observe the coexistence of several different periodic orbits with overlapping resonant tongues (not shown for clarity in Fig.~\ref{bifs2}). An indicative example of overlap between the seven- and nine-tongues (shown in Fig.~\ref{bifs2}) and those of periods 16 and 25 is presented in Fig.~\ref{coex_16_25}, which illustrates in (a) the phase space and in (b) the complexity of the basins of attraction. Apart from the obvious high values of the entropies, it is worth noting that in this case $S_b = S_{bb}$, meaning that every point in the basin is actually a boundary point. This represents an interesting case of complete mixing with $S_b = S_{bb}$, but characterized by non-intermingled boundaries \cite{entropy3}, as the entropies are not equal to the Wada index \cite{entropy3}--which is also qualitatively apparent from the figure. 

\begin{figure}[t!]
  \centering
 % \subfigure[]{
 %   \includegraphics[width=0.41\linewidth]{figures/bif_a_19.3_20.3-eps-converted-to.pdf}
%  }
%  \subfigure[]{
%    \includegraphics[width=0.41\linewidth]{figures/bif_a_20_26-eps-converted-to.pdf}
%  }
    \subfigure[]{
    \includegraphics[width=0.483\linewidth]{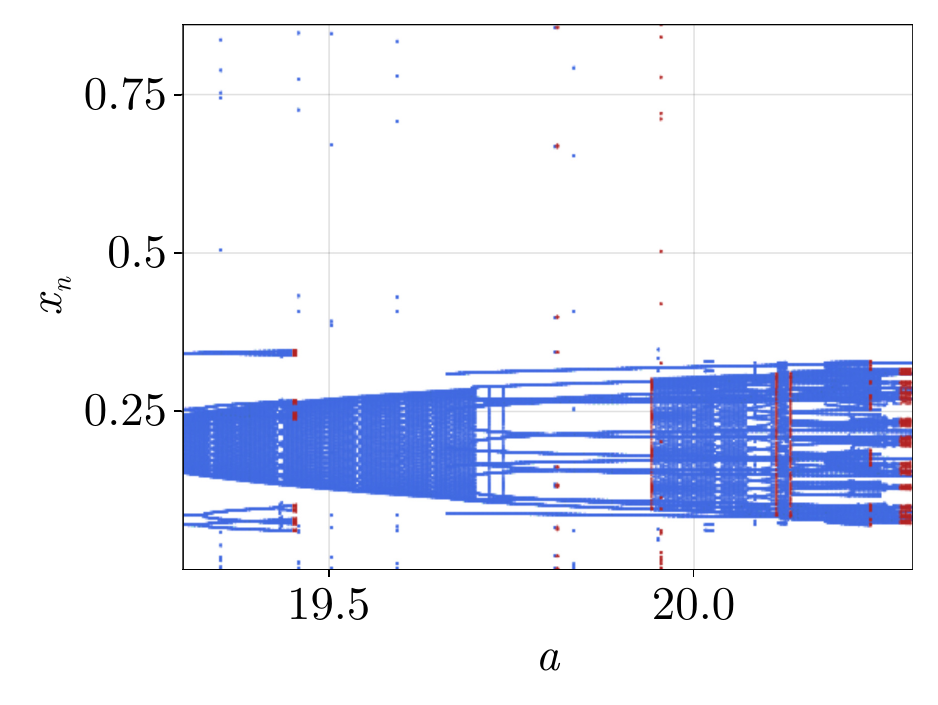}
  }
  \hfill
  \subfigure[]{
    \includegraphics[width=0.483\linewidth]{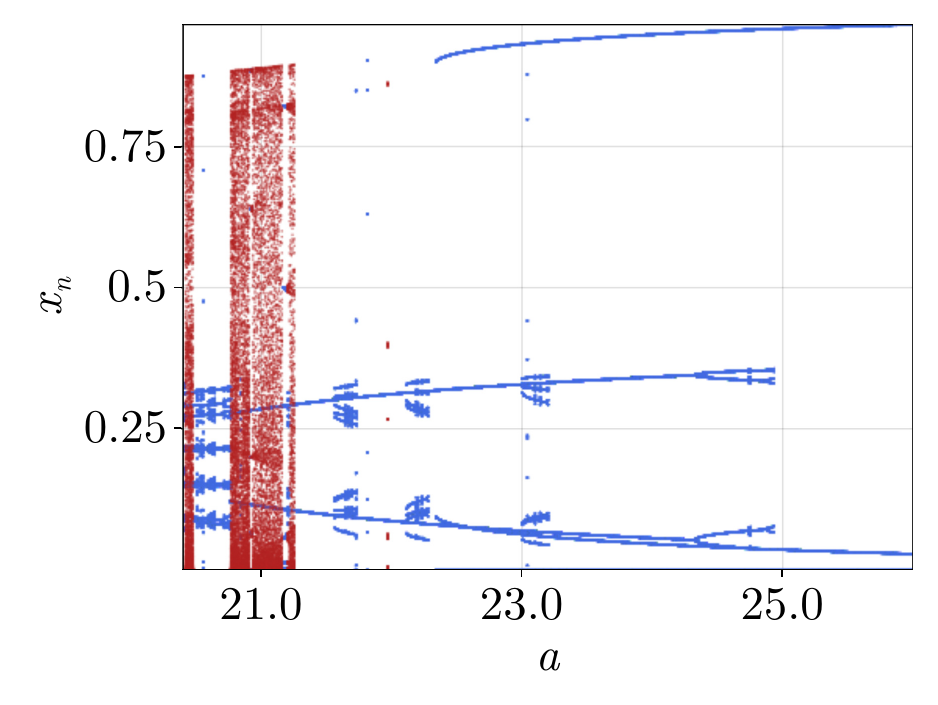}
  }
  \caption{\small The various emerging attractors as $k=0.865$ is held fixed and $a$ is gradually increased. Each attractor's birth or destruction corresponds to the crossing of a bifurcation curve of Fig.~\ref{bifs2} or to a crisis. Orbits colored in blue correspond to regular behavior, characterized by a maximum Lyapunov exponent $\lambda<0$. Orbits colored in red correspond to chaotic behavior, with $\lambda\geq 10^{-2}$.}
\label{bifa2}
%\vspace*{12pt}
\end{figure}

The presence of both the I/IV/V and II--$i$ chaotic saddles results in large transiently chaotic intervals; Fig.~\ref{coex_16_25}(c) shows an orbit spending a long time near the regions of the I/IV/V and II--$25$ chaotic saddles before finally converging to the seven-cycle. As expected, the seven-cycle undergoes a period-doubling route to chaos (as can also be observed from the Lyapunov exponent in Fig. \ref{bifa2}); the first PD curve, PD$_{2-7}$, is shown in Fig.~\ref{bifs2}, following the generic pattern of the II--$i$ family of attractors. The time spent near the saddles gradually increases, and eventually the I/IV/V chaotic saddle gains stability via a reverse boundary crisis bifurcation--cf. Fig.~\ref{bifa2}(b). The phase space and the basins of attraction when the chaotic attractor coexists with the nine-period cycle are shown in Fig.~\ref{basins4}(a) and (e), respectively--note that the basin of the nine-cycle can be loosely thought of as being organized around the (unstable) invariant curve and that $S_b$ attains a low value, while $S_{bb}<0.439$. Moreover, geometrically the nine-cycle captures roughly the region of the invariant curve. As illustrated in Fig.~\ref{bifa2}(b), further increasing $a$ results in the destruction of the I/IV/V attractor through a boundary crisis, as well as in the emergence and disappearance of other high-period orbits as the corresponding tongues are crossed. After the destruction of the II--$9$ attractor, the chaotic attractor (cf. also the Lyapunov exponent in Fig. \ref{bifa2}) reemerges and coexists with a period-two cycle.

\begin{figure}[t!]
  \centering
  \subfigure[]{
    \includegraphics[width=0.41\linewidth]{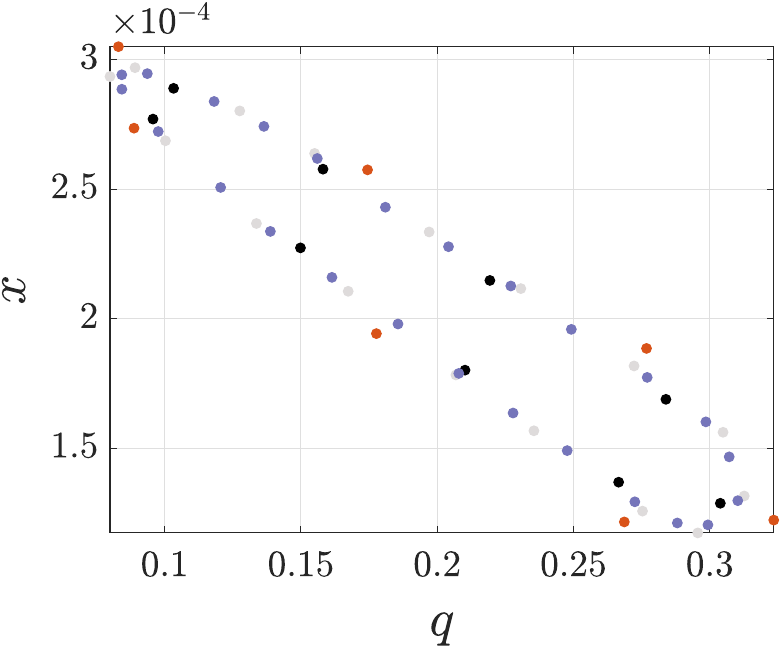}
  }
  \hfill
  \subfigure[]{
    \includegraphics[width=0.43\linewidth]{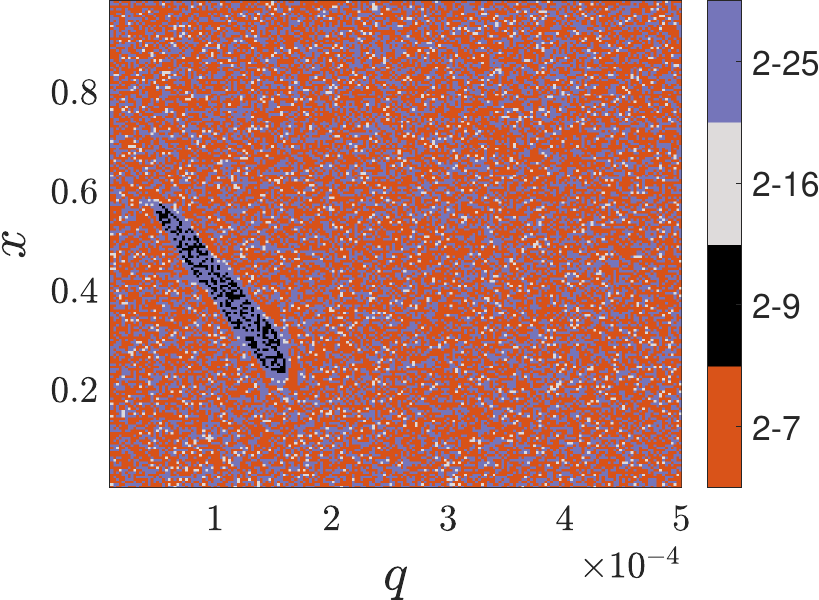}
  }
  
  \subfigure[]{
    \includegraphics[width=0.41\linewidth]{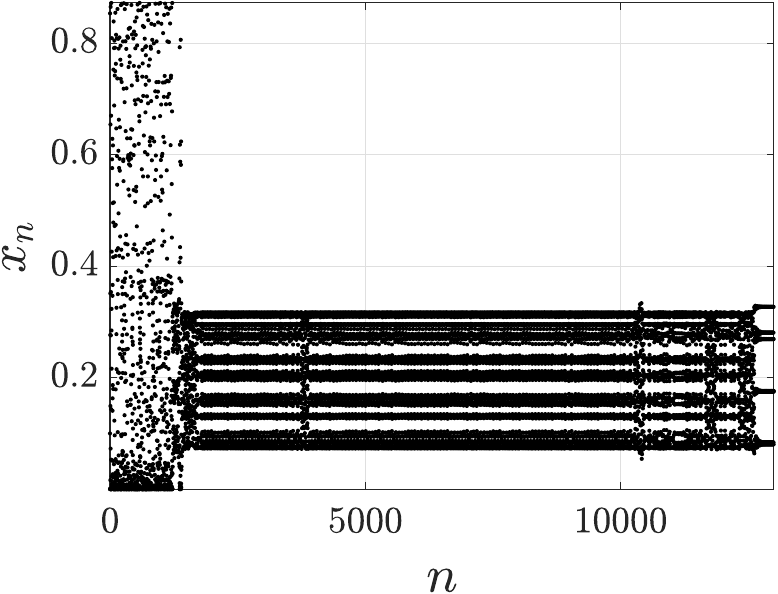}
  }
  \caption{\small Coexisting II--7, II--9, II--16, and II--25 attractors in the phase space (a) and their basins of attraction in (b), $a=20.14,\; k=0.865$. Note the intermingling of the basins. (c) Transient chaos due to the I/IV/V and II--$i$ chaotic saddles of an orbit converging to the period-seven attractor.}
\label{coex_16_25}
%\vspace*{12pt}
\end{figure}

Regarding this two-cycle, we note that as also discussed previously, NS$_{2}^1$ exhibits a codimension-two R$_2$ resonance point through which PD$_2^1$ passes. Since, for the value of $k$ under consideration, the PD curve is crossed from above, the initially emerging period-two orbit is of saddle type. It gains stability, corresponding to the attractor of Fig. \ref{bifa2}(b), upon crossing NS$_2^2$ emanating from the R$_2$ point and connecting it with an R$_1$ resonance point on the PD curve of the corresponding period-two orbit, denoted by PD$_2^2$. Since it originates from the fixed point and to highlight its difference from the II--$i$ families, we continue to denote it as family II. Its basin of attraction, shown in Fig. \ref{basins4}(e), is organized around the two points of the cycle (cf. the phase space in Fig. \ref{basins4} (b)) and, as is expected, the associated entropies attain significantly lower values compared to the cases of intermingled basins discussed earlier. We note that although the structure organized around the R$_2$ point appears similar to that examined in the case of the Arnold tongues, the normal-form coefficient in this case is positive; thus, the two-invariant curve that is destroyed when the period-two orbit gains stability is not stable.

\begin{figure}[!p]
  \centering
\subfigure[]{
    \includegraphics[width=0.38\linewidth]{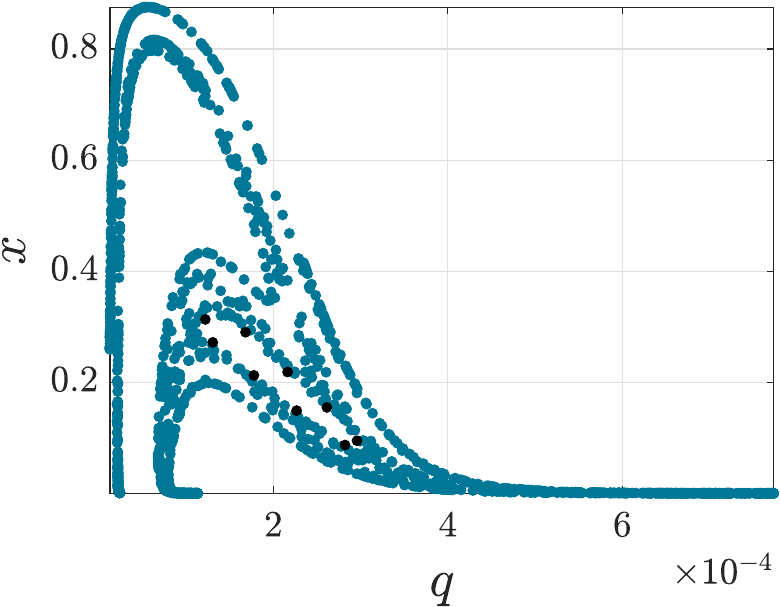}
  }
  \hfill
  \subfigure[]{
    \includegraphics[width=0.38\linewidth]{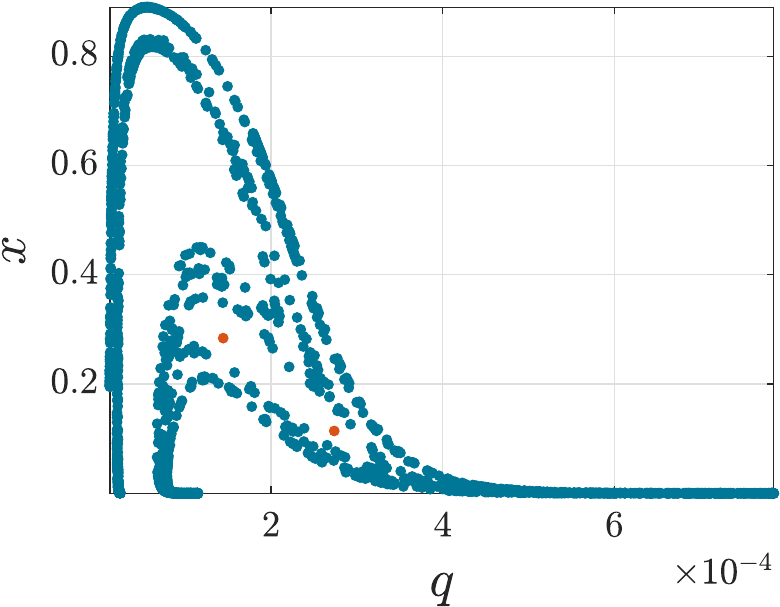}
  }
\hfill
  \subfigure[]{
    \includegraphics[width=0.38\linewidth]{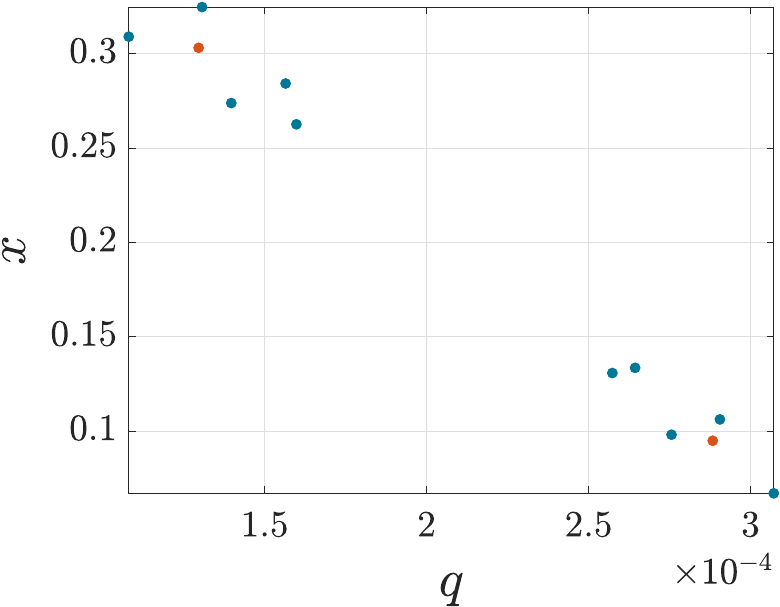}
  }
  \hfill
 \subfigure[]{
    \includegraphics[width=0.38\linewidth]{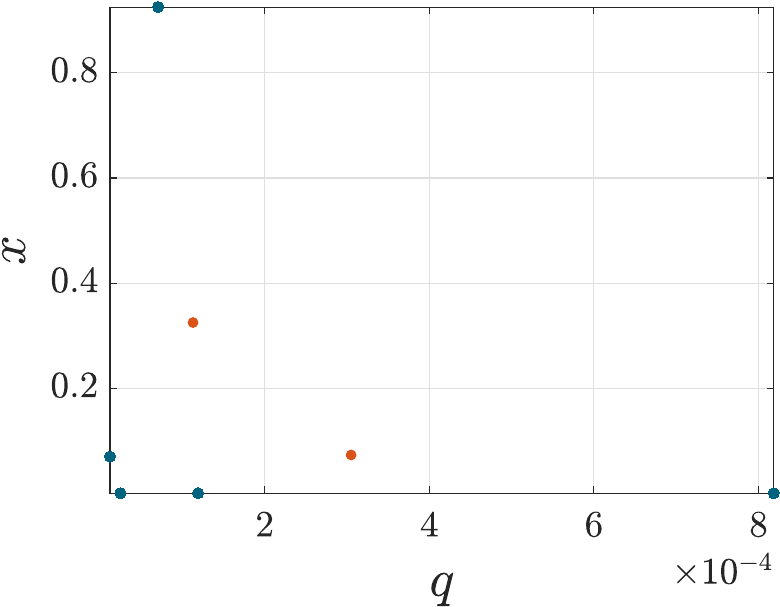}
  }
  \hfill
  \subfigure[]{
    \includegraphics[width=0.4\linewidth]{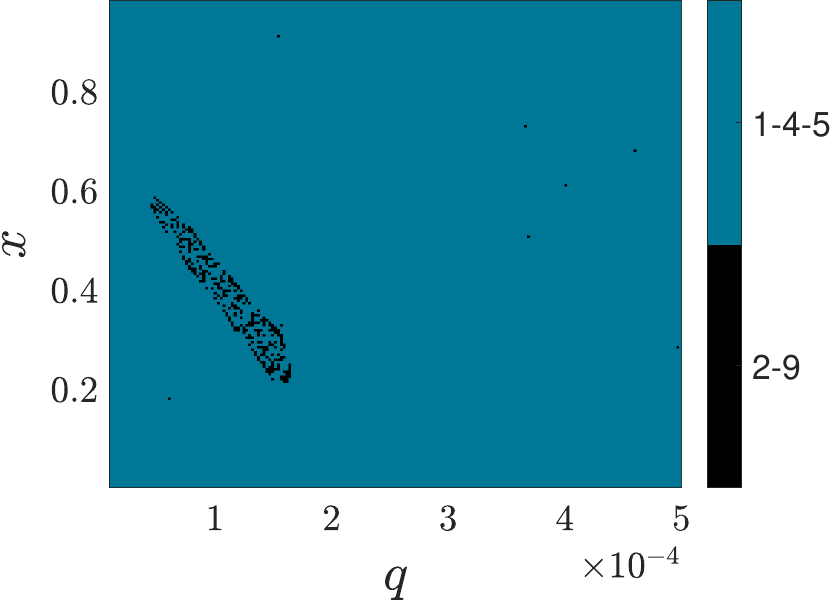}
  }
  \hfill
  \subfigure[]{
    \includegraphics[width=0.4\linewidth]{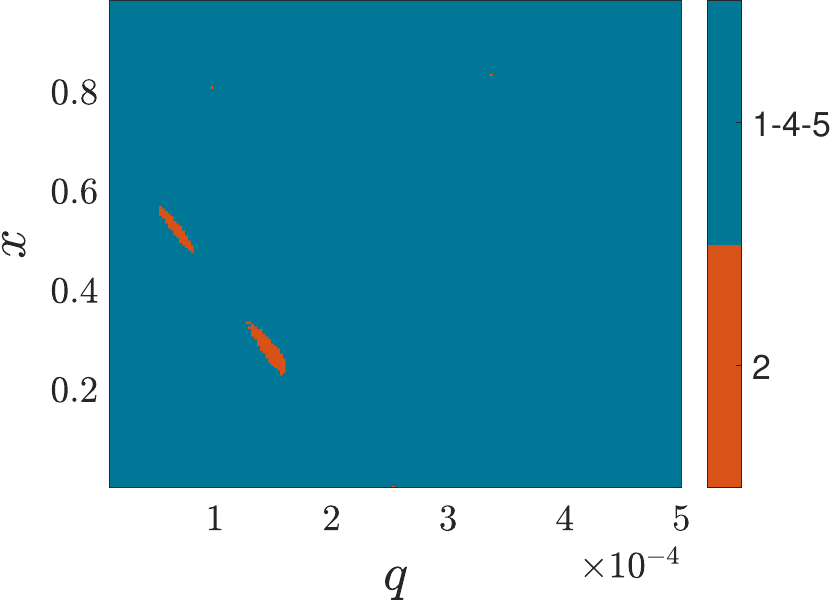}
  }
  \subfigure[]{
    \includegraphics[width=0.4\linewidth]{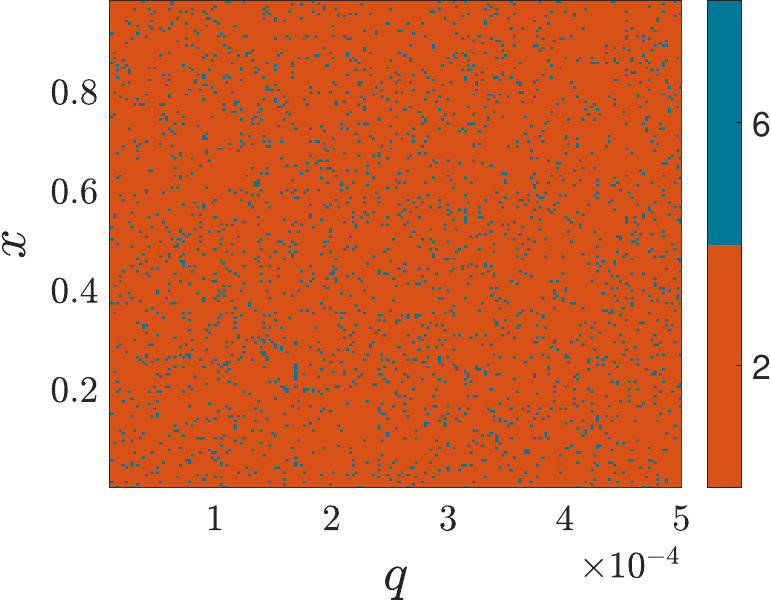}
  }
  \hfill
 \subfigure[]{
    \includegraphics[width=0.4\linewidth]{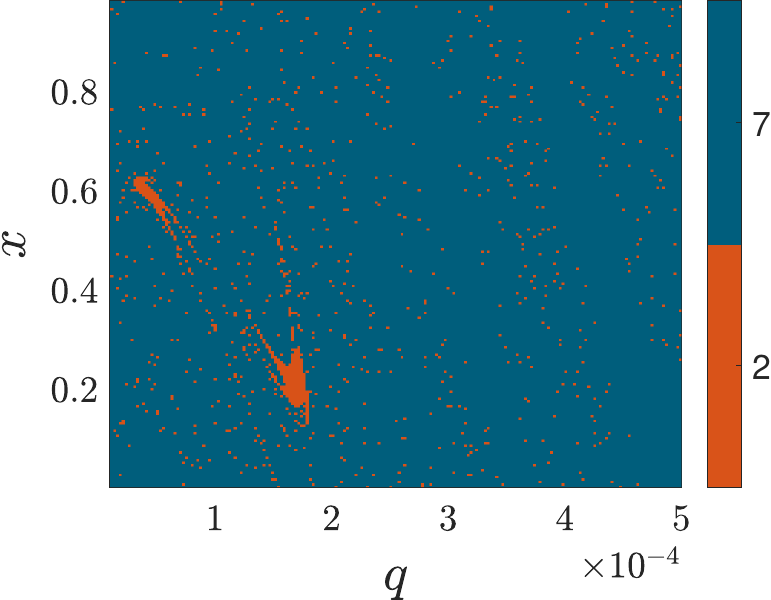}
  }
  \caption{\small Phase space (a)-(d) and associated basins of attraction (e)-(h) with $k=0.865$ fixed, corresponding to Fig.~\ref{bifa2} (a) and (b). (a) and (e): $a=20.43$, the period-nine attractor of the II--$9$ family and the chaotic attractor of the I/IV/V family coexist. (b) and (f): $a=20.97$, the period-two attractor of the II family and the chaotic attractor of the I/IV/V family coexist. (c) and (g): $a=21.64$, the period-two attractor of the II family and period-ten attractor of the VI family coexist. (d) and (h): $a=22.75$,  the period-two attractor of the II family and period-five attractor of the VII family coexist.}
\label{basins4}
%\vspace*{12pt}
\end{figure}

As the parameters vary further (cf. Fig. \ref{bifa2}), the chaotic attractor transitions to periodic behavior through a boundary crisis, with the periodic orbit undergoing a Feigenbaum route to chaos. The original chaotic attractor (cf. the Lyapunov exponent in Fig. \ref{bifa2}) reemerges through a reverse boundary crisis and is ultimately destroyed through another boundary crisis, leaving the period-two orbit as the sole attractor. Increasing $a$, a stable ten-cycle is born at the SN$_{10}$ curve of Fig.~\ref{bifs2}, as shown in Fig.~\ref{bifa2}(b) and in Fig. \ref{basins4}(c). Its basin of attraction is intermingled with that of the period-two cycle--cf. Fig.~\ref{basins4}(g) and Table \ref{tblentropy}, yet $S_b$ is small. The origin of this periodic attractor, however, is not the phase-locking Arnold tongues of the invariant curve encountered so far; for this reason, we denote its family by VI in Table~\ref{tbl1}. For clarity, we use as index in the associated bifurcations of the family in Fig. \ref{bifs2} the period of the cycle instead of $6-10$. The SN curve exhibits a codimension-two R$_1$ point, following a similar generic structure to that of the primary PD$_2^1$ curve. As the parameters cross PD$_{10}$, the orbit undergoes period-doubling, while the resulting twenty-cycle is destroyed before the VI family transitions to chaos--cf. the Lyapunov exponent in Fig. \ref{bifa2}. Attractors of the I/IV/V family also exist within narrow parameter regions, as is clear from Fig.~\ref{bifa2}(b).

The eight- and six-period cycles, which appear when their corresponding SN$_8$ and SN$_6$ curves in Fig.~\ref{bifs2} are crossed, as well as other high-period cycles that exist within very narrow parameter regions and are shown in Fig.~\ref{bifa2}(b), follow qualitatively similar behavior to that of the ten-period cycle discussed above. For this reason, with slight abuse of notation, we also denote the families of these attractors again by VI and use as index their corresponding period. However, the period-five cycle born at SN$_{7-5}$ (cf. Fig.~\ref{bifs2}) and shown in Fig.~\ref{bifa2}(b), belonging to family VII of Table~\ref{tbl1}, does not undergo any further bifurcation; accordingly, its SN curve exhibits no codimension-two resonances. In fact, it remains the sole attractor of the system for a sufficiently large learning rate $a$, after the period-two orbit period-doubles twice and is ultimately destroyed in a boundary crisis (before transitioning to chaos, similarly to the VI families). The phase space and the basins of attraction when the period-two II and period-five VII attractors coexist are shown in Fig.~\ref{basins4}(d) and Fig.~\ref{basins4}(h). It is worth noting that apart from a decrease in the $S_{bb}$, suggesting that the basin boundary no longer possesses a fractal nature, the five-period cycle captures most of the phase space, which, prior to its emergence, belonged to the basin of the two-period cycle. After the destruction of the family II attractor, this period-five cycle remains the sole attractor of the system for large $a.$

Summarizing the dynamics analyzed in the two regions shown in Figs.~\ref{bifs1} and \ref{bifs2}, we emphasize that, despite the complexity of the bifurcations occurring in the system, the attractors can be grouped and systematically classified into seven families, as reported in Table~\ref{tbl1}. A high-level overview of the key aspects of the dynamics is as follows: the I, III, IV and V families originate from period-three, period-nine, period-four, and period-seven cycles, respectively, which undergo period-doubling cascades to chaos (cf. the relevant Lyapunov exponent), transition to chaotic saddles after boundary crises, and eventually merge, forming a merged chaotic saddle. Within narrow parameter regions, periodic orbits located near the chaotic saddle gain stability and are subsequently destroyed in boundary crises. Transiently chaotic behavior is also observed and traces back to the chaotic saddle.

The fixed point (family II) is stable, coexisting with the attractors of the I, III, IV, V, or the merged I/IV/V family--and we emphasize this in contrast to the traditional MWU algorithm, in which, when the fixed point is stable, it is also the unique global attractor of the system--until it loses stability via an NS or a PD bifurcation, cf. the curves NS$_2^1$ and PD$_2$. The two curves meet at an R$_2$ codimension-two resonance point. The NS curve, which results in a stable invariant curve when crossed above the R$_2$ point, induces overlapping phase-locking Arnold tongues, leading to the birth and coexistence of periodic orbits belonging to the II--$i$ family. Each of these tongues is characterized by qualitatively similar behavior in terms of codimension-two resonances (R$_1$ on the upper branch of the tongue and R$_2$ on the first PD curve) and in terms of the underlying bifurcations of each II--$i$ attractor, involving Feigenbaum routes to chaos and destruction via boundary crises. The basin boundaries of the coexisting attractors are frequently fractal and highly intermingled.

As the various II--$i$ attractors appear and disappear and the invariant curve is destroyed, the I/IV/V chaotic saddle regains stability within certain parameter regions. The codimension-two points of the period-two orbit, although structurally similar to those of the tongues, are characterized by an opposite normal-form coefficient at the R$_2$ point; thus, when the primary NS curve, i.e., NS$_2^1$, is crossed above the R$_2$ point, the period-two cycle is born as a saddle at PD$_2$ and gains stability when the secondary NS curve is crossed. The periodic orbits of the VI family, which do not originate from phase-locking tongues and appear as their SN curves are crossed, exhibit a similar pattern concerning the codimension-two points as described above. They are destroyed in a boundary crisis before transitioning to chaos (cf. the associated Lyapunov exponent). Similarly, the period-two cycle undergoes two period-doublings and is destroyed in a boundary crisis. 

The system is thus left with a single, global, period-five attractor of the VII family, initially coexisting with the II and VI attractor families. This attractor undergoes no further bifurcations, and beyond the destruction of the period-two cycle, no attractors belonging either to the existing families or to new families are born. In fact, the behavior of the system for sufficiently large learning rate is predictable, and we highlight this property as being in stark contrast to the standard MWU, in the absence of extreme events.

We close this section by noting that the dynamical phenomena discussed above can also be viewed in analogy with the concept of chaotic itinerancy \cite{kaneko,nicolis,tsuda,basios}. Chaotic itinerancy refers to the phenomenon in which trajectories spend long intervals near quasi-stable states, often referred to as attractor ruins, and move between them through chaotic transitions. This phenomenon has been associated with the creation and destruction of coexisting attractors, changes in basin boundaries as parameters vary, and fractal basin boundaries \cite{kaneko,nicolis,tsuda,basios}. In our context, the different attractor families discussed do not represent perceptual categories, as in the cognitive literature, but can instead be viewed as competing collective learning outcomes. As the parameters vary, different attractor families emerge and are destroyed, while chaotic saddles are created following crisis bifurcations. These saddles play a role analogous to attractor ruins by generating long transient intervals during which trajectories remain close to dynamical states associated with previously attracting regimes. Finally, the complexity of the basin boundaries, quantified by $S_b$ and $S_{bb}$, implies that small variations in the initial conditions may select qualitatively different asymptotic learning outcomes. In this sense, the observed dynamics may be regarded as a transient, low-dimensional analogue of chaotic itinerancy.

\subsection{Dynamics within the third region}\label{sec3.3}
\begin{figure}[t!]
  \centering
  \subfigure[]{
    \includegraphics[width=0.49\linewidth]{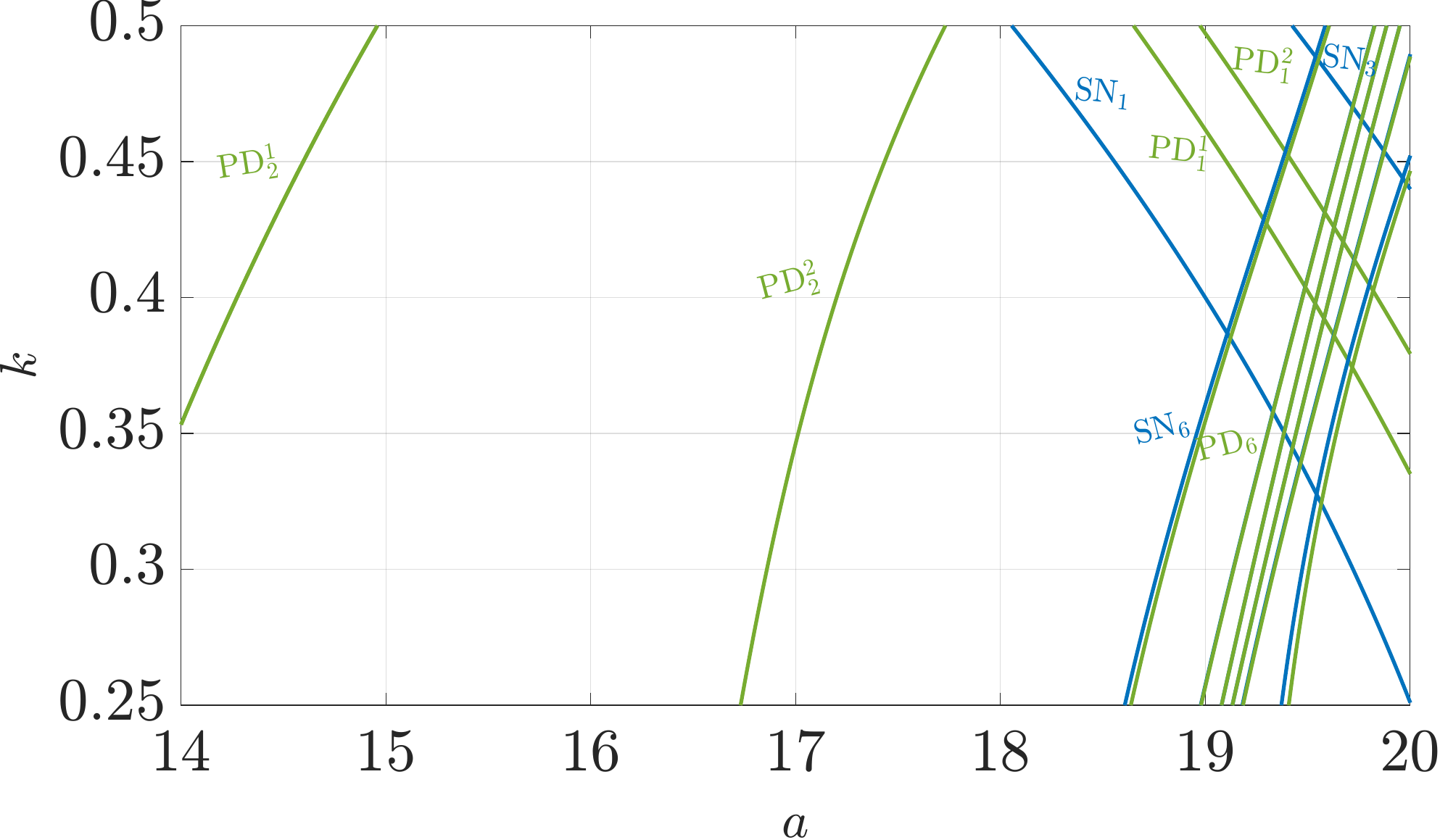}}
  \hfill
  \subfigure[]{
  \includegraphics[width=0.49\linewidth]{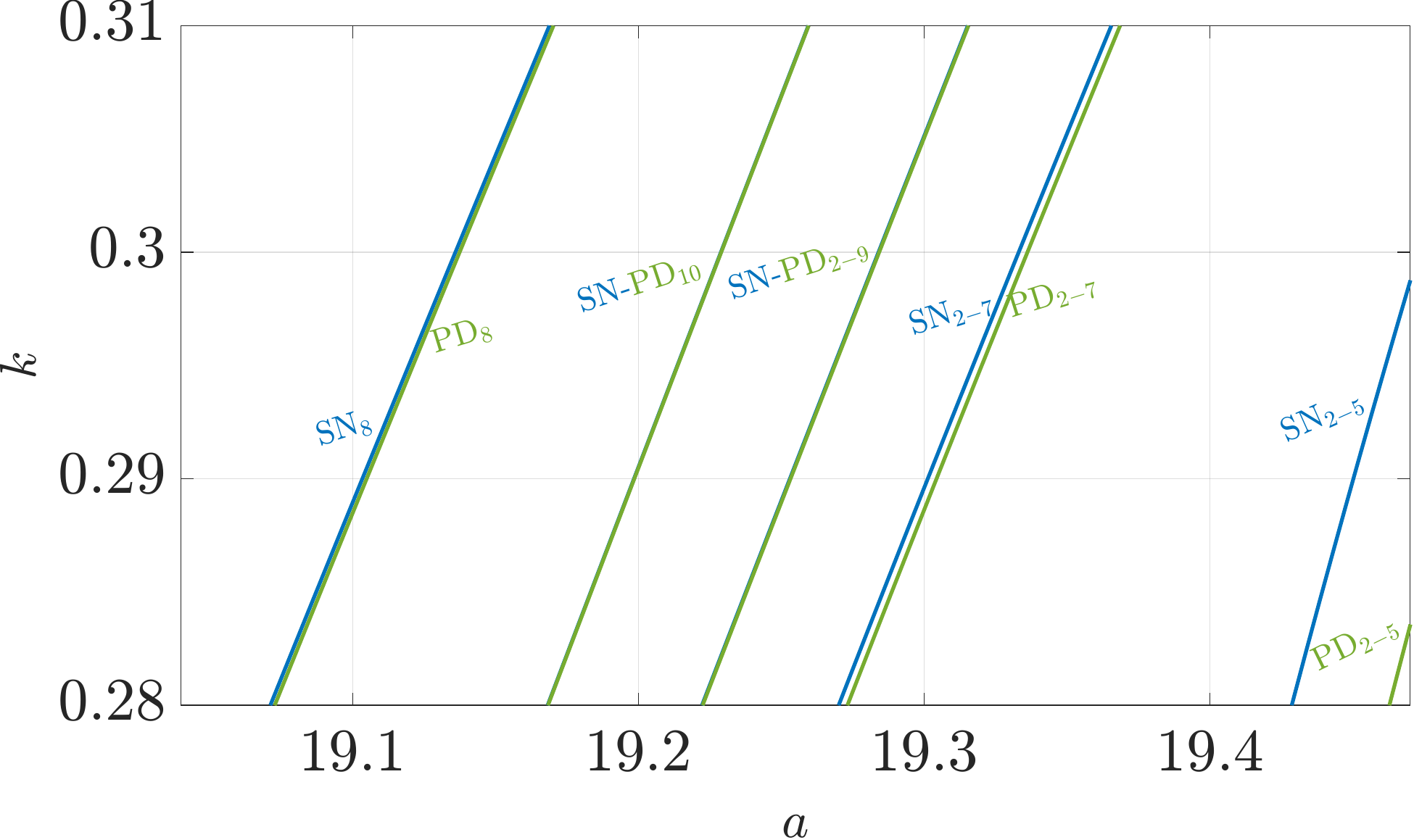}}
  %\hfill
  \subfigure[]{
  \includegraphics[width=0.49\linewidth]{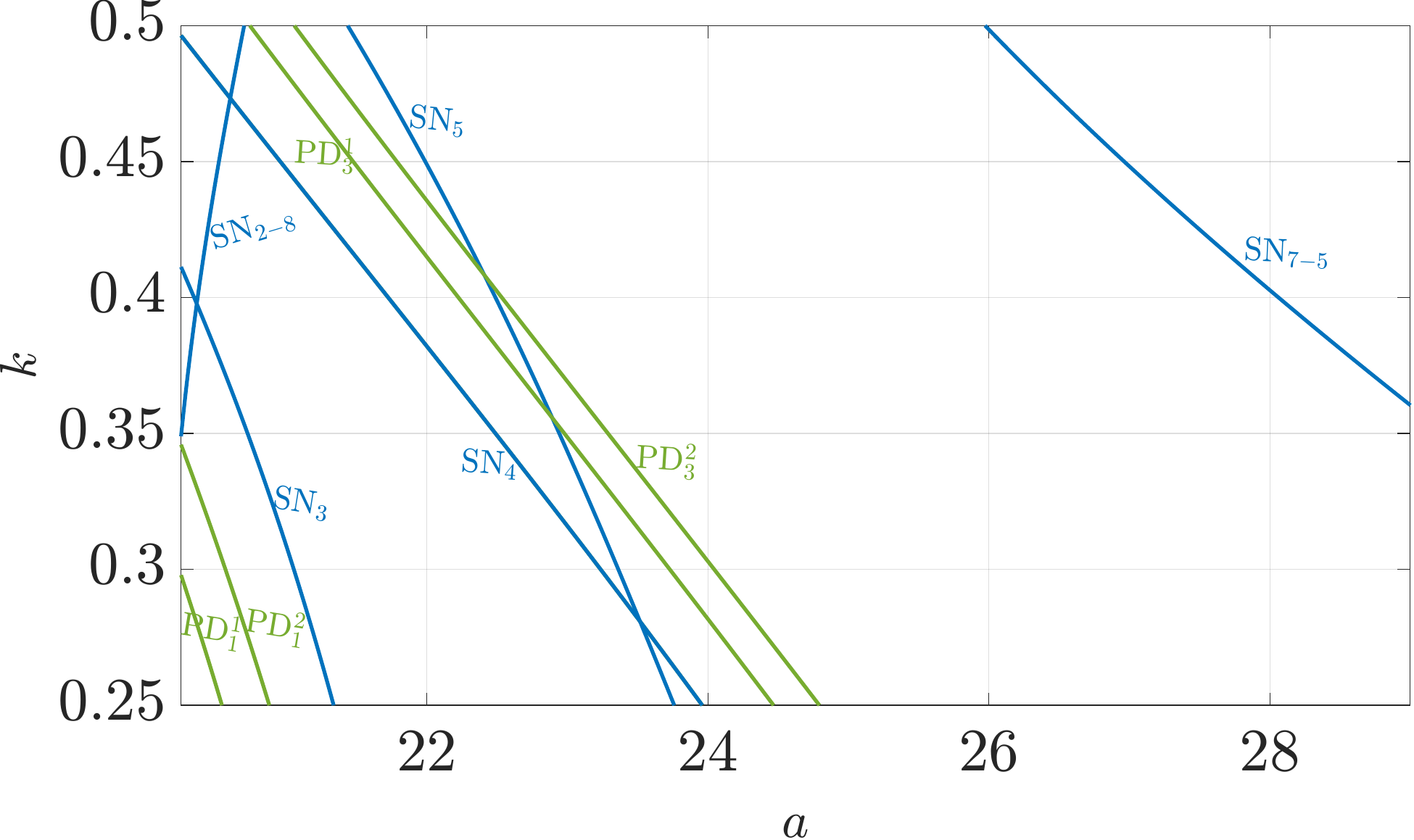}}
  \caption{\small Bifurcation curves in the third region, divided into two subregions shown in (a) and (c). In (b), an enlarged version of (a) is shown for clarity. Saddle-node bifurcations are denoted by SN and period-doubling bifurcations by PD. The index of each curve refers to the attractor family of Table~\ref{tbl1}. Bifurcation curves of the same type share the same color. For clarity, the SN$_{6-j}$ curves are denoted by SN$_j$. The superscript in the PD curves refers to their order in the cascade. Some SN and PD curves in (b) are so close that they are indistinguishable. }
\label{bifs3}
%\vspace*{12pt}
\end{figure}

From the discussion of the two regions above, we have essentially exhausted all possible bifurcations of the system, and all primary attractor families have been identified. The purpose of this subsection is not to illustrate new bifurcations, but to show how the order of emergence, the type, and the coexistence of attractors change depending on how each bifurcation curve is crossed. Moreover, it is shown that the different order of bifurcations leads to different merged attractor families, in a similar manner to the I/IV/V family discussed previously. To this end, in what follows we fix $k=0.3$. The bifurcation curves are shown in Fig. \ref{bifs3}, while Fig. \ref{bifa3} shows the emerging attractors as the parameter $a$ varies.
\begin{figure}[t!]
  \centering
  \subfigure[]{
    \includegraphics[width=0.48\linewidth]{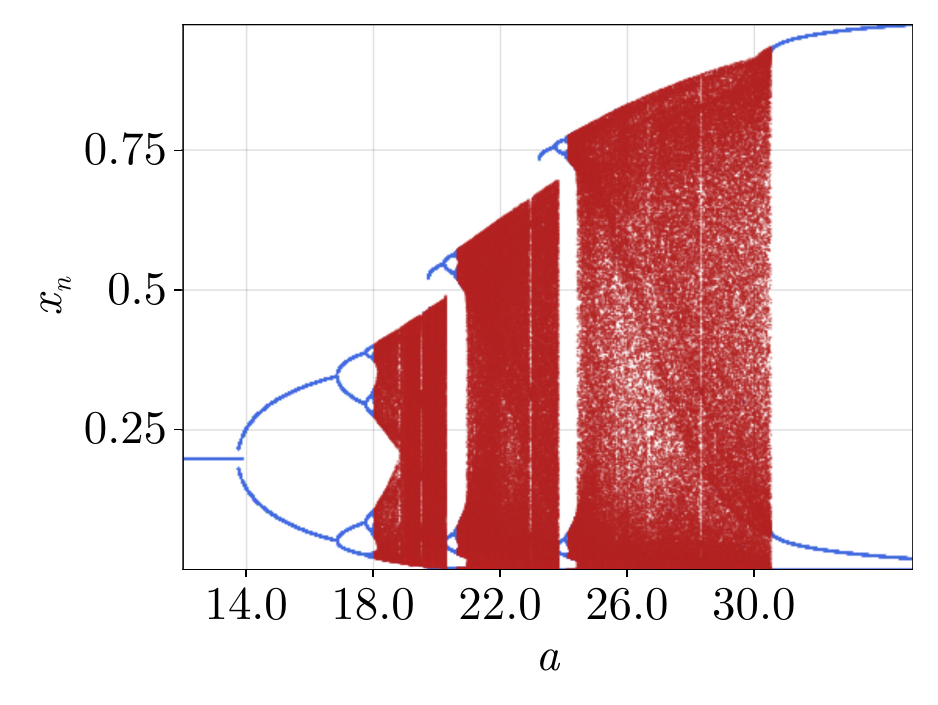}
  }
  \hfill
  \subfigure[]{
    \includegraphics[width=0.48\linewidth]{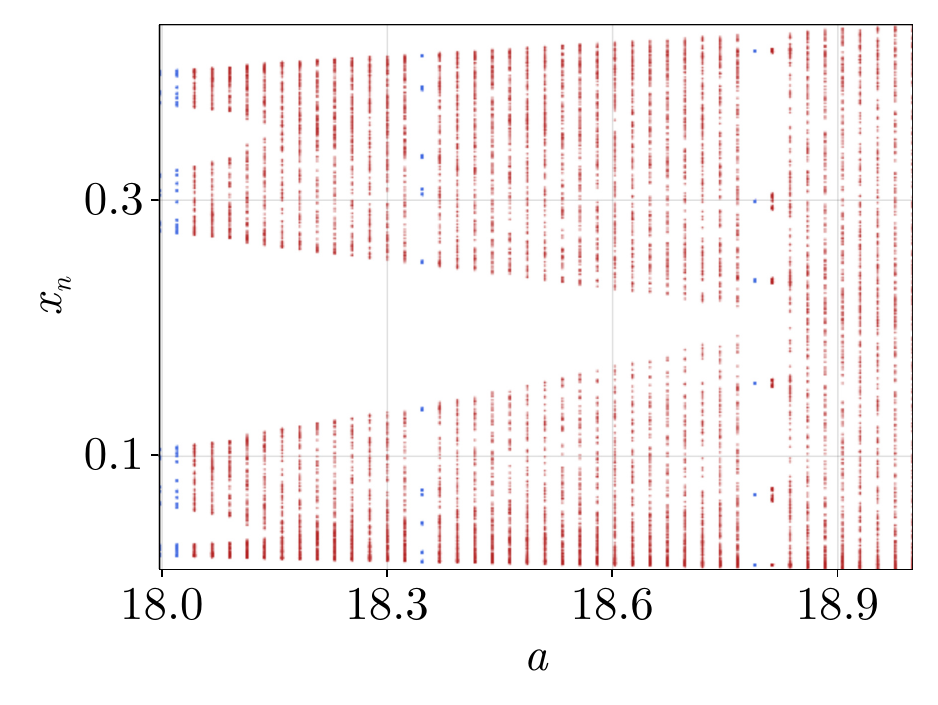}
  }

  \subfigure[]{
    \includegraphics[width=0.48\linewidth]{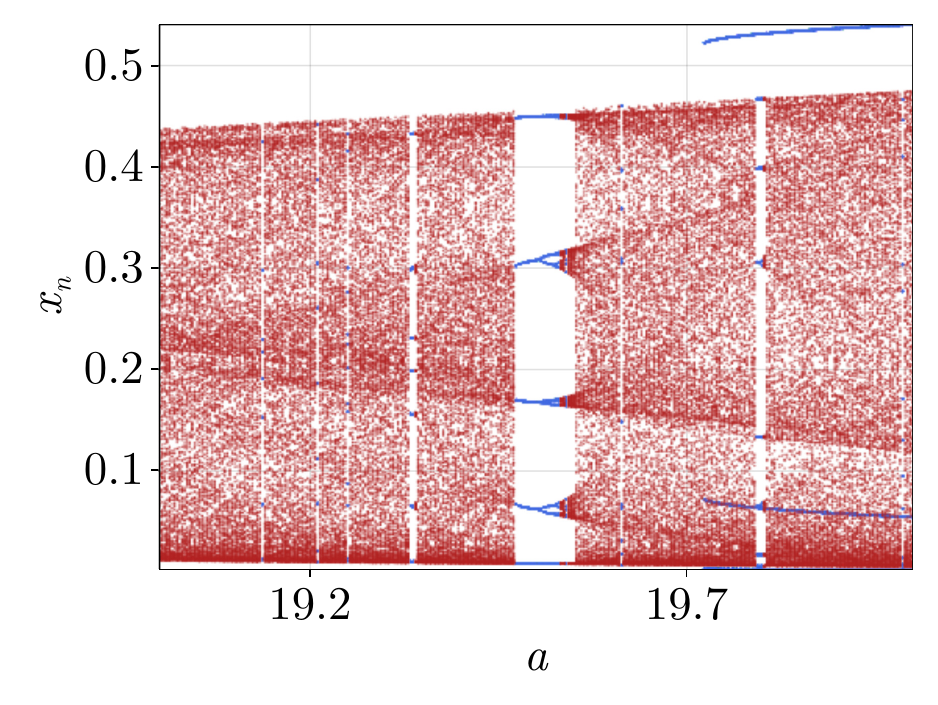}
  }

  \caption{\small (a) The various emerging attractors as $k=0.3$ is held fixed and $a$ is gradually increased. Each attractor's birth or destruction corresponds to the crossing of a bifurcation curve of Fig.~\ref{bifs3} or to a crisis. (b) and (c) correspond to enlarged versions of (a), highlighting the boundary and interior crises as well as the transition between periodic and chaotic behavior. Orbits colored in blue correspond to regular behavior, characterized by a maximum Lyapunov exponent $\lambda<0$. Orbits colored in red correspond to chaotic behavior, with $\lambda\geq 10^{-2}$.}
\label{bifa3}
%\vspace*{12pt}
\end{figure}

Because of the choice of the value of $k$ and based on the above discussion and Figs.~\ref{bifs1} and \ref{bifs2}, we know that the fixed point will not undergo an NS bifurcation; instead, a period-two cycle will emerge upon crossing PD$_2$. As is evident from Figs.~\ref{bifs3} and \ref{bifa3}(a), the cascade of period-doublings of the fixed point, namely of the II family, unfolds before any attractor belonging to a different family appears. In fact, a two-piece chaotic attractor emerges, which after an interior crisis expands into a one-piece attractor. This should be contrasted with the situation discussed in the previous subsections, where the period-doubled orbits were destroyed before transitioning to chaos (cf. the associated Lyapunov exponent), while coexistence with other families was also observed.

As $a$ increases, the SN curve of the period-six orbit of the VI family is crossed. This results in the appearance of a stable period-six orbit, while the chaotic attractor is destroyed through a boundary crisis--cf. the enlarged version in Fig.~\ref{bifa3}(b). This orbit then undergoes a period-doubling route to chaos, after which the chaotic attractor reemerges. As a result, the period-six cycle belongs to the same branch as the chaotic attractor, appearing as a periodic window. This property is in fact true for all the II--$i$ and VI periodic attractor families studied in the previous subsections. Thus, none of these families coexist with each other or with the chaotic attractor; rather, they are all embedded as periodic windows in the chaotic branch of the II family and undergo a period-doubling route to chaos before the chaotic attractor reemerges. The enlarged version of Fig.~\ref{bifa3}(c) shows two examples of periodic windows of the II--$5$ and II--$7$ families. Note that the proximity of the NS and PD bifurcation curves in Fig.~\ref{bifs3}(b) implies that the corresponding periodic windows are very narrow and are difficult to distinguish. We also note that the order of appearance of each of these attractors now differs, as it depends on the order in which the corresponding SN curves of Fig.~\ref{bifs3} are crossed.

However, the I, III, IV, and V families do not follow the above-discussed pattern, as they coexist with the chaotic attractor after their appearance. As shown in Fig.~\ref{bifa3}(a) and Fig.~\ref{bifa3}(c), the period-three attractor of family I, born at SN$_1$ in Fig.~\ref{bifs3}, coexists with the chaotic attractor, which becomes a chaotic saddle after a boundary crisis. The period-three orbit undergoes period-doubling to chaos and merges with the II family chaotic saddle, forming a larger chaotic attractor of the merged family I/II.

This attractor follows a similar pattern, exhibiting narrow periodic windows, coexisting with the four-period IV cycle, and ultimately transforming into a chaotic saddle through a boundary crisis. Correspondingly, the period-four attractor undergoes a period-doubling route to chaos and expands, incorporating the chaotic saddle I/II via an interior crisis, thus forming a larger chaotic attractor of the merged family I/II/IV. We note that the III and V families appear alone only within very narrow windows before merging into the larger attractor; for this reason, we omit them from the notation of the merged attractor. The merged attractor is now qualitatively similar to that of family I/IV/V discussed in the previous sections. The chaotic regime ends with the birth of the VII family period-five attractor, which becomes the sole attractor of the system for large $a$, as expected from the analysis of the previous sections.

\subsection{A note on different choices of fixed parameters}\label{sec3.4}
In this subsection, we briefly elaborate on the initial statement of the present Section regarding the changes in the system dynamics when different fixed parameters are chosen. At this point, we again note that the dynamics of \eqref{sys} depend on $(a,b,k,cp_0,cp_1)$, thereby reducing the dimensionality of the parameter space. The fixed point always undergoes a PD bifurcation similar to the PD$_2$ curve discussed previously. On this curve lies a codimension-two R$_2$ point, which may not lie within the admissible parameter set, i.e., $k\in(0,1)$. In particular, we have the following.
\begin{proposition}\label{prop2}
The coordinates in the $(k,a)$ plane of the codimension-two $\mathrm{R}_2$ point defined by the intersection of the curves $\mathrm{NS}_2^1$ and $\mathrm{PD}_2$ are given by
\begin{equation}
k_{\mathrm{R}_2}=2\left(1-\sqrt{ \frac{cp_1}{1+cp_1}} \right), \; \; a_{\mathrm{R}_2}
=
\frac{
2\left(
1+\sqrt{\dfrac{cp_1}{1+cp_1}}
\right)
}{
\displaystyle
\frac{b-cp_0}{1+cp_1}
\left(
1-\frac{b-cp_0}{1+cp_1}
\right)
}.
\label{R2_coords}
\end{equation}
Thus, the $\mathrm{R}_2$ point lies within the admissible parameter region if and only if $cp_1> 1/3.$
\end{proposition}
\begin{proof}
The Jacobian evaluated at the fixed point \eqref{fp} is 
\[
J=
\begin{pmatrix}
1-ad & -acd\\
kp_1 & 1-k
\end{pmatrix},
\]
where we have set $d=x^\star(1-x^\star)$. Hence $\text{tr}\; J=2-k-ad$ and $\det J=1-k+ad(k\beta-1),$ where $\beta=1+cp_1.$ The NS$_2^1$ curve is determined by the conditions $\det J=1, \; -2\leq\mathrm{tr}\; J\leq 2$ while the PD$_2$ curve is determined by the condition $1+\text{tr} \;J+\det J=0.$ The R$_2$ resonance point is the intersection of these two curves and is, thus, characterized by the conditions $\det J=1$ and $\text{tr} J=-2,$ which yield the relations in \eqref{R2_coords}. From these equations it is clear that $k_{\text{R}_2}\in(0,1)$ if and only if $cp_1>/3.$

\end{proof}

As a result, the $\mathrm{NS}_2^1$ curve crossing this point may happen to lie outside the admissible parameter region. In that case, no stable invariant curve is created, and the bifurcation diagram has a structure similar to Fig.~\ref{bifs3}. The overall dynamics can then be analyzed as in Subsection~\ref{sec3.3}, and Fig.~\ref{bifa3} can be regarded as a generic picture. Otherwise, if the NS curve does lie within the admissible parameter region, the generic picture of Figs.~\ref{bifs1}, \ref{bifs2}, and \ref{bifs3}, as well as the analysis conducted in the corresponding sections, applies. Note, however, that the exact periods of the periodic orbits may differ. Thus, when the extreme event is characterized by different parameters, the analysis can be carried out in a similar manner, and no qualitatively new phenomena are expected; only quantitative differences in the bifurcation curves arise.

As an example, we consider the fixed parameters $b=0.3,\; p_0=p_1=10^{-5},\; \alpha=0.97,\; M=35$. For these parameter choices, the NS$_2^1$ curve of the fixed point does not lie within the admissible parameter region. As a result, we expect behavior similar to that discussed in Subsection~\ref{sec3.3}. Inspecting Fig.~\ref{bifa4}, we indeed observe qualitative similarity with Fig.~\ref{bifa3}(a). Apart from the quantitative differences regarding the critical parameter values at which each attractor appears, note that the periodic attractor born after the boundary crisis of the chaotic attractor, which, similarly to the period-five VII attractor discussed in the previous sections, does not undergo any further bifurcations, has period seven. 

\begin{figure}
    \centering
    \includegraphics[width=0.48\linewidth]{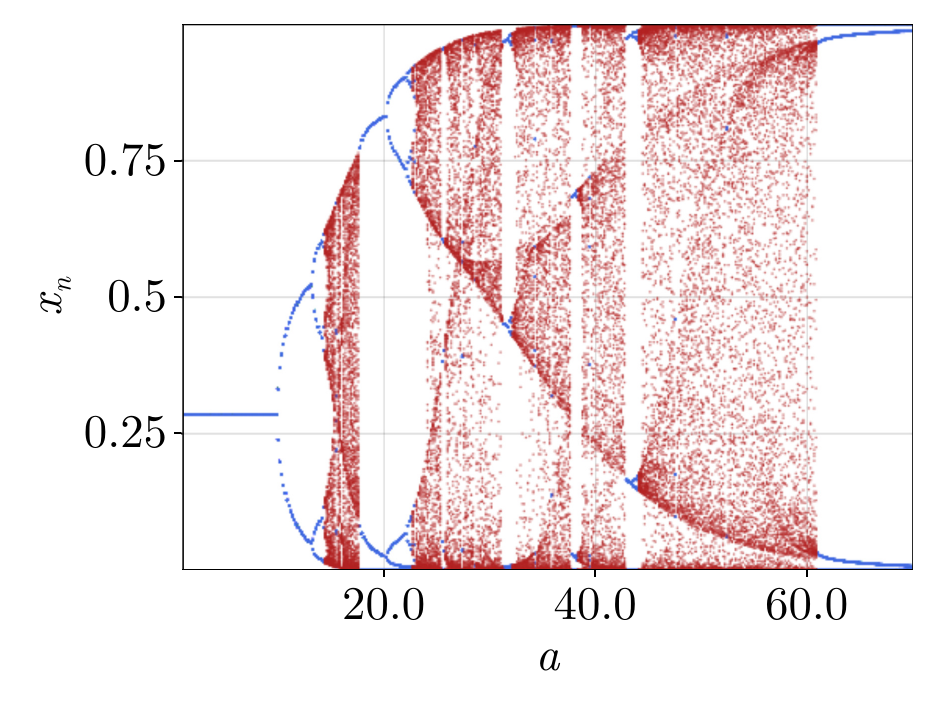}
    \caption{\small Attractors as $a$ varies and $k=0.865$ for the different set of fixed parameters.}
    \label{bifa4}
\end{figure}

\section{Summary and conclusions}\label{section3}
In this paper, we studied the nonlinear dynamics of a risk-sensitive variant of the Multiplicative Weights Update (MWU) for two-strategy population congestion games in the presence of extreme events affecting the first resource. Agents are assumed to learn and react to the perceived risk, captured by a belief variable. This variable is updated through a standard discrete-time adaptive expectation rule that accounts for the risk signal generated by the environment. The population state evolves according to the MWU, incorporating the agents' belief about the risk of the first resource. As a result, the learning dynamics becomes two-dimensional.

Dynamically, the system always admits two trivial boundary equilibria, one of which is always a saddle, while the stability of the other depends on the parameters; when stable, it is the global attractor. The nontrivial regime arises when both boundary equilibria are saddles and an interior equilibrium exists, which is stationary from a game-theoretic perspective. In that case, a detailed two-parameter bifurcation analysis with respect to the parameters controlled by the agents is conducted to identify regions of qualitatively different behavior. Despite the complexity of the dynamics, the attractors of the system can be grouped into specific families. 

The fixed point loses stability via either a supercritical period-doubling (PD) or Neimark--Sacker (NS) bifurcation--the latter being absent in one-dimensional variants of MWU. Unlike the traditional MWU dynamics, in its stable regime the fixed point coexists with different families of attractors. These families originate from periodic orbits, undergo Feigenbaum cascades to chaos, crisis bifurcations and merging, giving rise to transient chaos and narrow periodic windows. The stable invariant curve born at the NS bifurcation of the fixed point leads to the creation of infinitely many Arnold tongues. The phase locking also constitutes a new phenomenon that has not been identified in traditional MWU dynamics, even in the heterogeneous setting.

The dynamics within each tongue is qualitatively similar: PD and NS bifurcations occur, while codimension-two points serve as organizing centers. Similar codimension-two structures are observed for the period-two cycle originating from the fixed point. The overlap of different tongues results in coexistence among multiple periodic attractors as well as the invariant curve, with intermingled basins of attraction. As the learning rate increases, each family is destroyed through boundary crises, except for a period-five cycle which remains the sole attractor of the system. Variations of the fixed parameters result in qualitatively similar dynamics, making the described picture generic. Despite the complexity of the dynamics, the time-averages of the trajectories are proved to converge to the stationary equilibrium, similarly to the traditional MWU algorithm.

We highlight that incorporating the risk associated with the extreme event into the decision-making process leads to the emergence of several phenomena that are absent from the MWU variants studied so far. These include invariant curves appearing as attractors, phase-locking phenomena, and codimension-two resonances. Notably, the stationary interior equilibrium may be stable while coexisting with both chaotic and periodic attractors. Moreover, for large values of the learning rate, the system exhibits predictable periodic behavior, in contrast to traditional MWU variants. Finally, we highlight the interesting analogy that the dynamical phenomena arising in the analyzed system can be viewed in direct analogy with the concept of chaotic itinerancy discussed in the cognitive literature.

The results presented in this work suggest that adaptive responses to extreme events may substantially modify the behavior of agents interacting according to the standard multi--agent learning algorithm MWU in congested systems, even when the underlying congestion game remains simple. In applied settings, such as a transportation network, for instance, the results indicate that the involved agents may be driven toward non-stationary policies, redistributing between the two alternative routes in a predictable or even unpredictable, random-like manner. In fact, it is notable that small variations in the agents' initial beliefs and initial resource allocation may lead to substantially different long-term strategies. A similar interpretation can be directly applied to allocation settings involving computational, communication, or energy resources that are subject to occasional failures. As a result, convergence of the long-run averages should be interpreted with caution, as it may conceal substantially more complex behavior and sensitivity to initial conditions.

This work aims to serve as a first step toward understanding the effect of extreme events, and more generally the coupling of game dynamics with a belief variable. Several directions remain to be explored. We highlight two major avenues. The first concerns extending the risk-sensitive analysis to other classes of games, explicitly accounting for the increase in system dimensionality induced by perceived risk, thereby complementing the algorithmic perspective. For instance, one could study the effect of extreme events as perturbations of zero-sum games under replicator dynamics, where the unperturbed system possesses a Hamiltonian structure. The second involves stochastic models in which extreme events directly perturb the game dynamics, rather than being incorporated through a deterministic risk-sensitive correction, with particular emphasis on the resulting large-deviation estimates.

\section*{Acknowledgement} This work has been supported through a Vannevar Bush Faculty Fellowship administered through the ONR Award N000142512059.

\section*{Competing interests}
The authors have no relevant financial or non-financial interests to disclose.

\section*{Data availability}
The bifurcation data are available at \url{https://doi.org/10.5281/zenodo.22173914}

\appendix

\section{Omitted details from Section 2.3} \label{app}
In this appendix, we provide the technical details omitted from Section \ref{sec2.2} regarding the boundary equilibria $(0,p_0)$ and $(1,p_0+p_1)$, as well as the connection between the evaluation of the risk in \eqref{eq:c1tilde} and the CVaR risk measure. In what follows, it is always assumed that $x_0\in (0,1)$.

\subsection{Boundary Equilibria}\label{sec1_app}
First, we show the following property regarding the boundary equilibria of \eqref{sys}, which was used in the proof of Proposition \ref{prop1} and is also used in the proof of the global stability of $(0,p_0)$ in Lemma \ref{lem1}.
\begin{lemma} \label{lem_delta}
If $x_0\in(0,1)$ and $q_0\in[0,1],$ then there exists a $\delta_1\in(0,1)$ such that $x_n<1-\delta_1$ for all $n\geq 0.$ Furthermore, if $b>cp_0,$ then there is a $\delta\in(0,1)$ such that $x_n\in(\delta,1-\delta)$ for all $n.$
\end{lemma}

\begin{proof}
For this proof it is more convenient to work with the logits $y_n=\log \left(x_n/(1-x_n) \right).$ Then, using \eqref{sys} we have
\[
y_{n+1}=y_n-a(x_n-b+cq_n).
\]
Let $\varepsilon_1=(1-b)/2.$ If \(x_n\geq 1-\varepsilon_1\), then $x_n-b+cq_n>0$ and, thus, $y_{n+1}<y_n$, implying $x_{n+1}<x_n.$ By continuity of the \(x\)-component of the map on the compact set \([0,1-\varepsilon_1]\times[0,1]\), there exists \(\mu<1\) such that  $x_n\leq \mu<1$ whenever $(x_{n-1},q_{n-1})\in[0,1-\varepsilon_1]\times[0,1]$. As a result, there is a $1>\delta_1>0$ such that $x_n<1-\delta_1$ for all $n.$

Now suppose that $b>cp_0.$ Let $\bar q\in(p_0,b/c)$ and \(\varepsilon_2>0\) small enough so that 
\[
p_0+p_1\varepsilon_2<\bar q, \; \; \varepsilon_2 - b + c\bar q<0.
\]
Let $\theta=p_0+p_1\varepsilon_2<\bar q$ and $\beta=b-c\bar q-\varepsilon_2>0$. If \(x_n\leq\varepsilon_2\) and \(q_n\leq\bar q\), then
\[
x_n-b+cq_n\leq \varepsilon_2 - b+c\bar q=-\beta,
\]
and hence $y_{n+1}\geq y_n+a\beta.$ Thus, once \(x_n\) is sufficiently small and \(q_n\leq\bar q\), the \(x\)-coordinate is pushed away from zero. It remains to control the indices for which \(x_n\leq\varepsilon_2\) but \(q_n\) may still be larger than \(\bar q\). Consider $N$ indices such that \(x_s,x_{s+1},\ldots,x_{s+N-1}\leq\varepsilon_2\). Then 
\[
q_{s+N}
\leq
(1-k)^N q_s+\bigl(1-(1-k)^N\bigr)\theta
\leq
(1-k)^N+\bigl(1-(1-k)^N\bigr)\theta.
\]
But since \(\theta<\bar q\), there is an $L\geq 1$ such that 
\[
(1-k)^L+\bigl(1-(1-k)^L\bigr)\theta\leq \bar q.
\]
As a result, if $L\leq N-1$, then \(q_{s+L}\leq\bar q\). Moreover, $q_{s+L+1}< \bar q$ and, thus, $q_{s+j}\leq \bar q$ for every $N-1\geq j\geq L,$ since $x_{s+j}\leq \varepsilon_2$ continues to hold. That is, after the first $L$ steps, $y_n,$ and hence $x_n$, begins to increase, as long as $x_n\leq \varepsilon_2$, suggesting that there cannot be infinitely many consecutive indices with $x_n<\varepsilon_2$.

Now let $A=a(1-b+c)$ and $y_{\varepsilon}=\log \left( \varepsilon_2/( 1-\varepsilon_2) \right)$. Then, $y_{n+1}\geq y_n-A$ for all $n$. Consider a maximal block of consecutive indices starting at $s>0$ and of length $N$ during which \(x_n\leq\varepsilon_2\). First, we have $x_{s-1}>\varepsilon_2$, implying $y_s\geq y_{\varepsilon}-A$. Based on the previous discussion, the logits $y_n$ can decrease by at most $LA,$ as after the first $L$ steps, $q_n \leq \bar q$ and, so, $y_n$ increases. Therefore, within any such block we have
\[
y_{n} \geq y_\varepsilon-(L+1)A,\; n=s,s+1,\dots,s+N-1.
\]
If there is a block with $s=0,$ then one has $y_n\geq y_0-LA.$ Outside the region \(x\leq\varepsilon_2\), we have \(y>y_{\varepsilon}\) and, thus, there exists a constant $Y$ such that $y_n\geq Y$ for all $n.$ Equivalently, there exists $\delta_2 \in(0,1)$ such that $x_n > \delta_2$ for all $n$. Combining this result with the upper bound, we obtain the claim.

\end{proof}

Next, we show that if $b\leq cp_0,$ then the dynamics induced by \eqref{sys} is trivial.
\begin{lemma}\label{lem1}
If $b\leq cp_0,$ then the equilibrium $(0,p_0)$ is globally attracting.
\end{lemma}

\begin{proof}
First suppose that $b<cp_0$. Let $q'\in(b/c,p_0).$ Inductively, we have that $q_n-p_0\geq (1-k)^n(q_0-p_0).$ Thus, there exists $n_0$ such that for all $n\geq n_0,$ we have $q_n\geq q'.$ Then, for $n\geq n_0$ we have that
\[
x_{n+1}\leq\frac{x_{n}}{x_{n}+(1-x_{n}) \exp{\left(a(x_n-b+q'c) \right)}}=h(x_n)<x_n
\]
Let $1>C>e^{-a(q'c-b)}$. There exists an $\eta>0$ such that $h(x)\leq Cx$ for $x< \eta.$ From Lemma \ref{lem_delta}, there exists a $\delta>0$ such that $x_n\leq 1-\delta$ for all $n$. For $x\in[\eta,1-\delta],$ since $h(x)<x,$ there exists $\ell_\eta>0$ such that $h(x)\leq x- \ell_\eta$. If for some $n_1\geq n_0$ we have $x_{n_1} \in [\eta, 1-\delta],$ then $x_{n_1+1}\leq h(x_{n_1})\leq x_{n_1} -\ell_\eta$. As a result, there is an $n_2\geq n_0$ such that $x_n \in [0,\eta]$ for all $n\geq n_2.$ This implies that $x_{n+1}\leq Cx_n$ for all $n\geq n_2,$ yielding $\lim x_n =0.$ The convergence of $q_n$ follows.

Now suppose that $b=cp_0.$ Then the interior equilibrium coincides with the boundary one $(0,p_0).$ Let $y_n=q_n-p_0$ and as in the proof of Lemma \ref{lem_delta}, introduce the logits $z_n=\log \left(x_n/(1-x_n) \right).$ Then we have that
\[
z_{n+1}-z_n=-a(x_n+cy_n), \qquad  y_{n+1}-y_n=k(p_1x_n-y_n).
\]
Thus,
\[
z_n=z_0-a(1+cp_1)\sum_{j=0}^{n-1}x_j +\frac{ac}{k}(y_n-y_0).
\]
Now, if $\sum_{j=0}^{\infty}x_j <\infty$, then $x_n\to 0$, which, from the recursion for $y_n$, also implies that $y_n\to0$. To see the latter, note that 
\[
y_n=(1-k)^n y_0+kp_1\sum_{j=0}^{n-1}(1-k)^{n-1-j}x_j.
\]
The convergence $y_n \to 0$ combined with $\sum_{j=0}^{\infty}x_j <\infty$ results in $z_n$ converging to a finite number, a contradiction with the definition of $z_n$ and the fact that $x_n\to 0.$ Thus, $\sum_{j=0}^{\infty}x_j =\infty$, which, since $y_n$ is bounded, yields $z_n \to -\infty$ and thus $x_n\to 0$. This also implies that $q_n \to p_0$, using similar arguments as previously for $y_n\to 0,$ proving the claim.

\end{proof}

\subsection{Connection to the CVaR risk measure}\label{sec2_app}
Suppose that the agents evaluate the believed probability of the extreme-event component of action \(1\) through the Conditional Value-at-Risk measure \(\mathrm{CVaR}_\alpha\), at confidence level \(\alpha\in(0,1)\). Then, the perceived cost of action \(1\), instead of \eqref{eq:c1tilde}, is now given by
\[\tilde c_1(x_n,q_n)
=
N\left(\gamma x_n+M\min\left(1,\frac{q_n}{1-\alpha}\right)\right).
\]
Following the same steps as in the derivation of \eqref{sys}, we arrive at the following system
\begin{align*}
x_{n+1}&=\frac{x_n}{x_n+(1-x_n)\exp\!\left(
a\left(x_n-b+M\min\!\left(1,\frac{q_n}{1-\alpha}\right)\right)
\right)} \nonumber\\
q_{n+1}&=(1-k)q_n+k(p_0+p_1x_n),
\end{align*}
which governs the dynamics of the agents. We denote the above map by $\tilde f$, i.e., $(x,q)\mapsto\tilde f(x,q)$. In what follows, $x_n$ and $q_n$ refer to the $n$-th iterate of the system defined by $\tilde f$. The induced dynamics can be interpreted more clearly by noting that
\[
\mathrm{CVaR}_\alpha(C_1)=N\left(\gamma x_n+M\min\left(1,\frac{Q_n}{1-\alpha}\right)\right),
\]
where $C_1$ is given by \eqref{eq0} and we have denoted by $Q_n$ the probability of occurrence of the extreme event $U_n$, i.e., $U_n \sim \mathrm{Bernoulli}(Q_n)$. As a result, $\tilde f$ corresponds exactly to the dynamics induced by the CVaR-based perceived cost, with the perceived probability $q_n$ replacing the true--yet not known by the agents--$Q_n$, whereas $f$ corresponds to its unclipped version.

We generally assume that $p_0, \; p_1\leq 1-\alpha$, which is natural since they are associated with the extreme-event probabilities. The purpose of this Section is to discuss the connections between the dynamics induced by the maps $f$ and $\tilde f$ in a semi-analytical manner.

The discussion in Section \ref{sec2.2} on the equilibria of the system, from both dynamical and game-theoretical perspectives, can be repeated verbatim for the system defined by $\tilde f$. In particular, the content of Lemmas \ref{lem_delta} and \ref{lem1} also holds in this case. For this reason, in what follows, it is enough to assume that $b>cp_0.$ However, the result of Proposition \ref{prop1} does not directly transfer. In general, one only has that
\[
\lim_{n \to \infty} \frac{1}{n+1} \sum_{i=0}^n \left(x_i+M\min \left(1, \frac{q_i}{1-\alpha} \right) \right)=b,
\]
which is proved by following the same steps as in Proposition \ref{prop1}. 

Apart from the similar behavior of the equilibria, the two systems also share deeper connections in terms of their respective dynamics. To this end, we first show the following.
\begin{lemma}\label{lem2}
We cannot have $q_n>p_0+p_1$ for all $n.$
\end{lemma}
\begin{proof}
Assume the opposite. From Lemma \ref{lem_delta}, there exists a $\delta>0$ such that $x_n\leq 1-\delta$ for all $n.$ Now from the dynamics of $q$ in \eqref{sys}, we have that
\[
q_n-(p_0+p_1)=(1-k)^n (q_0-(p_0+p_1))-kp_1 \sum_{j=0}^{n-1} (1-k)^{n-1-j} (1-x_j)
\]
Thus,
\[
(1-k)^n(q_0-p_0-p_1)> kp_1 \sum_{j=0}^{n-1}(1-k)^{n-1-j}(1-x_j)\geq\delta p_1 (1-(1-k)^n).
\]
Letting $n\to \infty$ yields again the desired contradiction.   

\end{proof}

Using Lemma \ref{lem2} and induction, we obtain that $q_n\leq p_0+p_1$ for all $n\geq n_0$. As a result, if $p_0+p_1\leq 1-\alpha$, then the dynamical phenomena present in the systems defined by $f$ and $\tilde f$ coincide. Requiring the condition $p_0+p_1\leq 1-\alpha$ to hold is again natural, recalling the definitions of the parameters $p_0$, $p_1$, and $\alpha$. If, however, this condition does not hold, we can still show that for small to moderate values of $a$, the long-term dynamics of the two systems coincide. To this end, we need first the following lemmas.
\begin{lemma}\label{lem3}
Suppose $1-\alpha< p_0+p_1$ and let $x_c=(1-\alpha-p_0)/p_1.$ Then $x_n>x_c$ cannot hold for all $n.$
\end{lemma}

\begin{proof}
Assume the opposite. Let $m(q)=\min{\left(1,q/(1-\alpha) \right)}.$ Then, by the dynamics of $q$ we have that
\[
q_n\geq (1-\alpha)-(1-k)^np_0-(1-k)^np_1x_c.
\]
As a result, $\lim m(q_n)=1.$ From Lemma \ref{lem_delta}, there exists a $\delta>0$ such that $\delta\leq x_n\leq 1-\delta$ for all $n.$ By the uniform continuity of 
\[
g(x,m)=\frac{x}{x+(1-x)\exp{\left (a(x-b+Mm) \right)}}
\]
on $[\delta,1-\delta]\times [0,1]$ we get that $|e_n| \to 0$, where we have set $e_n=f(x_n,q_n)-f(x_n,1-\alpha)$. 

Now we write $h(x)=g(x,1)$ and $x_{n+1}=h(x_n)+e_n.$ Let $1>C>e^{-a(M-b)}$. Similar to the proof of Lemma \ref{lem1}, there exists an $\eta \in (0,x_c)$ such that $h(x)\leq Cx$ for $x< \eta$ and an $\ell_\eta>0$ such that $h(x)\leq x- \ell_\eta$ for $x\in[\eta,1-\delta],$ where $\delta$ is such that $\delta\leq x_n\leq 1-\delta$ for all $n.$ Let $n_0$ be such that $|e_n|< \ell_\eta/2$ for $n \geq n_0.$ Thus, if for some $n\geq n_0$ we have $x_n\in [\eta, 1-\delta],$ then $x_{n+1}\leq x_n-\ell_\eta/2.$ Therefore, there exists an $n_2\geq n_0$ such that $x_{n_2} \in [0,\eta],$ yielding a contradiction.

\end{proof}

\begin{lemma}\label{lem4}
Suppose $1-\alpha< p_0+p_1$ and let $x_c=(1-\alpha-p_0)/p_1.$ If $x_{n_0}\leq x_c,$ then $q_n\geq1-\alpha$ cannot hold for all $n.$
\end{lemma}

\begin{proof}
Assume the opposite. First note that
\[
x_{n_0+1} \leq \sup_{0\leq x \leq x_c} \frac{x}{x+(1-x) \exp{\left(  a(x-b+M)\right)}}=\tilde x<x_c
\]
and so $x_n<x_c$ for all $n> n_0.$ Thus, for $n> n_0$ we have
\[
q_{n+1}=(1-k)q_n+k(p_0+p_1x_n)< (1-k)q_n +k(p_0+p_1\tilde x).
\]
As a result, $\limsup q_n\leq p_0+p_1\tilde x<  1-\alpha,$ contradiction.

\end{proof}

The next lemma uses Lemmas \ref{lem3} and \ref{lem4} to argue that for $a$ up to a certain threshold, the maps $f$ and $\tilde f$ induce the same attractors.

\begin{lemma}\label{lem5}
There exists an $a_0$ such that for all $a \in (0,a_0),$ $x_0 \in (0,1)$ and $q_0>0,$ there exists an $n_0$ such that $q_n\leq1-\alpha$ for all $n\geq n_0$.

\end{lemma}
\begin{proof}
Let $x_c=(1-\alpha-p_0)/p_1$ and $q_-=(1-\alpha)(b-x_c)/M.$ Then, under the assumptions on the parameters, i.e., that $p_0, \; p_1 \leq 1-\alpha$ and $1-\alpha<p_0+p_1$ we have that $q_-<(1-\alpha)(1-x_c)\leq p_0.$ To see the first inequality, since $1-\alpha<p_0+p_1$, we have $x_c\in[0,1)$. Since $b<1$ and $M>1$, the first inequality follows. To see the second inequality,
$$
(1-\alpha)(1-x_c)-p_0=\frac{(1-\alpha-p_0)(p_1-(1-\alpha))}{p_1}\leq0.
$$ 
Define 
\[
a_0= \max \left \{a\geq 0: \max_{0\leq x\leq x_c} \frac{x}{1-x} e^{-a(x-x_c)}\leq \frac{x_c}{1-x_c} \right\},
\]
and let $z_n=x_n/(1-x_n)$. Then, 
\[
z_{n+1}\leq z_n e^{a(x_c-x_n)}.
\]
For $a< a_0,$ if $x_n\leq x_c$ and $q_n\in[q_-,1-\alpha],$ then $q_-<q_{n+1}\leq 1-\alpha$ and 
\[
z_{n+1}\leq \frac{x_c}{1-x_c},
\]
from which $x_{n+1}\leq x_c$. As a result, for $a\leq a_0$ the rectangle $[0,x_c] \times [q_-,1-\alpha]$ is forward-invariant.

Now notice that since
\[
q_{n+1}\geq (1-k)q_n+kp_0,
\]
there exists $N_0$ such that $q_n\geq q_-$ for all $n\geq N_0$. By Lemma \ref{lem3}, if at some $N_1\geq N_0$ we have $x_{N_1}\geq x_c,$ then there exists $N\geq N_1$ such that $x_{N}\leq x_c.$ If $q_N\leq 1-\alpha,$ then $x_n \leq x_c$ for all $n\geq N$ by the forward-invariance property we showed previously. If $q_N > 1-\alpha,$ then $x_{N+1}<x_N\leq x_c$ since $M>1>b.$ In any case we conclude that $x_n \leq x_c$ for all $n\geq N.$

Using Lemma \ref{lem4} and forward-invariance, we get that there exists an $n_0>N$ such that $q_{n}<1-\alpha$ for all $n\geq n_0.$ By the property of forward invariance the result follows.

\end{proof}

The previous results suggest that, under the natural assumption $1-\alpha\geq p_0+p_1$, or, when this condition does not hold, but provided that $a\leq a_0$, the sequence $q_n$ eventually becomes and remains smaller than $1-\alpha$. Thus, the tail dynamics of the two systems, and consequently their attractors, are identical. For larger values of $a$, one can numerically verify that the attractors of $f$ and $\tilde f$ no longer coincide, but remain qualitatively similar. As an example, consider the parameters $b=0.9,\; p_0=10^{-5},\; p_1=3\times10^{-5},\; k=0.75,\; M=1.1,\; 1-\alpha=3.5 \times 10^{-5},\; a=60.2$. For these values, there exist infinitely many indices $n$ for which $q_n>1-\alpha$. Both the system defined by $f$ and that defined by $\tilde f$ have a unique attractor, which is chaotic, shown in Fig.~\ref{large_alpha}(a) and (b), respectively. The same behavior can be observed for other parameter choices.

\begin{figure}[h!]
  \centering
  \subfigure[]{
    \includegraphics[width=0.31\linewidth]{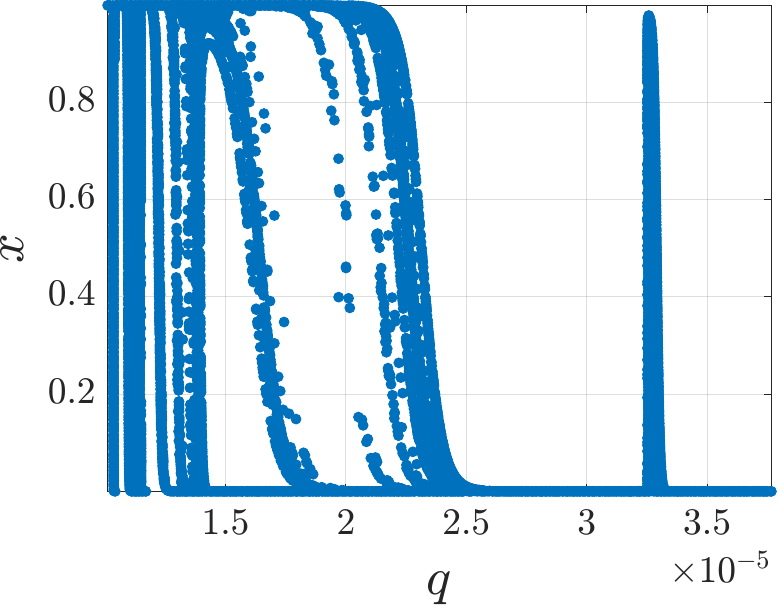}
  }
  \subfigure[]{
    \includegraphics[width=0.31\linewidth]{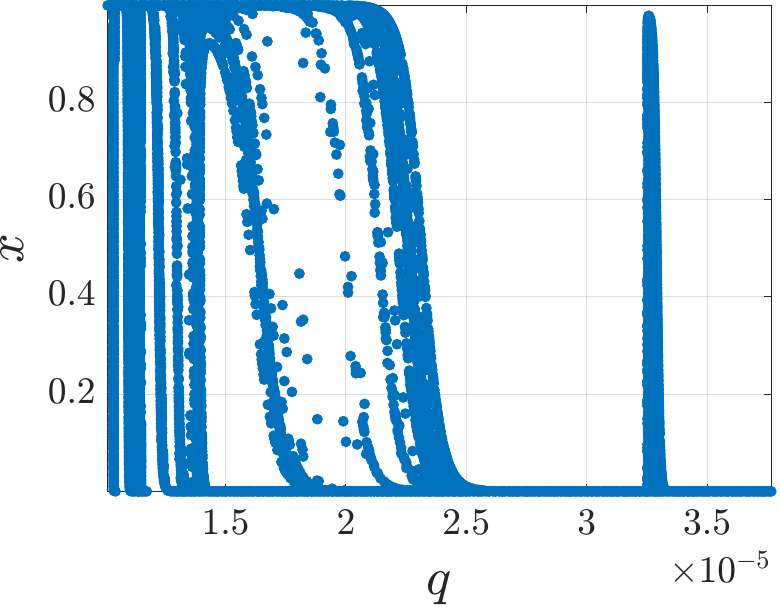}
  }
  \caption{\small (a) Attractor of the system defined by the map $ f$. (b) Attractor of the system defined by the map $\tilde f.$}
\label{large_alpha}
%\vspace*{12pt}
\end{figure}

\section{Computation details} \label{app2}
In this appendix, we provide details for reproducing the results presented in Section \ref{section2} of the paper. We have also deposited the bifurcation data and made them publicly available on Zenodo at \url{https://doi.org/10.5281/zenodo.22173914}.

The bifurcation diagrams presented in Figs. \ref{bifs1}, \ref{bifs2}, and \ref{bifs3} have been computed using the numerical continuation tool \cite{matcont}. In particular, the bifurcations undergone by each attractor have been continued in the two parameters $(a,k)$ to produce the corresponding curves. The tolerances in the computations have been set equal to $10^{-12}$. Further details on how the numerical continuation is carried out can be found in \cite{matcont}. For each parameter combination and each of the coexisting attractors, initial conditions belonging to the corresponding basins of attraction, as shown in Figs. \ref{basins1}, \ref{basins2}, \ref{basins2.5}, \ref{basins3}, and \ref{basins4}, have been used to generate the one-parameter bifurcation diagrams shown in Figs. \ref{bifa1}, \ref{bifa2}, \ref{bifa3}, and \ref{bifa4}, as well as the phase-space plots in Figs. \ref{phspace1}, \ref{5phspace}, \ref{basins2.5}, \ref{coex_16_25}, and \ref{basins4}. It is sufficient to iterate the map \eqref{sys} for $10^6$ steps to discard transient effects. The maximum Lyapunov exponents have been computed using the algorithm in \cite{datseris}. The transient times were set to $10^6$, and the total times to $2\times 10^6$.

For the dependence of the mean lifetime near the boundary crisis presented in Fig. \ref{crisis_plots}(b), initial conditions were sampled from a long trajectory on the chaotic attractor immediately before the crisis, after discarding $3\times10^6$ iterations. For each post-crisis value of $a$, the transient lifetime $\tau$ was defined as the time required for the trajectory to enter and remain for $30$ consecutive iterations within a small neighborhood of the surviving attractor, and the mean lifetime $\langle\tau\rangle$ was computed over $2\times10^3$ initial conditions. A very small interval around $a_c$ was used, and the resulting characteristic time was fitted to the crisis scaling law $\langle\tau\rangle=C(a-a_c)^{-\gamma}$ on logarithmic axes. The robustness of the fitted exponent was checked with respect to the fitting interval and the numerical definition of escape, and the resulting coefficient of determination was $R^2=0.998$.

For the interior crisis shown in Fig. \ref{crisis_plots}(a), the pre-crisis support of the chaotic band was determined from $3\times10^6$ iterates, and a burst was identified whenever the post-crisis trajectory entered a region outside this support, enlarged by a small buffer of $3\times10^{-4}$. The characteristic time was then defined as the mean laminar interval between successive bursts. As in the boundary-crisis case, a very small interval around $a_c$ was used, and the characteristic time was fitted to the corresponding power law. The robustness of the fit was checked with respect to the fitting interval and the numerical definition of a burst, and the resulting coefficient of determination was $R^2=0.979$.

The basin and basin-boundary entropies $S_b$ and $S_{bb}$ in Table \ref{tblentropy} have been calculated following the methods in \cite{entropy1,entropy2,entropy3}. More specifically, a two-dimensional grid $[0,1]^2$, consisting of $1000$ points in each dimension in the variables $(x,q/(5\times 10^{-4}))$, has been used. Finer grids yielded negligible changes in the computed quantities. Disks with a grid radius of $\varepsilon_g=2.5$ were found to yield good results. The discarded transient is again equal to $10^6$. The fraction $p_j$ of points inside the disk belonging to each attractor $j$ was determined. The local entropy was then evaluated as $S_i=-\sum_j p_j\log p_j$. The basin entropy $S_b$ was obtained by averaging $S_i$ over all disks, whereas the boundary basin entropy $S_{bb}$ was obtained by averaging only over disks containing points from at least two different basins. Statistical uncertainty in $S_{bb}$ was estimated by bootstrap resampling of the boundary-disk entropies using $10^3$ bootstrap realizations and the statistical test proposed in \cite{entropy2} was then applied to verify the reliability of the fractality criterion. All computed values of $S_{bb}$ classified as fractal exceeded the corresponding threshold by at least four standard errors.

\bibliographystyle{unsrt}
\bibliography{bibliography}

\end{document}